\documentclass[11pt,a4paper,twoside]{article}
\usepackage{a4wide, amssymb, amsmath, amsthm, graphics, comment, xspace, enumerate}
\usepackage{graphicx,multirow}
\usepackage[a4paper,colorlinks=true,citecolor=blue,urlcolor=blue,linkcolor=blue
                   , bookmarksopen=true,unicode=true,pdffitwindow=true]{hyperref}
\usepackage[section]{algorithm}
\usepackage[noend]{algpseudocode}
\usepackage{caption}
\usepackage{tabularx}
\usepackage{enumitem, multirow}
\usepackage[usenames,dvipsnames]{color}
\usepackage{alltt}
\usepackage{color}
\usepackage{listings}
\usepackage{cite}

\usepackage{array, longtable}

\definecolor{string}{rgb}{0.7,0.0,0.0}
\definecolor{comment}{rgb}{0.13,0.54,0.13}
\definecolor{keyword}{rgb}{0.0,0.0,1.0}
\usepackage{amsmath,amsfonts,amssymb,subfigure}
\usepackage{tikz}
\usetikzlibrary{arrows}
\usetikzlibrary{decorations}
\usetikzlibrary{patterns}

\subfiguretopcaptrue
\tikzstyle{vtx}=[circle, inner sep= 0pt, minimum size= 1.2mm, fill]

\hypersetup{pdftitle={Efficient computation of trees with minimal atom-bond connectivity index}}

\newtheorem{te}{Theorem}[section]
\newtheorem{pro}[te]{Proposition}
\newtheorem{de}{Definition}[section]

\newtheorem{co}[te]{Corollary}
\newtheorem{lemma}[te]{Lemma}

\newtheorem{conjecture}{Conjecture}[section]

\newcommand{\beq}{\begin{eqnarray}}
\newcommand{\eeq}{\end{eqnarray}}

\newcommand{\beqs}{\begin{eqnarray*}}
\newcommand{\eeqs}{\end{eqnarray*}}

\newcommand{\ABC}{{\rm ABC}}

\newcommand{\ds}{\displaystyle}
\allowdisplaybreaks

\begin{document}
\date{December 22, 2014}
\title{ {On structural properties of trees \\with minimal atom-bond connectivity index  II}}
\maketitle
%\vspace{10mm}
%\vspace{1mm}
\begin{center}
{\large \bf  Darko Dimitrov}
\end{center}
\baselineskip=0.20in
\begin{center}
{\it Hochschule f\"ur Technik und Wirtschaft Berlin,
\\ Wilhelminenhofstra{\ss}e 75A, D--12459 Berlin, Germany} 
\\E-mail: {\tt darko.dimitrov11@gmail.com}
\end{center}

\baselineskip=0.20in
\vspace{6mm}
\begin{abstract}
The {\em atom-bond connectivity  (ABC) index}  is a degree-based graph topological index
that found chemical applications.
The problem of complete characterization of trees with minimal $ABC$ index
is still an open problem.
In~\cite{d-sptmabci-2014}, it was shown that trees with minimal ABC index do not
contain so-called {\em $B_k$-branches}, with $k \geq 5$,
and that they do not have more than four $B_4$-branches.
Our main results here reveal that the number of $B_1$ and $B_2$-branches
are also bounded from above by small fixed constants.
Namely, we show that trees with minimal ABC index 
do not contain more than four $B_1$-branches and more than eleven $B_2$-branches.
\end{abstract}
%
%{\small \hspace{0.25cm} \textbf{Keywords:} First, Secomd, Third Zagreb indices, Uperbound
%
% ------------------------------------------------------------------------------------
%-------------------    Section one:  INTRODUCTION and related results      ----------
% ------------------------------------------------------------------------------------
%
\medskip
%
%
% -------------------------------------------------------------------------
% ----------------------           bibliography       ---------------------
% -------------------------------------------------------------------------
%
%------------------------------------------------------------------------------
%
\section[Introduction]{Introduction}
% ----------------------------------------------------------------------------
 Let $G=(V, E)$ be a simple undirected graph of order $n=|V|$ and size $m=|E|$.
For $v \in V(G)$, the degree of $v$, denoted by $d(v)$, is the number of edges incident
to $v$.
For an edge $uv$ in $G$, let
\beq \label{eqn:000}
f(d(u), d(v))=\sqrt{\frac{d(u) +d(v)-2}{d(u)d(v)}}.
\eeq
Then, the {\em atom-bond connectivity (ABC) index} of $G$ is defined as
%\beq \label{eqn:001}
%\ABC(G)=\sum_{uv\in E(G)}\sqrt{\frac{(d(u) +d(v)-2)}{d(u)d(v)}}, {\text{or}}
%\eeq
\beq \label{eqn:001}
\ABC(G)=\sum_{uv\in E(G)}f(d(u), d(v)), \nonumber
\eeq
 The ABC index was introduced  in 1998 by Estrada, Torres, Rodr{\' i}guez and Gutman \cite{etrg-abc-98}, who
showed that it can be a valuable predictive tool in the study of the heat of formation in alikeness.
Ten years later Estrada~\cite{e-abceba-08} elaborated a novel quantum-theory-like justification for this topological index.
After that revelation, the interest of ABC-index has grown rapidly.
Additionaly, the physico-chemical applicability of the ABC index was confirmed and extended in several studies
\cite{as-abciic-10,cll-abcbsp-13, dt-cbfgaiabci-10, gg-nwabci-10, gtrm-abcica-12, k-abcibsfc-12, yxc-abcbsp-11}.

As a new and well motivated graph invariant, the ABC index has attracted a lot of interest in the last several years both in 
mathematical and chemical research communities and numerous results and structural properties of ABC index  
were established~\cite{cg-eabcig-11, cg-abccbg-12, clg-subabcig-12, d-abcig-10, d-sptmabci-2014, dgf-abci-11, dgf-abci-12, ftvzag-siabcigo-2011,
fgv-abcit-09, ghl-srabcig-11, gly-abctgds-12, gf-tsabci-12, gfi-ntmabci-12, gzx-rabchi-2014, llgw-pcgctmabci-13, p-rubabci-14, vh-mabcict-2012, xz-etfdsabci-2012, 
xzd-abcicg-2011, xzd-frabcit-2010}.

The fact that adding an edge in a graph strictly increases its ABC index~\cite{dgf-abci-11} 
(or equivalently that deleting an edge in a graph strictly decreases its ABC index~\cite{cg-eabcig-11})  
has  the following two immediate consequences.

\begin{co}
Among all connected  graphs with $n$ vertices, the complete graph $K_n$ has maximal value of ABC index.
\end{co}

\begin{co}
Among all connected  graphs with $n$ vertices, the graph with minimal ABC index is a tree.
\end{co}

Although it is fairly easy to show that the star graph $S_n$
is a tree with maximal ABC index~\cite{fgv-abcit-09}, despite many attempts in the last years, it is still an open problem
the characterization of trees with minimal ABC index (also refereed as  minimal-ABC trees). 
The aim of this research is to make a step forward towards the full characterizations of minimal-ABC trees.

%Before we continue with the results, we present some additional notation that will be used in the rest of the paper.
In the sequel, we present an additional notation that will be used in the rest of the paper.
A tree is called a {\em rooted tree} if one vertex has been designated the {\em root}.
In a rooted tree, the  {\em parent} of a vertex is the vertex connected to it on the path to the root; every vertex except the root has a unique parent.
A vertex is a parent of a subtree, if the subtree is attached to the vertex.
A  {\em child} of a vertex $v$ is a vertex of which $v$ is the parent.
A vertex of degree one is a {\it pendant vertex}.

For the next two definitions, we adopt the notation from \cite{gfahsz-abcic-2013}.
Let $S_k=v_0 \, v_1 \dots v_k, v_{k+1}$,  $k \leq n-3$, be a sequence of vertices of a graph $G$
with  $d(v_0)>2$ and $d(v_i)=2$,  $i=1,\dots k-1$.
If $d(v_k)=1$, then $S_k$ is a {\it pendant path} of length $k+1$.
If $d(v_k) > 2$, then $S_k$ is an {\it internal path} of length $k$.

%As in \cite{gfahsz-abcic-2013}, a sequence of vertices of a graph $G$, $S_k=v_0 \, v_1 \dots v_k$,  will be  called a {\it pendant path} if 
%each two consecutive vertices in $S_k$ are adjacent in $G$, $d(v_0)>2$, $d(v_i)=2$, for $i=1,\dots k-1$, and $d(v_k)=1$.
%The length of the pendant path $S_k$ is $k$.

%Definition of the "breadth first search tree BFS-tree".

%A sequence $D=(d_1, d_2, \dots, d_n)$ is {\it graphical} if there is a graph
 %whose vertex degrees are $d_i$, $i=1,\dots,n$. If in addition
 %$d_1 \geq d_2\geq \dots \geq d_n$, then  $D$ is a
 %{\it degree sequence}.
%Let ${\bf D_n}$ be the set of all degree sequences of trees of length $n$.

In Section~\ref{sec:known} we give an overview of already known structural properties of the minimal-ABC trees.
In Sections~\ref{sec:B_1} and \ref{sec:B_2} 
%we consider the minimal-ABC trees and the possible number of $B_1$ and $B_2$-branches, respectively,
%that 
we present some results  and bounds on the number of $B_1$ and $B_2$-branches, respectively,
that may occur in minimal-ABC trees.
Conclusion and open problems are presented in Section~\ref{sec:Conclusion}.

%------------------------------------------------------------------------------
%
%\section[Known structural properties of the minimal-ABC trees]
%{Known structural properties of the minimal-ABC trees and some related results}\label{sec:known}
\section[Preliminaries and known structural properties of the minimal-ABC trees]
{Preliminaries and known structural properties of \\ the minimal-ABC trees}\label{sec:known}

% ----------------------------------------------------------------------------

%\subsection[Exactly proven properties]{Exactly proven properties}\label{sec:exactly}

A thorough overview of the known structural properties of the minimal-ABC trees was given in~\cite{gfahsz-abcic-2013}.
In addition to the results mentioned  there,
%, in the second part of this section, 
we present here also the recently obtained related results that we are aware of.

To determine the minimal-ABC tress of order less than $10$ is a trivial task, and those
trees are depicted in Figure~\ref{fig-all-till-9}. To simplify the exposition in the rest of the paper, we assume
that the trees of interest are of order at least $10$.

\begin{figure}[h]
\begin{center}
%\vspace{-0.3cm}
\includegraphics[scale=0.75]{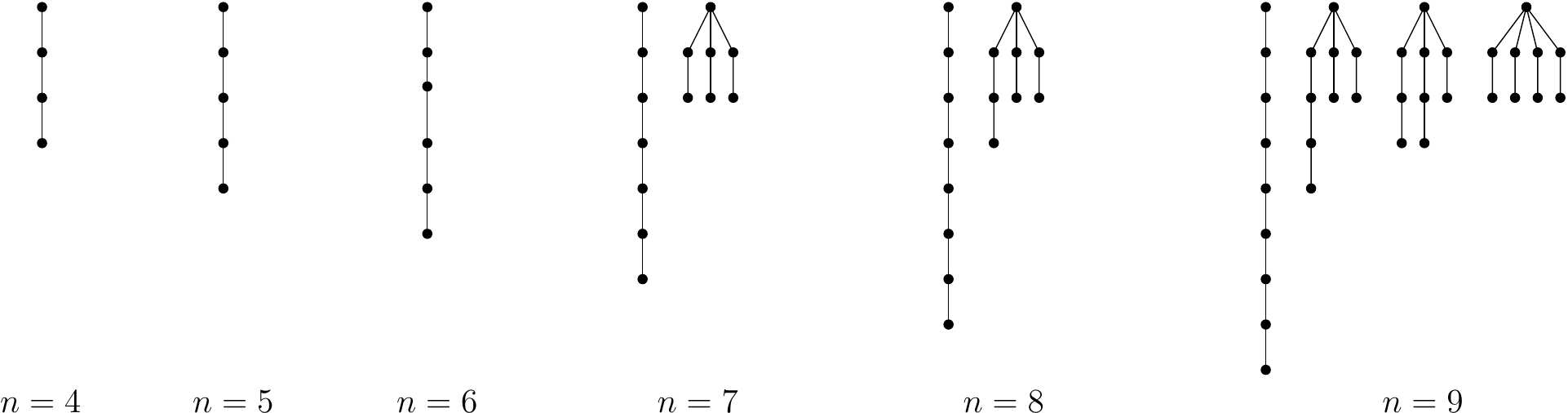}% [height=3.5cm,width=15cm]
\caption{Minimal-ABC trees of order $n$, $4 \leq n  \leq 9$.}
%\label{Unicyclic-Max-M3}
\label{fig-all-till-9}
%\vspace{-0.3cm}
\end{center}
\end{figure}

\smallskip
\noindent
In \cite{gfi-ntmabci-12}, Gutman, Furtula and Ivanovi{\' c}  obtained the following
results.

\begin{te}\label{thm-GFI-10}
The $n$-vertex tree with minimal ABC-index does not
contain internal paths of any length $k \geq 1$.
\end{te}

\begin{te}\label{thm-GFI-20}
The $n$-vertex tree with minimal ABC-index does not
contain pendant paths of length $k \geq 4$.
\end{te}

\noindent
An immediate, but important, consequence of Theorem~\ref{thm-GFI-10} is the next corollary.

\begin{co}\label{co-GFI-10}
 Let $T$ be a tree with minimal ABC index. Then the subgraph induced by the vertices of $T$ whose
degrees are greater than two is also a tree.
\end{co}

\noindent
An improvement of Theorem~\ref{thm-GFI-20} is the following result by Lin, Lin, Gao and Wu~\cite{llgw-pcgctmabci-13}.

\begin{te}\label{thm-LG-10}
Each pendant vertex of an $n$-vertex tree with minimal
ABC index belongs to a pendant path of length $k $, $2 \leq k \leq 3$.
\end{te}

%\noindent
%As a corollary of Theorems~\ref{thm-GFI-20} and \ref{thm-LG-10} the next result follows.

\begin{te}[\cite{gfi-ntmabci-12}] \label{thm-GFI-30}
The $n$-vertex tree with minimal ABC-index contains at
most one pendant path of length $3$.
\end{te}

\noindent
Before we state the next important result, we consider the following  definition
of a {\em greedy tree} provided by Wang in~\cite{w-etwgdsri-2008}.
%In~\cite{w-etwgdsri-2008} Wang defined a {\em greedy tree} as follows.

\begin{de}\label{def-GT}
Suppose the degrees of the non-leaf vertices are given, the greedy tree is achieved by the following `greedy algorithm':
\begin{enumerate}
\item Label the vertex with the largest degree as $v$ (the root).
\item Label the neighbors of $v$ as $v_1, v_2,\dots,$ assign the largest degree available to them such that $d(v_1) \geq d(v_2) \geq \dots$
\item Label the neighbors of $v_1$ (except $v$) as $v_{11}, v_{12}, \dots$ such that they take all the largest
degrees available and that $d(v_{11}) \geq d(v_{12}) \geq . . .$ then do the same for $v_2, v_3,\dots$
\item Repeat 3. for all newly labeled vertices, always starting with the neighbors of the labeled vertex with largest whose neighbors are not labeled yet.
\end{enumerate}
\end{de}
\noindent
%Further, we assume that the root vertex $v$ is at level $0$ of the tree. The vertices at level $i$ are at distance $i$ to $v$.

\noindent
The following result  by Gan, Liu and You~\cite{gly-abctgds-12} characterizes the trees with minimal ABC index with prescribed degree sequences. 
The same result, using slightly different notation and approach, was obtained by
Xing and Zhou \cite{xz-etfdsabci-2012}.

\begin{te}\label{thm-DS}
Given the degree sequence, the greedy tree minimizes the ABC index.
\end{te}

\noindent
The next result was obtained in~\cite{gfahsz-abcic-2013}. Alternatively it can be obtained as a corollary of Theorem~\ref{thm-DS}. 
\begin{te}\label{thm-GFAS}
If a minimal-ABC tree possesses three mutually adjacent vertices $v_1$, $v_2$, $v_3$, such that
$$
d(v_1) \geq d(v_2) \geq d(v_3),
$$
then $v_3$ must not be adjacent to both $v_1$ and $v_2$.
\end{te}

\begin{te}[\cite{d-sptmabci-2014}] \label{te-no5branches-10}
A minimal-ABC tree does not contain a $B_k$-branch, $k \geq 5$.
\end{te}

\begin{lemma} [\cite{d-sptmabci-2014}] \label{lemma-15}
A minimal-ABC tree does not contain 
\begin{enumerate}
\item[(a)] a $B_1$-branch and a $B_4$-branch,
\item[(b)] a $B_2$-branch and a $B_4$-branch, 
\end{enumerate}
that have a common parent vertex.
\end{lemma}

\begin{te}[\cite{d-sptmabci-2014}]  \label{thm-20}
A  minimal-ABC tree does not contain more than four $B_4$-branches.
\end{te}

\noindent
To best of our knowledge, the above mentioned results seems to be the only proven properties of the minimal-ABC trees.
%
%However, we think that it is worthy to mention in the next section some related basically obtained by computer assisted search.

%\subsection[Related results and computer search for minimal-ABC trees]
%{Related results and computer search for minimal-ABC trees}\label{sec:computer-assisted}

For complete characterization of the minimal-ABC trees,
besides the theoretically proven properties, % of  the trees with minimal ABC index, 
computer supported search can be of enormous help.
Therefore, 
we would like to mention in the sequel few related results.
%Therefore,  in the sequel we will mention few related results.

%In~\cite{fgiv-cstmabci-12} Furtula et al.~\cite
A  first significant example of using computer search was done by Furtula, Gutman, Ivanovi{\' c} and Vuki{\v c}evi{\' c}~\cite{fgiv-cstmabci-12}, 
where the trees with minimal ABC index of up to size of $31$ were computed,
and an initial conjecture of the general structure of the minimal-ABC trees was set.
There, a brute-force approach of generating all trees of a given order, 
speeded up by using  a distributed computing platform, was applied.
%The computation was speeded up by using  a distributed computing platform.
%There, a brute-force approach of generating all trees of a given order was applied.
%The computation was speed up by using  a distributed computing platform.
The plausible structural computational model  and its refined version presented there 
was based on the main assumption that  the minimal ABC tree posses a single {\it central vertex},
or said with other words, it is based  on the assumption that 
the vertices of a minimal ABC tree of degree $\geq 3$ induce a star graph.
This assumption was shattered by counterexamples presented in \cite{ahs-tmabci-13, ahz-ltmabci-13, d-ectmabci-2013, adgh-dctmabci-14}.
In this context, it is worth to mention that for a special class of trees, so-called {\em Kragujevac trees},
that are comprised of a central vertex and $B_k$-branches, $k \geq 1$ (see Figure~\ref{B_k-branches} for an illustration), the minimal-ABC tress were
fully characterized  by Hosseini, Ahmadi and Gutman~\cite{hag-ktmabci-14}.

In \cite{d-ectmabci-2013}  by considering only the degree sequences
of trees and some known structural properties of the trees with minimal ABC index
all trees with minimal ABC index of up to size of $300$ were computed.

Recently, in \cite{lccglc-fcstmaibtds-14}, by slightly modified version of the approach in  \cite{d-ectmabci-2013} 
all  minimal-ABC tree of up to size of $350$ were computed.

%The $B_k$-branches are depicted in Figure~\ref{B_k-branches}.

%
\begin{figure}[h]
\begin{center}
%\vspace{-0.3cm}
\includegraphics[scale=0.75]{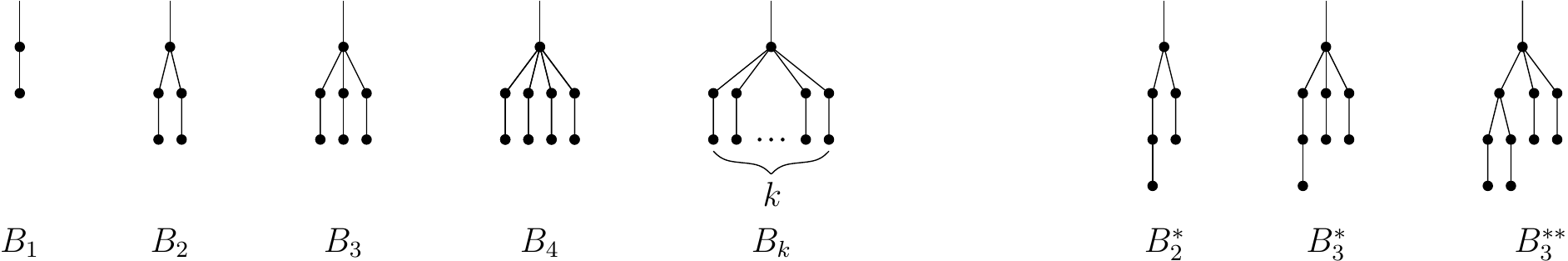}% [height=3.5cm,width=15cm]
\caption{$B_k$-branches.}
%\label{Unicyclic-Max-M3}
\label{B_k-branches}
%\vspace{-0.3cm}
\end{center}
\end{figure}

%\noindent
%The same result as in Theorem~\ref{thm-DS}, using slightly different notation and approach, was obtained by
%Xing and Zhou \cite{xz-etfdsabci-2012}.
%Since the  Theorem~\ref{thm-DS} plays a crucial role in our computation, the first important issue
%is how  to enumerate efficiently degree sequences of trees. This problem is considered in the next section. 
%Next we consider the problem of enumerating degree sequences of trees. 

%A vertex of BSF tree (greedy tree)  $u$, with $d(u)=k\geq 2$ we call a {\em k-degree terminal-vertex} if
%it closest pendant vertex is at distance at most three, and $u$ has at list one neighboring vertex of degree two.

%\subsection[Extra text for Intro (maybe)]{Extra text for Intro (maybe)}

By Theorem~\ref{thm-LG-10} and Corollary~\ref{co-GFI-10}, it follows that the minimal-ABC tree is comprised
of a tree $T$ to whose each pendant vertiex a $B_k$-branch is attached.
If $T$ is just a single vertex, then the minimal-ABC trees are the same trees that are minimal 
with respect to Kragujevac trees.
In this section, we present new results considering the types of $B_k$-branches that a minimal-ABC cannot contain.

Theorem~\ref{thm-GFI-30} says that there is at most one pendant path of length $k \geq 3$ in the tree with minimal ABC-index.
%Together with Theorem~\ref{thm-LG-10}, it follows that the rest of the pendant paths are of length $2$.
%
It was already observed in~\cite{hag-ktmabci-14} that the position of the path of length $k \geq 3$ does not have an influence 
on the value of the ABC index. Therefore, we assume that it is attached to the vertex of degree $4$, 
forming a $B_3^*$-branch (see Figure~\ref{B_k-branches} for an illustration).

%\section[Newly proven properties of the minimal-ABC trees]{New results} \label{sec:new}
%\subsection[Preliminaries]{Preliminaries} \label{sec:new}

The next  proposition is from~\cite{d-sptmabci-2014} and will be used in the proofs in the main text here.
The function $f(x,y)$ is defined as in (\ref{eqn:000}).

\begin{pro}  \label{appendix-pro-030}
Let $g(x,y)=-f(x,y)+f(x+\Delta x,y-\Delta y)$, with  real numbers $x, y \geq 2$,  $\Delta x \geq 0$,  $0 \leq \Delta y < y$.
Then, $g(x,y)$ increases in $x$ and decreases in $y$.
\end{pro}

\noindent
Due the symmetry of the function $f(.,.)$  Proposition~\ref{appendix-pro-030} can be rewritten as follows.

\begin{pro}  \label{appendix-pro-030-2}
Let $g(x,y)=-f(x,y)+f(x -\Delta x,y+\Delta y)$, with  real numbers $x, y \geq 2$,  $\Delta y \geq 0$,  $0 \leq \Delta x < x$.
Then, $g(x,y)$ decreases in $x$ and increases in $y$.
\end{pro}

\noindent
A $k$-{\em terminal vertex} of a rooted tree is a vertex with degree $k \geq 3$, that is adjacent to a pendant path of length two or three.
A $k$-{\em terminal branch}, referred as a {\it $T_k$-branch}, is a subtree induced by  a $(k+1)$-terminal vertex  and all its (direct and indirect) children vertices.
If the terminal vertex has at least one child with degree at least $3$, then we say that the $k$-terminal branch is {\em proper}.
Notice that $B_k$-branches are $T_k$-branches, but not proper $T_k$-branches, and the only proper $T_k$-branch in
Figure~\ref{B_k-branches} is the $B_3^{**}$-branch.
If a tree is comprised only of  one (proper)  $T_k$-branch, then we call the tree a {\it (proper)  $T_k$-tree}.
Observe that Kragujevac trees are $T_k$-trees.

\begin{pro} \label{pro-terminal-branches-10}
A minimal-ABC tree can contain at most one proper $T_k$-branch, $k \geq 2$.
\end{pro}
\begin{proof}
Let $u$ and $v$ be root vertices of two $T_k$-branches, and therefore terminal vertices, 
such that $d(u) \geq d(v)$.
Since $u$ is a root of $T_k$-branch, it has a child of degree $2$.
Due to Theorem~\ref{thm-DS} all direct children of $v$ cannot have degree bigger than $2$, 
which is a contradiction that $v$ is a terminal vertex.
\end{proof}

\noindent
All subtrees of a terminal vertex of a proper $T_k$-branch are $B_l$-branches. By Theorem~\ref{thm-DS}, we have $k \geq l$.

\section[Number of $B_1$-branches]{Number of $B_1$-branches} \label{sec:B_1}

In this section we analyze the occurrence of the $B_1$-branches in a minimal-ABC tree and we give
an upper bound on their number. Since the $B_1$-branches can occur only in the proper $T_k$-branches,
we focus our investigation here to these type of branches.

\begin{lemma} \label{lemma-B1-10}
A minimal-ABC tree does not contain a proper $T_k$-branch, $k \geq 13$, as a subtree.
Moreover, a minimal-ABC tree cannot be a proper $T_k$-tree itself if $k \geq 12$.
\end{lemma}
\begin{proof}
%It follows from  Lemma~\ref{lemma-10}($a$) and Theorems~\ref{te-no5branches-10}, 
%~\ref{thm-330}, and~\ref{thm-350}.
Denote by $T$ a proper $T_k$-branch with a root vertex $u$.
By Theorems~\ref{thm-DS}, \ref{te-no5branches-10} and Lemma~\ref{lemma-15}($a$), it follows that $T$ in
addition to $B_1$-branches may contain only $B_3$ and  $B_4$-branches.
Let $v$ be a child of $u$ with smallest degree larger than two.
 
First, we consider the case when $T$ is a subtree of a minimal-ABC tree $G$.
Assume that the number of  $B_1$-branches  contained in $T$ is $k_1 >0$, while the number of $B_2$ and $B_3$-branches is $k_2>0$.
It holds that $k_1+k_2=k$.
Perform the transformation $\mathcal{T}$ depicted in Figure~\ref{fig-B1-1}.
\begin{figure}[h]
\begin{center}
%\vspace{-0.3cm}
\includegraphics[scale=0.750]{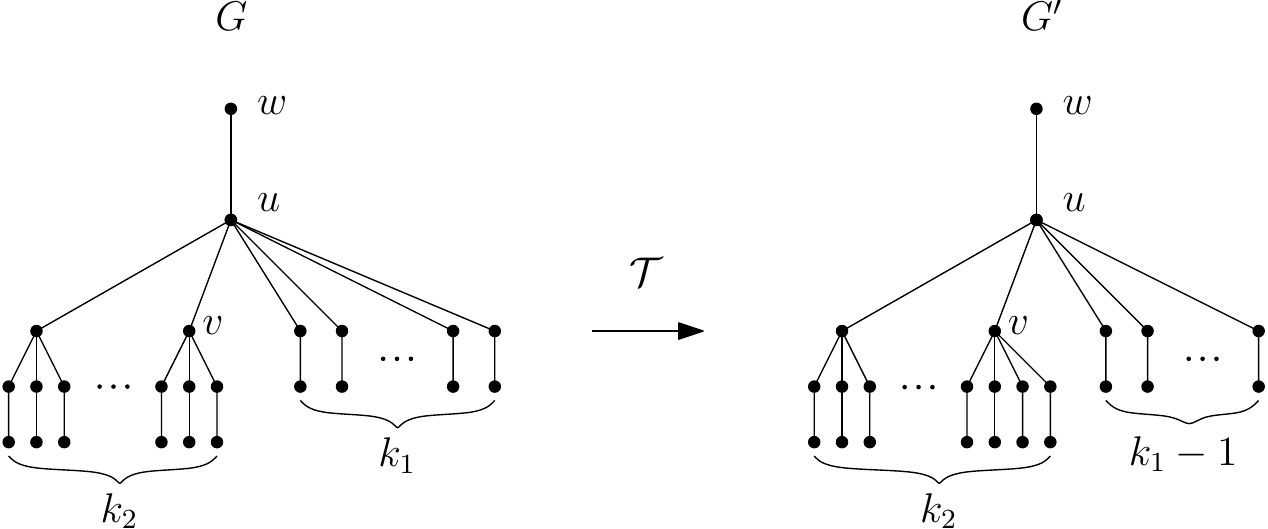}% [height=3.5cm,width=15cm]
\caption{Transforamation $\mathcal{T}$ from the proof of  Lemma~\ref{lemma-B1-10}. Note that 
in this illustration all $k_2$ children of $u$ are of degree $4$. As it is shown in the proof below,
in that case the change of the ABC index is the largest.}
%\label{Unicyclic-Max-M3}
\label{fig-B1-1}
%\vspace{-0.3cm}
\end{center}
\end{figure}
After this transformation the degree of the vertex $u$ decreases by one, while the degree of the vertex $v$ increases by one.
The degrees of the other vertices remain unchanged.
The change of the ABC index is
\beq \label{change-10}
&&-f(d(u),d(v))+f(d(u)-1,d(v)+1)+\sum_{i=1}^{d(u)-k_1-2}(-f(d(u),d(x_i))+f(d(u)-1,d(x_i))) \nonumber \\
&&-f(d(u),d(w))+f(d(u)-1,d(w)),
%-\sqrt{\frac{d(u_2)+6-2}{ 6 d(u_2)}} + \sqrt{\frac{d(u_2)+ 6 -2}{ 5 (d(u_2)+1)}}.
\eeq
where $x_i, i=1,\dots, d(u)-k_1-2$ are  children vertices of $u$ different than $v$, with degrees $3$ or $4$, and
$w$ is a parent vertex of $u$.
By Proposition~\ref{appendix-pro-030-2} the expression
$-f(d(u),d(v))+f(d(u)-1,d(v)+1)$ increases in $d(v)$, and therefore it is maximal for $d(v)=4$.
By the same proposition, the expressions
$-f(u,d(x_i))+f(d(u)-1,d(x_i))$ and $-f(d(u),d(w))+f(d(u)-1,d(w))$ increase in $x_i$ and $w$, respectively, and thus, 
the expressions are maximal for $d(x_i)=4, i=1,\dots, d(u)-3$, and $d(w) \to \infty$. Hence,
an upper bound on the expression  (\ref{change-10}) is
\beq \label{change-20}
&&-f(d(u),4)+f(d(u)-1,5)+(d(u)-k_1-2)(-f(d(u),4)+f(d(u)-1,4)) \nonumber \\
&&+\lim_{d(w) \to \infty} (-f(d(u),d(w))+f(d(u)-1,d(w))).
\eeq
Since $-f(u,4)+f(d(u)-1,4)$ is positive for $d(u)  > 1$, (\ref{change-20}) is bounded from above,  by
\beq \label{change-20-20}
&&-f(d(u),4)+f(d(u)-1,5)+(d(u)-3)(-f(d(u),4)+f(d(u)-1,4)) \nonumber \\
&&+\lim_{d(w) \to \infty} (-f(d(u),d(w))+f(d(u)-1,d(w))).
\eeq
The expressions
$-f(d(u),4)+f(d(u)-1,5)$ and $-f(d(u),d(w))+f(d(u)-1,d(w))$ decrease in $d(u)$ by Proposition~\ref{appendix-pro-030-2}.
Next we show that the expression $(d(u)-3)(-f(u,4)+f(d(u)-1,4))$ also decreases in $d(u)$.
The first derivative of $(d(u)-3)(-f(u,4)+f(d(u)-1,4))$ 
after a simplification is
\beq \label{change-40}
%&\ds \frac{d \, g(d(u_1))}{d  \, d(u_1)} = \nonumber \\
\ds &\frac{1}{2} \left(\sqrt{\frac{1+d(u)}{-1+d(u)}}-\sqrt{\frac{2+d(u)}{d(u)}}+\frac{1}{2} (d(u)-3)
 \left(-\frac{2}{(-1+d(u))^2 \sqrt{\frac{1+d(u)}{-1+d(u)}}}+\frac{2}{d(u)^2 \sqrt{\frac{2+d(u)}{d(u)}}}\right)\right), \nonumber
\eeq
which is a negative function for positive values of $d(u)$.
It follows that (\ref{change-20-20}) decreases in $d(u)$.
The smallest $d(u)$ for which  (\ref{change-20-20})  is negative ($\approx -0.0000943005$) is $d(u)=14$. 
Therefore, we may conclude that for any $k_1$, also (\ref{change-20}) and (\ref{change-10})  are negative if $d(u) \geq 14$.
Hence, the  change of the ABC index (\ref{change-10}), 
after applying the transformation $\mathcal{T}$, is negative, which is a contradiction to the assumption 
that $T$ is a subtree
of a tree with minimal-ABC index.

Consider now the case when $u$ is the root vertex of the tree with a minimal-ABC index.
We have the same configuration and apply the same transformation as in Figure~\ref{fig-B1-1}, with the only difference that
$u$ does not have a parent vertex.
Hence, it holds that $d(u)=k_1+k_2$.
Now the change of the ABC index is
\beq \label{change-60}
&&-f(d(u),d(v))+f(d(u)-1,d(v)+1)+\sum_{i=1}^{d(u)-k_1-1}(-f(u,d(x_i))+f(d(u)-1,d(x_i))). \nonumber \\
%-\sqrt{\frac{d(u_2)+6-2}{ 6 d(u_2)}} + \sqrt{\frac{d(u_2)+ 6 -2}{ 5 (d(u_2)+1)}}.
\eeq
Similarly as above we obtain that  (\ref{change-60}) as most
\beq \label{change-70}
&&-f(d(u),4)+f(d(u)-1,5)+(d(u)-k_1-1)(-f(u,4)+f(d(u)-1,4)),
\eeq
and it decreases in $d(u)$ and is maximal for $k_1=1$. 
The smallest $d(u)$ for which (\ref{change-70}) is negative ($\approx -0.000580929$) is $d(u)=12$.
%By Proposition~\ref{appendix-pro-040} the expression $-f(x,y)+f(x-1,y)$ increases in $y$, which yields
% $\lim_{d(w)\to \infty}$ $(-f(d(u),d(w))+f(d(u)-1,d(w)))>-f(u,4)+f(d(u)-1,4)$.
%Since, (\ref{change-30}) is negative, it follows that (\ref{change-70}), and consequently (\ref{change-60}) are negative.
Thus, in this case we again obtain that  
after applying the transformation $\mathcal{T}$, the value of the ABC index decrease, which is a contradiction to the assumption 
that $T$, $k \geq 12$, is a tree with minimal-ABC index.
\end{proof} 

\noindent
The next proposition is based on Lemma~\ref{lemma-B1-10}, and presents few configurations
that cannot be contained in a minimal-ABC tree.

\begin{pro} \label{pro-B1-10-02}
The proper $T_k$-branches depicted in Figure~\ref{fig-B1-05} cannot be subtrees of a minimal-ABC tree,
and the proper $T_k$-branches depicted in Figure~\ref{fig-B1-06} cannot be minimal-ABC trees.
\end{pro} 
\begin{proof}
%In Lemma~\ref{lemma-B1-10} it was shown that
%a minimal-ABC tree does not contain a proper $T_k$-branch, $k \geq 13$ (denoted by $T$),  as a subtree, and
%that a minimal-ABC tree cannot be a proper $T_k$-tree itself if $k \geq 12$.

First, consider the cases when $T$ is a proper subtree of  a minimal-ABC tree. % and $k \leq 12$ ($d(u) \leq 13$).
For a given value of $d(u)=d_g$, the first derivative of (\ref{change-20}) with respect to $k_1$ is
$$
-\frac{1}{2} \sqrt{\frac{1+d_g}{-1+d_g}}+\frac{1}{2} \sqrt{\frac{2+d_g}{d_g}},
$$
and it s negative for any positive $d_g$, 
from which follows that (\ref{change-20}) decreases in $k_1$ for any fixed value of $d(u)$.
Thus, for $d(u) \geq 14$ the smallest value of $k_1$ for which (\ref{change-20}) is negative is $1$, and
for $d(u)=13,12,11,10,9,8$  the smallest values of $k_1$ for which (\ref{change-20}) is negative are $3,4,5,6,6,6$, respectively.
Or expressed differently, (\ref{change-20}) is negative for
%A bit reformulated, for the case $d(u)=14$, it holds that (\ref{change-20}) is negative if
\beq \label{change-75}
& k_1+k_2 \geq 13  \quad \text{and} \quad k_1 \geq  1  \quad \quad (\text{the case}  \quad d(u) \geq 14);  \nonumber  \\
& k_1+k_2 =12 \quad \text{and} \quad k_1 \geq 3  \quad \quad (\text{the case}  \quad d(u) = 13);  \nonumber  \\ 
& k_1+k_2 = 11 \quad \text{and} \quad k_1 \geq 4  \quad \quad (\text{the case}  \quad d(u) = 12); \nonumber  \\
& k_1+k_2 = 10 \quad \text{and} \quad k_1 \geq 5  \quad \quad (\text{the case}  \quad d(u) = 11);   \nonumber  \\
& k_1+k_2 = \,9 \quad \; \text{and} \quad k_1 \geq 6  \quad \quad (\text{the case}  \quad d(u) = 10);   \nonumber  \\
& k_1+k_2 = 8 \quad  \text{and} \quad k_1 \geq 6  \quad \quad (\text{the case}  \quad d(u) = 9);   \nonumber  \\
& k_1+k_2 = 7 \quad  \text{and} \quad k_1 \geq 6  \quad \quad (\text{the case}  \quad d(u) = 8).   \nonumber
\eeq
From the above constrains, one can conclude that for $k_2=1,2,3,4$, the smallest value of $k_1$ for which (\ref{change-20}) is negative is $6$.
Similarly, (\ref{change-20}) is negative, for $k_2=5,6$ and $k_1 \geq 5$, 
for $k_2=7,8$ and $k_1 \geq 4$, 
for $k_2=9,10$ and $k_1 \geq 3$,
for $k_2=11$ and $k_1 \geq 2$, and 
for $k_2 \geq 12$ and  $k_1 \geq 1$.
From here, it follows that  the subtrees depicted in Figure~\ref{fig-B1-05} cannot occur in a minimal-ABC tree.
\begin{figure}[h!]
\begin{center}
%\vspace{-0.3cm}
\includegraphics[scale=0.750]{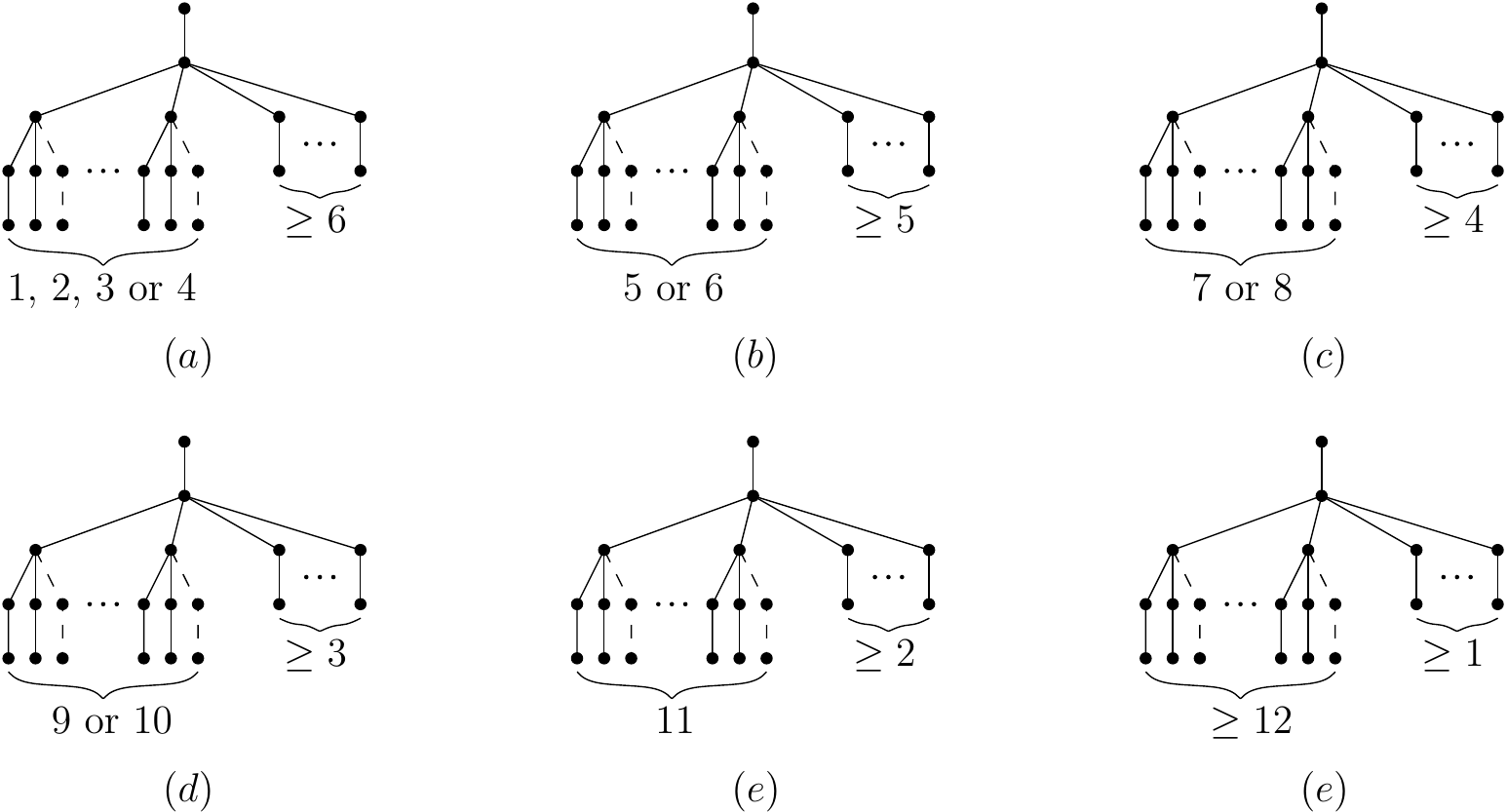}% [height=3.5cm,width=15cm]
\caption{Some subtrees that cannot occur in a minimal-ABC tree.
The dashed lines indicates that the particular branch can be a $B_3$ or a $B_2$-branch.}
%\label{Unicyclic-Max-M3}
\label{fig-B1-05}
%\vspace{-0.3cm}
\end{center}
\end{figure}
%
%It follows from Figure~\ref{fig-B1-05}, that
%if  a minimal-ABC tree $T$ contains 
%\begin{itemize}
%\item[$(a)$] one or two $B_1$-branches, than it  cannot contain a $T_{\geq 13}$-branch, $k \geq 13$;
%\item[$(b)$] three $B_1$ branches, than it  cannot contain a $T_{\geq 12}$-branch, $k \geq 12$;
%\item[$(c)$] four $B_1$ branches, than it  cannot contain a $T_{\geq 11}$-branch, $k \geq 11$;
%\item[$(d)$] five $B_1$ branches, than it  cannot contain a $T_{\geq 10}$-branch, $k \geq 10$;
%\item[$(e)$] six or more $B_1$ branches, than it  cannot contain a $T_k$-branch,  $k>6$.
%\end{itemize}

\noindent
In the case when $T$ is  a minimal-ABC tree itself, we have obtain in  Lemma~\ref{lemma-B1-10} that
there is no minimal-ABC tree that is $B_{\geq 12}$-branch.
Analogously, as in the case when $T$ is a subtree of a minimal-ABC tree, 
in this case we obtain that for $d(u)=11,10,9,8,7,6$  the smallest values of $k_1$ for which (\ref{change-70}) 
is negative are $2,3,4,4,5,5$, respectively, and an identical analysis as above show that 
the trees depicted in Figure~\ref{fig-B1-06} cannot be  minimal-ABC trees.
\begin{figure}[h!]
\begin{center}
%\vspace{-0.3cm}
\includegraphics[scale=0.750]{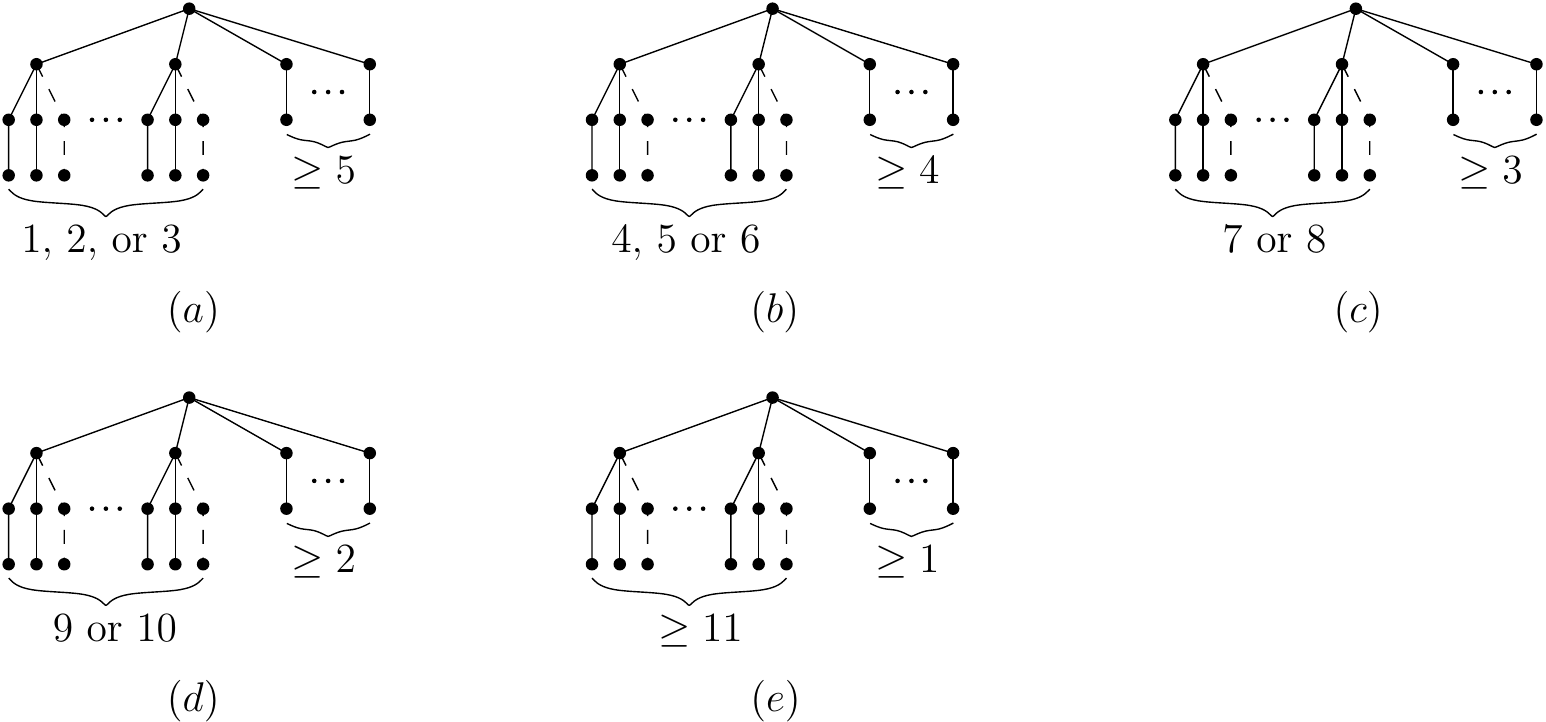}% [height=3.5cm,width=15cm]
\caption{Some trees that cannot be  a minimal-ABC trees.
The dashed lines indicates that the particular branch can be a $B_3$ or a $B_2$-branch.}
%\label{Unicyclic-Max-M3}
\label{fig-B1-06}
%\vspace{-0.3cm}
\end{center}
\end{figure}
%

%It follows from Figure~\ref{fig-B1-06}, that
%if  a minimal-ABC tree $T$ contains
%\begin{itemize}
%\item[$(a)$] one $B_1$-branch, than it  cannot  be a $T_{\geq 12}$-branch;
%\item[$(b)$] two $B_1$-branches, than it  cannot  be a $T_{\geq 11}$-branch;
%\item[$(c)$] three $B_1$ branches, than it  cannot  be a $T_{\geq 10}$-branch;
%\item[$(d)$] four $B_1$ branches, than it  cannot  a be a $T_{\geq 8}$-branch;
%\item[$(e)$] five or more $B_1$ branches, than it  cannot be a $T_k$-branch, $k>5$.
%\end{itemize}
%
\end{proof}

\noindent
Next we present special cases of Lemma~\ref{lemma-B1-10} and  Proposition~\ref{pro-B1-10-02}, 
with a relaxation that a proper $T_k$-branch contains only $B_2$ and $B_1$-branches.

%Next results, presented in Lemma~\ref{lemma-B1-10-2} and  Proposition~\ref{pro-B1-10-00}, 
%will be used in the proof of the Theorem~\ref{theorem-B1-10} and some of the results in the next section (state them later explicitly!).
%The following Lemma~\ref{lemma-B1-10-2} and  Proposition~\ref{pro-B1-10-00} are special cases
%of Lemma~\ref{lemma-B1-10} and  Proposition~\ref{pro-B1-10-02}, since in the former two
%the $B_3$-branches are excluded.

\begin{lemma} \label{lemma-B1-10-2}
A minimal-ABC tree $T$ does not contain a proper $T_k$-branch, $k \geq 8$, as a subtree, if the $T_k$-branch is comprised only
of $B_2$ and $B_1$-branches.
Moreover, also in this case $T$ cannot be a proper $T_k$-tree itself if $k \geq 7$.
\end{lemma}
\begin{proof}
%It follows from  Lemma~\ref{lemma-10}($a$) and Theorems~\ref{te-no5branches-10}, 
%~\ref{thm-330}, and~\ref{thm-350}.
%This case is illustraded in Figure~\ref{fig-B1-1-10}.
%This lemma is a particular case of Lemma~\ref{lemma-B1-10}, cause here the $B_3$-branches are excluded
%from  the $T_k$-branch.
%the $B_3$-branches are excluded
We proceed with the analog transformation (see Figure~\ref{fig-B1-1-10}) and analysis as in Lemma~\ref{lemma-B1-10}.
So, we omit most of the details that were  mentioned in Lemma~\ref{lemma-B1-10}.
\begin{figure}[h]
\begin{center}
%\vspace{-0.3cm}
\includegraphics[scale=0.750]{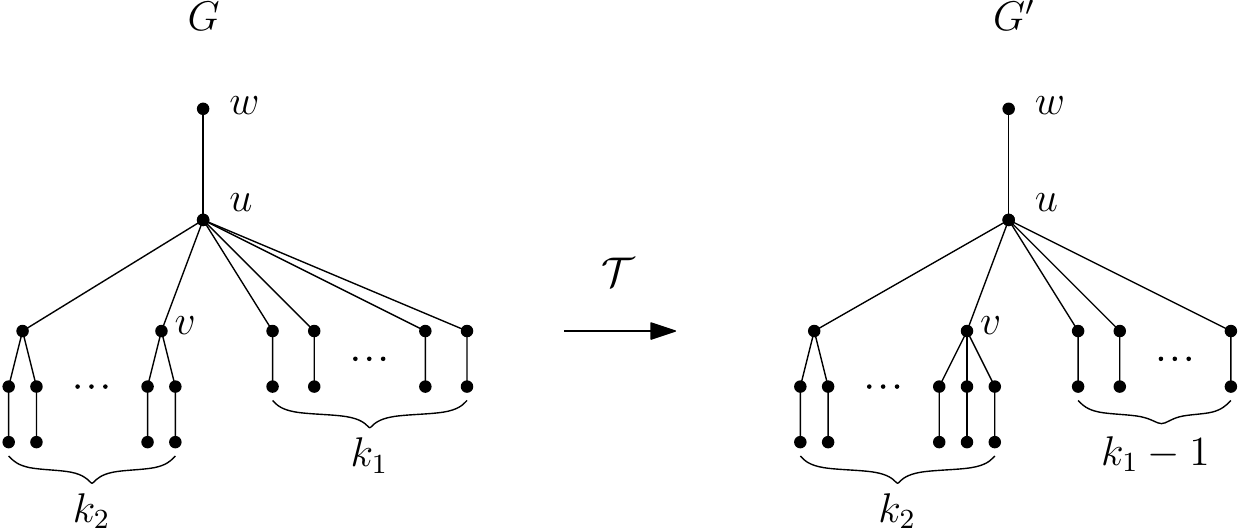}% [height=3.5cm,width=15cm]
\caption{Transforamation $\mathcal{T}$ from  the proof of Lemma~\ref{lemma-B1-10-2}.}
%\label{Unicyclic-Max-M3}
\label{fig-B1-1-10}
%\vspace{-0.3cm}
\end{center}
\end{figure}
After this transformation the degree of the vertex $u$ decreases by one, while the degree of the vertex $v$ increases by one.
The degrees of other vertices remain unchanged.
The change of the ABC index is at most
\beq \label{change-20-2a}
&&-f(d(u),3)+f(d(u)-1,4)+(d(u)-k_1-2)(-f(d(u),3)+f(d(u)-1,3)) \nonumber \\
&&+\lim_{d(w) \to \infty} (-f(d(u),d(w))+f(d(u)-1,d(w))). 
\eeq
and it is  bounded from above by
\beq \label{change-20-2}
&&-f(d(u),3)+f(d(u)-1,4)+(d(u)-3)(-f(u,3)+f(d(u)-1,3)) \nonumber \\
&&+\lim_{d(w) \to \infty} (-f(d(u),d(w))+f(d(u)-1,d(w))).
\eeq
By the same arguments as in the proof of Lemma~\ref{lemma-B1-10}, it follows that the expressions
$-f(u,3)+f(d(u)-1,3)$, $(d(u)-3)(-f(u,3)+f(d(u)-1,3))$ and $-f(d(u),d(w))+f(d(u)-1,d(w))$ decrease in $d(u)$.  
The smallest $d(u)$ for which (\ref{change-20-2}) is negative is 
$d(u)=9$. Hence, for $d(u)\geq 9$ or $k \geq 8$, the  change of the ABC index, 
after applying the transformation $\mathcal{T}$, is negative, which is a contradiction to the assumption 
that $T_{k}$ belongs to a tree with minimal-ABC index.

Consider now the case when $u$ is the root vertex of the tree with a minimal-ABC index.
We have the same configuration and apply the same transformation as in Figure~\ref{fig-B1-1-10}, with the only difference that
$u$ does not have a parent vertex.
Hence, it holds that $d(u)=k_1+k_2$.
Now the change of the ABC index is
\beq \label{change-60-2}
&&-f(d(u),d(v))+f(d(u)-1,d(v)+1)+\sum_{i=1}^{d(u)-k_1-1}(-f(u,d(x_i))+f(d(u)-1,d(x_i))). \nonumber \\
%-\sqrt{\frac{d(u_2)+6-2}{ 6 d(u_2)}} + \sqrt{\frac{d(u_2)+ 6 -2}{ 5 (d(u_2)+1)}}.
\eeq
Similarly as above, we obtain that  (\ref{change-60-2}) is as most
\beq \label{change-70-2}
&&-f(d(u),3)+f(d(u)-1,4)+(d(u)-k_1)(-f(u,3)+f(d(u)-1,3)),
\eeq
it decreases in $d(u)$ and is maximal for $k_1=1$. 
The smallest $d(u)$ for which (\ref{change-70-2}) is negative is $d(u)=7$.
%By Proposition~\ref{appendix-pro-040} the expression $-f(x,y)+f(x-1,y)$ increases in $y$, which yields
% $\lim_{d(w)\to \infty}$ $(-f(d(u),d(w))+f(d(u)-1,d(w)))>-f(u,4)+f(d(u)-1,4)$.
%Since, (\ref{change-30}) is negative, it follows that (\ref{change-70}), and consequently (\ref{change-60}) are negative.
Thus, in this case we again obtain that  
after applying the transformation $\mathcal{T}$, the value of the ABC index decreases, which is a contradiction to the assumption 
that $T_{k}$ is a tree with minimal-ABC index.
\end{proof}

\begin{pro} \label{pro-B1-10-00}
The proper $T_k$-branches depicted in Figure~\ref{fig-B1-1-20} cannot be subtrees of a minimal-ABC tree,
and the proper $T_k$-branches depicted in Figure~\ref{fig-B1-1-30} cannot be minimal-ABC trees.
\end{pro} 
\begin{proof}
Since this proposition is a special case of Proposition~\ref{pro-B1-10-02},  the derivation of their proofs are analogous.
%In Lemma~\ref{lemma-B1-10} it was shown that
%a minimal-ABC tree does not contain a proper $T_k$-branch, $k \geq 13$ (denoted by $T$),  as a subtree, and
%that a minimal-ABC tree cannot be a proper $T_k$-tree itself if $k \geq 12$.
First, we consider the cases when $T$ is a proper subtree of  a minimal-ABC tree. % and $k \leq 12$ ($d(u) \leq 13$).
For a given value of $d(u)=d_g$, the first derivative of (\ref{change-20-2a}) with respect to $k_1$ is
$$
\frac{\sqrt{1+\frac{1}{d_g}}-\sqrt{\frac{d_g}{-1+d_g}}}{\sqrt{3}},
$$
and it s negative for any positive $d_g$, 
from which follows that (\ref{change-20-2a}) decreases in $k_1$ for any fixed value of $d(u)$.
Thus, for $d(u) \geq 9$ the smallest value of $k_1$ for which (\ref{change-20-2a}) is negative is $1$, and
for $d(u)=8,7,6$  the smallest values of $k_1$ for which (\ref{change-20-2a}) is negative are $2,4,5,$ respectively.
Or expressed differently, (\ref{change-20-2a}) is negative for
%A bit reformulated, for the case $d(u)=14$, it holds that (\ref{change-20}) is negative if
\beq \label{change-75}
& k_1+k_2 \geq 8  \quad \text{and} \quad k_1 \geq 1   \quad \quad (\text{the case}  \quad d(u) \geq 9);  \nonumber  \\
& k_1+k_2 =7 \quad \text{and} \quad k_1 \geq 2  \quad \quad (\text{the case}  \quad d(u) = 8);  \nonumber  \\ 
& k_1+k_2 = 6 \quad \text{and} \quad k_1 \geq 4  \quad \quad (\text{the case}  \quad d(u) = 7). \nonumber
\eeq
For $d(u) = 6$, i.e., $k_1+k_2 = 5$, (\ref{change-20-2a}) is negative if  $k_1 \geq 5$. However, 
this is not a feasible combination, since $k_2$ must be positive.
From the above constrains, one can conclude that for $k_2=1$, the smallest value of $k_1$ for which (\ref{change-20-2a}) is negative is $5$.
Similarly, (\ref{change-20-2a}) is negative, for $k_2=2,3$ and $k_1 \geq 4$, 
for $k_2=4$ and $k_1 \geq 3$, 
for $k_2=5,6$ and $k_1 \geq 2$, and
for $k_2 \geq 7$ and $k_1 \geq 1$.
From here, it follows that  the subtrees depicted in Figure~\ref{fig-B1-1-20} cannot occur in a minimal-ABC tree.
\begin{figure}[h!]
\begin{center}
%\vspace{-0.3cm}
\includegraphics[scale=0.750]{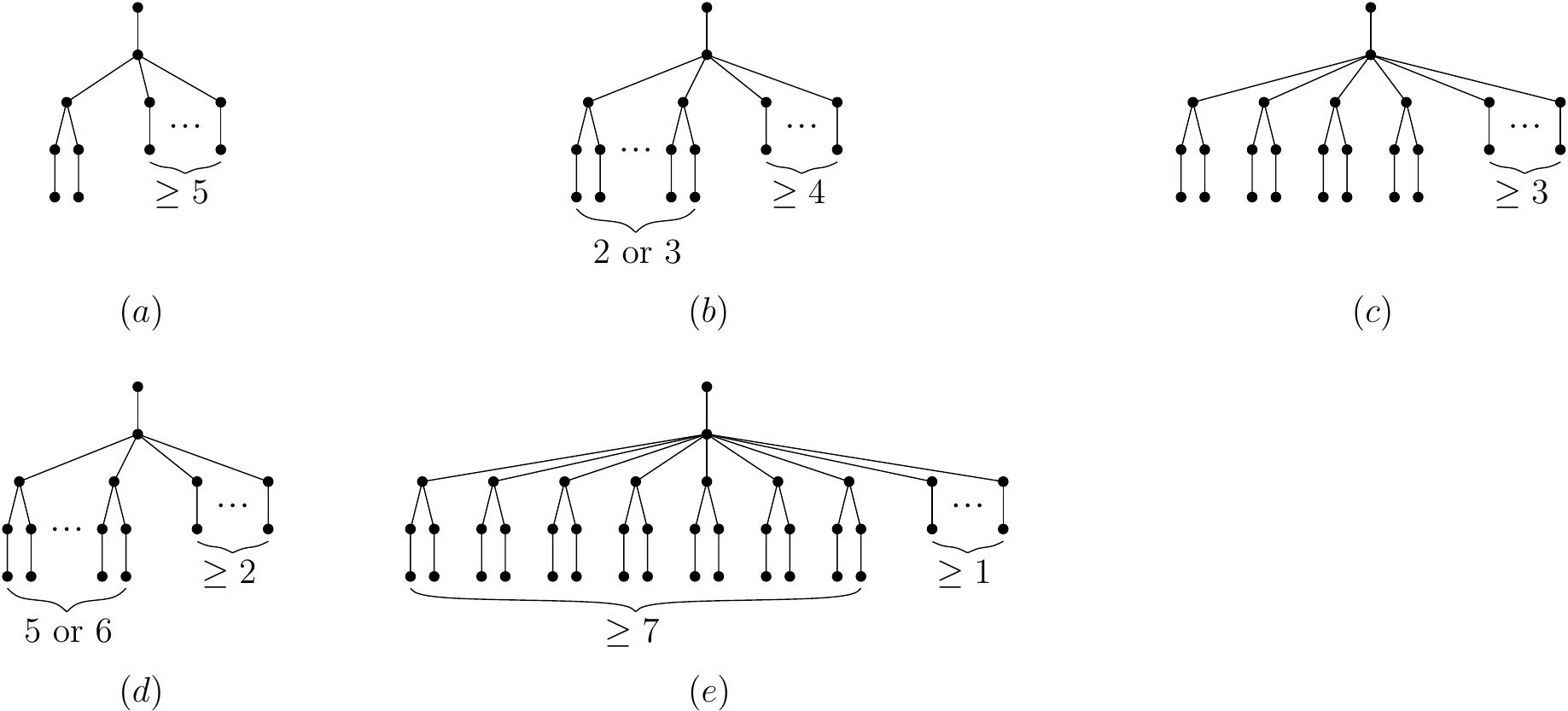}% [height=3.5cm,width=15cm]
\caption{Some subtrees that cannot occur in a minimal-ABC tree.}
%\label{Unicyclic-Max-M3}
\label{fig-B1-1-20}
%\vspace{-0.3cm}
\end{center}
\end{figure}
%
%It follows from Figure~\ref{fig-B1-05}, that
%if  a minimal-ABC tree $T$ contains 
%\begin{itemize}
%\item[$(a)$] one or two $B_1$-branches, than it  cannot contain a $T_{\geq 13}$-branch, $k \geq 13$;
%\item[$(b)$] three $B_1$ branches, than it  cannot contain a $T_{\geq 12}$-branch, $k \geq 12$;
%\item[$(c)$] four $B_1$ branches, than it  cannot contain a $T_{\geq 11}$-branch, $k \geq 11$;
%\item[$(d)$] five $B_1$ branches, than it  cannot contain a $T_{\geq 10}$-branch, $k \geq 10$;
%\item[$(e)$] six or more $B_1$ branches, than it  cannot contain a $T_k$-branch,  $k>6$.
%\end{itemize}

\noindent
In the case when $T$ is  a minimal-ABC tree itself, we have obtain in Lemma~\ref{lemma-B1-10-2} that
there is no minimal-ABC tree that is a proper $T_k$-branch, $k \geq 7$.
Analogously, as in the case when $T$ is a subtree of a minimal-ABC tree, 
in this case we obtain that for $d(u)=6,5$  the smallest values of $k_1$ for which (\ref{change-70-2}) 
is negative are $2,3$, respectively, and identical analysis as above show that 
the trees depicted in Figure~\ref{fig-B1-1-30} cannot be  minimal-ABC trees.

\begin{figure}[h]
\begin{center}
%\vspace{-0.3cm}
\includegraphics[scale=0.750]{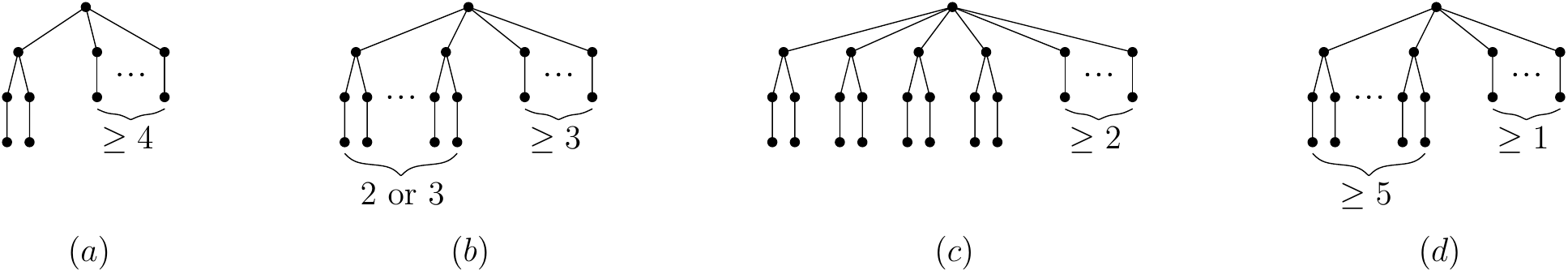}% [height=3.5cm,width=15cm]
\caption{Some trees that are not minimal-ABC trees.}
%\label{Unicyclic-Max-M3}
\label{fig-B1-1-30}
%\vspace{-0.3cm}
\end{center}
\end{figure}
\end{proof}

\noindent
The following result is the main result in this section and it gives an upper bound 
on the number of  $B_1$-branches in a minimal-ABC tree.

\begin{te} \label{theorem-B1-10}
A minimal-ABC $G$ tree can contain at most four $B_1$-branches.
Moreover, if $G$ is a $T_k$-branch itself, then it can contain at most three $B_1$-branches.
\end{te}
\begin{proof}
Here we consider again two cases: when $G$ has a $T_k$-branch as subtree or when
$G$ is a $T_k$-branch itself.
Recall that by Theorem~\ref{te-no5branches-10} and Lemma~\ref{lemma-15}($a$), a $T_k$-branch does not contains a $B_4$-branch.

\bigskip
\noindent
{\bf Case $1$.} $G$ has a $T_k$-branch as subtree.

\noindent
Let $G$ be a  minimal-ABC tree that have more than three $B_1$-branches.
If the $T_k$-branch contains only $B_2$-branches as its children, then by Proposition~\ref{pro-B1-10-00} (Figure~\ref{fig-B1-1-20}($a$)),
it cannot contains more than $4$ $B_1$-branches.
So we assume that the $T_k$-branch contains at least one $B_3$-branches.
Observe, that by Proposition~\ref{pro-B1-10-02} (Figure~\ref{fig-B1-05}($a$)),
it $G$ cannot contain a $T_k$-branch with more than $5$ $B_1$-branches.
In this case, we perform the transformation $\mathcal{T}$ depicted in Figure~\ref{fig-B1-2}.
\begin{figure}[h]
\begin{center}
%\vspace{-0.3cm}
\includegraphics[scale=0.750]{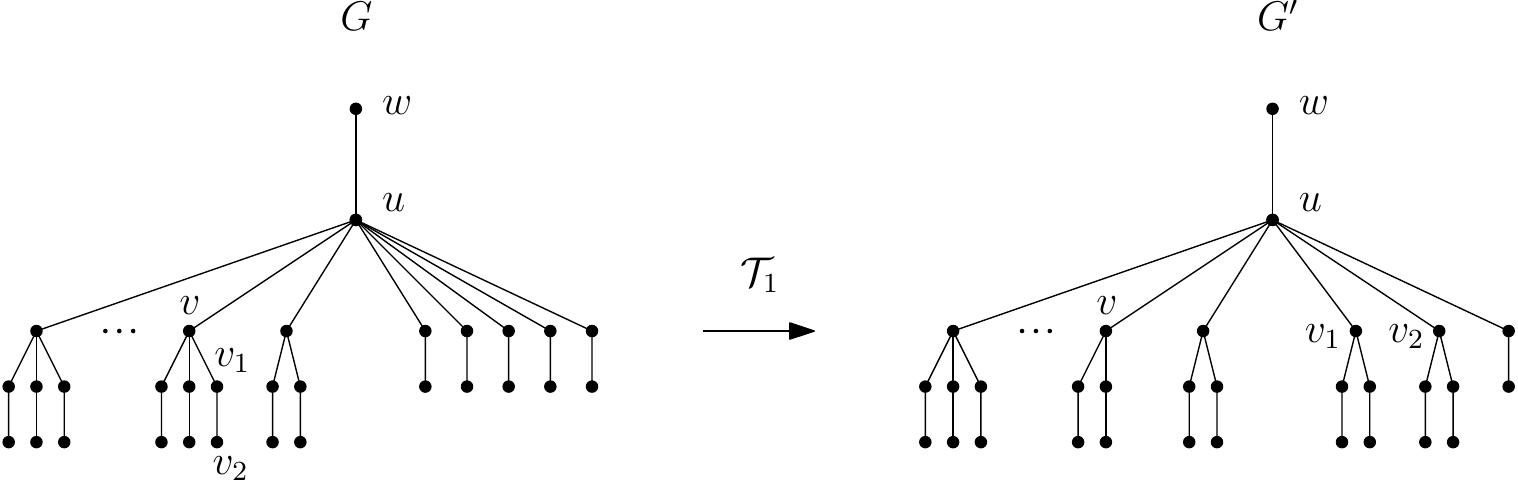}% [height=3.5cm,width=15cm]
\caption{Transforamation $\mathcal{T}_1$ from proof of  Theorem~\ref{theorem-B1-10},~Case $1$.}
%\label{Unicyclic-Max-M3}
\label{fig-B1-2}
%\vspace{-0.3cm}
\end{center}
\end{figure}
After this transformation the degree of the vertex $u$ decreases by two, while the degree of the vertex $v$ decreases by one.
The degree of the vertex $v_1$ increases by one, and degree of the vertex $v_2$ increases by two.
The degrees of other vertices remain unchanged.
The change of the ABC index is
\beq \label{change-80}
&&-f(d(u),4)+f(d(u)-2,3)+\sum_{i=1}^{d(u)-7}(-f(d(u),d(x_i))+f(d(u)-2,d(x_i))) \nonumber \\
&&-f(4,2)+f(d(u)-2,3)-f(2,1)+f(d(u)-2,3)) -f(d(u),d(w))+f(d(u)-2,d(w)), \nonumber \\
%-\sqrt{\frac{d(u_2)+6-2}{ 6 d(u_2)}} + \sqrt{\frac{d(u_2)+ 6 -2}{ 5 (d(u_2)+1)}}.
\eeq
where $x_i, i=1,\dots, d(u)-7$ are  children vertices of $u$ different than $v$, with degrees $3$ or $4$, and
$w$ is a parent vertex of $u$.
By Proposition~\ref{appendix-pro-030-2} the expressions $-f(d(u),d(x_i))+f(d(u)-2,d(x_i))$ and $-f(d(u),d(w))+f(d(u)-2,d(w))$
increase in $x_i$, $i=1, \dots, d(u)-7$, and $d(w)$, respectively. Thus, (\ref{change-80}) is bounded from above by
\beq \label{change-90}
&&-f(d(u),4)+f(d(u)-2,3)+(d(u)-7)(-f(d(u),4)+f(d(u)-2,4)) \nonumber \\
&&-f(4,2)+f(d(u)-2,3)-f(2,1)+f(d(u)-2,3))   \nonumber \\
&&\lim_{d(w) \to \infty}(-f(d(u),d(w))+f(d(u)-2,d(w))).
%-\sqrt{\frac{d(u_2)+6-2}{ 6 d(u_2)}} + \sqrt{\frac{d(u_2)+ 6 -2}{ 5 (d(u_2)+1)}}.
\eeq
by Proposition~\ref{pro-B1-10-02} (Figure~\ref{fig-B1-05}($a$)), it follows that 
$u$ may have at most $4$ children of degree larger than $2$.
Thus, $7 \leq d(u) \leq 10$.
For all $4$ possible values of  $d(u)$, (\ref{change-90}) is largest for $d(u)=7$, and its value is
$\approx -0.0145446$.
Thus, after applying the transformation from Figure~\ref{fig-B1-2} the  ABC index decreases,
which is in a contradiction that $G$ is a minimal-ABC tree.

\bigskip
\noindent
{\bf Case $2$.} $G$ is a $T_k$-branch itself.

\noindent
Let $G$ be a  minimal-ABC tree that have more than four $B_1$-branches.
If the $T_k$-branch contains only $B_2$-branches as its children, then by Proposition~\ref{pro-B1-10-00} (Figure~\ref{fig-B1-1-30}($a$)),
it cannot contains more than $3$ $B_1$-branches.
So we assume that the $T_k$-branch contains at least one $B_3$-branches.
Observe, that by Proposition~\ref{pro-B1-10-02} (Figure~\ref{fig-B1-06}($a$)),
$G$ cannot contain a $T_k$-branch with more than $4$ $B_1$-branches.
In this case, we perform the same transformation $\mathcal{T}$ as in Figure~\ref{fig-B1-2}
(the only difference in this case is that there is no vertex $w$ -$u$ is the root of the tree, and 
there are $4$ $B_1$-branches).
After this transformation the degree of the vertex $u$ decreases by two, while the degree of the vertex $v$ decreases by one.
The degree of the vertex $v_1$ increases by one, and degree of the vertex $v_2$ increases by two.
The degrees of other vertices remain unchanged.
The change of the ABC index is
\beq \label{change-100}
&&-f(d(u),4)+f(d(u)-2,3)+\sum_{i=1}^{d(u)-5}(-f(d(u),d(x_i))+f(d(u)-2,d(x_i))) \nonumber \\
&&-f(4,2)+f(d(u)-2,3)-f(2,1)+f(d(u)-2,3)),
%-\sqrt{\frac{d(u_2)+6-2}{ 6 d(u_2)}} + \sqrt{\frac{d(u_2)+ 6 -2}{ 5 (d(u_2)+1)}}.
\eeq
where $x_i, i=1,\dots, d(u)-5$ are  children vertices of $u$ different than $v$, with degrees $3$ or $4$.
By Proposition~\ref{appendix-pro-030-2} the expressions $-f(d(u),d(x_i))+f(d(u)-2,d(x_i))$ 
increase in $x_i$, $i=1, \dots, d(u)-7$. Thus, (\ref{change-100}) is bounded from above by
\beq \label{change-110}
&&-f(d(u),4)+f(d(u)-2,3)+(d(u)-5)(-f(d(u),4)+f(d(u)-2,4)) \nonumber \\
&&-f(4,2)+f(d(u)-2,3)-f(2,1)+f(d(u)-2,3)).  
%-\sqrt{\frac{d(u_2)+6-2}{ 6 d(u_2)}} + \sqrt{\frac{d(u_2)+ 6 -2}{ 5 (d(u_2)+1)}}.
\eeq
by Proposition~\ref{pro-B1-10-02} (Figure~\ref{fig-B1-06}($a$)), it follows that 
$u$ may have at most $3$ children of degree larger than $2$.
Thus, $5 \leq d(u) \leq 7$.
For all $3$ possible values of  $d(u)$, (\ref{change-110}) is largest for $d(u)=5$, and its value is
$\approx -0.00582154$.
Thus, after applying the above  transformation the  ABC index decreases,
which is in a contradiction that $G$ is a minimal-ABC tree.
\end{proof}

\noindent
In the next section we analyze the $B_2$-branches and there occurrence in the minimal-ABC trees.

\section[Number of $B_2$-branches]{Number of $B_2$-branches} \label{sec:B_2}

First, we present two configurations that cannot occur as subtrees of a minimal-ABC tree.
Their exclusion will be considered in the proofs of some of the results presented later in 
this section.

\begin{pro} \label{pro-B2-10}
The tree depicted in Figure~\ref{fig-B2-1}~$(a)$ cannot be a subtree of a minimal-ABC tree. 
%if $d(u)  \geq 12$ (or $d(u)  \geq 10$ if $u$ is a root vertex).
\end{pro}
\begin{proof}
\begin{figure}[h!]
\begin{center}
%\vspace{-0.3cm}
\includegraphics[scale=0.75]{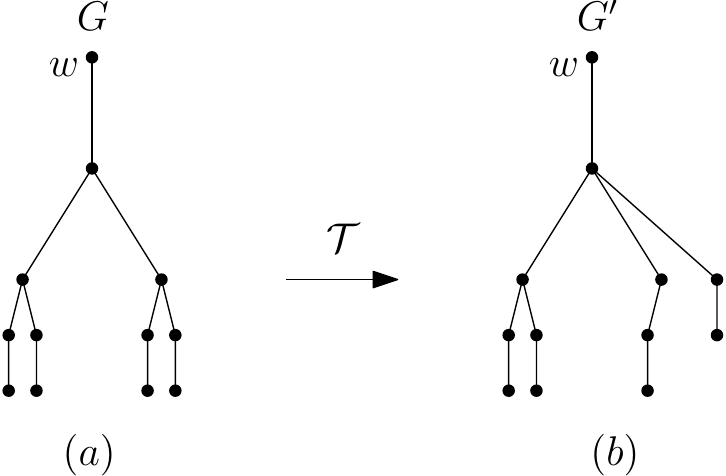}% [height=3.5cm,width=15cm]
\caption{An illustration of the transformation $\mathcal{T}$ from the proof of Proposition~\ref{pro-B2-10}.}
%\label{Unicyclic-Max-M3}
\label{fig-B2-1}
%\vspace{-0.3cm}
\end{center}
\end{figure}
After applying the transformation $\mathcal{T}$ from  Figure~\ref{fig-B2-1}, the change of the ABC index is
\beq \label{change-B2-10}
&&-f(d(w),3)+f(d(w),4)-f(3,3)+f(3,2)-f(3,3)+f(4,3)).
\eeq
By Theorem~\ref{thm-DS}, $d(w)$ cannot be smaller than the degrees of its children vertices.
By Proposition~\ref{appendix-pro-030-2}, $-f(d(w),3)+f(d(w),4)$ decreases in $d(w)$, and thus, it is maximal for $d(w)=3$. 
Therefore, (\ref{change-B2-10}) is at most
\beq \label{change-B2-20}
&&-f(3,3)+f(3,4)-f(3,3)+f(3,2)-f(3,3)+f(4,3)) < -0.0018988. \nonumber
\eeq
\end{proof}

\begin{pro} \label{pro-B2-20}
The tree depicted in Figure~\ref{fig-B2-2}~$(a)$ cannot be a subtree of a minimal-ABC tree.
\begin{figure}[h]
\begin{center}
%\vspace{-0.3cm}
\includegraphics[scale=0.75]{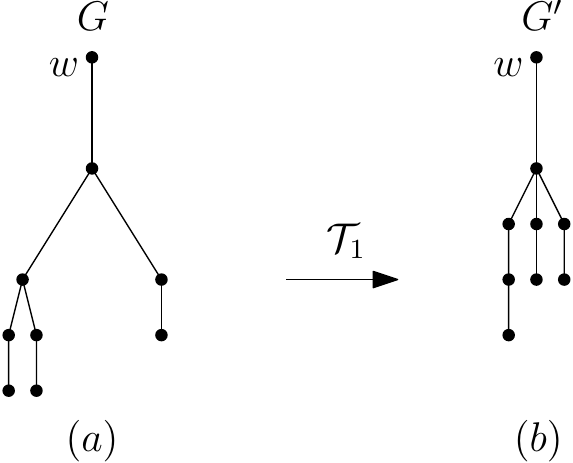}% [height=3.5cm,width=15cm]
\caption{An illustration of the transformation $\mathcal{T}_1$ from proof of Proposition~\ref{pro-B2-20}.}
%\label{Unicyclic-Max-M3}
\label{fig-B2-2}
%\vspace{-0.3cm}
\end{center}
\end{figure}
\end{pro}
\begin{proof}
The change of the ABC index after applying the transformation  $\mathcal{T}_1$ depicted in  Figure~\ref{fig-B2-2},  is
\beq \label{change-B2-30}
&&-f(d(w),3)+f(d(w),4)-f(3,3)+f(4,2).
\eeq
Similarly as in the previous proposition, we conclude that the expression $-f(d(w),3)+f(d(w),4)$ decreases in $d(w)$ 
and it is maximal for $d(w)=3$.
A straightforward verification of (\ref{change-B2-30}) shows that it is negative for $d(w) \geq 5$.
Next, we consider separately the cases $d(w)= 3$ and $d(w)= 4$.

\bigskip
\noindent
{\bf Case $1$.} $d(w)= 3$.

\smallskip
\noindent
In this case,
the vertex $w$, beside the child $u$, has one more child,denoted by $v$, which by Theorems~\ref{thm-GFI-10}, \ref{thm-DS}  
and Proposition~\ref{pro-terminal-branches-10} has degree $3$ (see Figure~\ref{fig-B2-3}~(a) for an illustration).
\begin{figure}[h!]
\begin{center}
%\vspace{-0.3cm}
\includegraphics[scale=0.75]{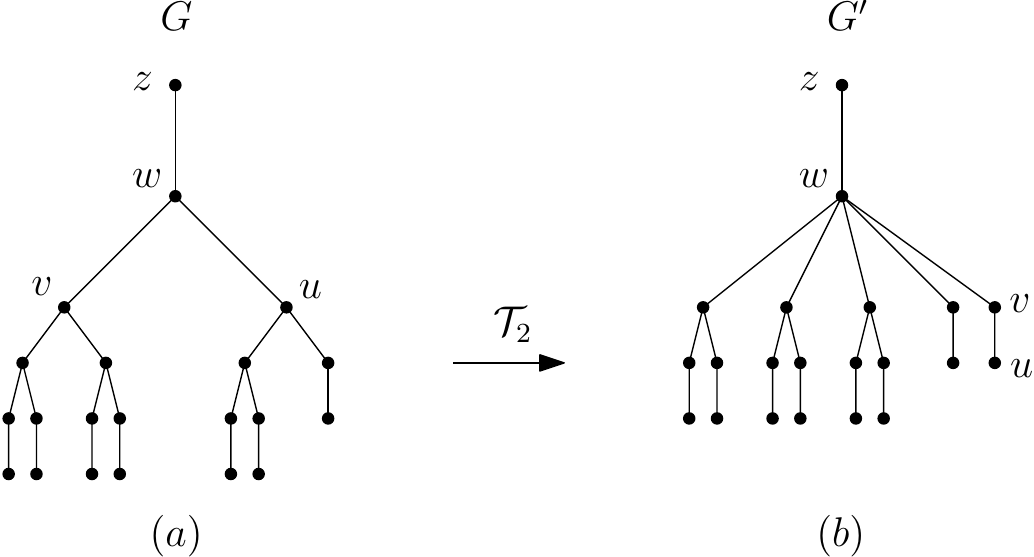}% [height=3.5cm,width=15cm]
\caption{An illustration of the transformation $\mathcal{T}_2$ from the proof of Proposition~\ref{pro-B2-20}, the case $d(w)=3$.}
%\label{Unicyclic-Max-M3}
\label{fig-B2-3}
%\vspace{-0.3cm}
\end{center}
\end{figure}
%
%In this case we apply the transformation $\mathcal{T}$ depicted in  Figure~\ref{fig-B2-3}.
After applying the transformation $\mathcal{T}_2$ depicted in  Figure~\ref{fig-B2-3}, the degree of 
the  vertex $w$ increases from $3$ to $6$, while the degrees of 
the  vertices $v$ and $u$ decreases from $3$ to $2$ and $1$, respectively. Thus, the total change of the ABC index of
$G$ is
\beq \label{change-B2-40}
&&-f(d(z),3)+f(d(z),6)-f(3,3)+f(6,2)-f(3,3)+f(2,1) \nonumber \\ 
&& +3(-f(3,3)+f(6,3)).
\eeq
By Proposition~\ref{appendix-pro-030-2}, $-f(d(z),3)+f(d(z),6)$ decreases in $d(z)$, and thus, 
(\ref{change-B2-40}) reaches its maximum of  $\approx -0.0913482$ for $d(z)=3$.

If $w$ is a root vertex of $G$, i.e., $z$ is a child of $w$, then, $z$ for the same reasons as $v$ must have degree $3$,
and thus, in this case, the change of the ABC index after applying the transformation $\mathcal{T}_2$ is smaller than $-0.091348$.

\bigskip
\noindent
{\bf Case $2$.} $d(w)= 4$.

\smallskip
\noindent
Similarly as in the previous case, we may conclude that $w$, in addition to $u$, has two more
children vertices $v_1$ and $v_2$ that by Theorems~\ref{thm-GFI-10}, \ref{thm-DS}  
and Proposition~\ref{pro-terminal-branches-10} have degrees $3$ or $4$
%Here, $w$ have at least three children vertices. By Theorems~\ref{thm-GFI-10},  \ref{thm-LG-10} and \ref{thm-DS} must have degree two or three.
(see Figure~\ref{fig-B2-4}~(a) for an illustration).
\begin{figure}[h!]
\begin{center}
%\vspace{-0.3cm}
\includegraphics[scale=0.75]{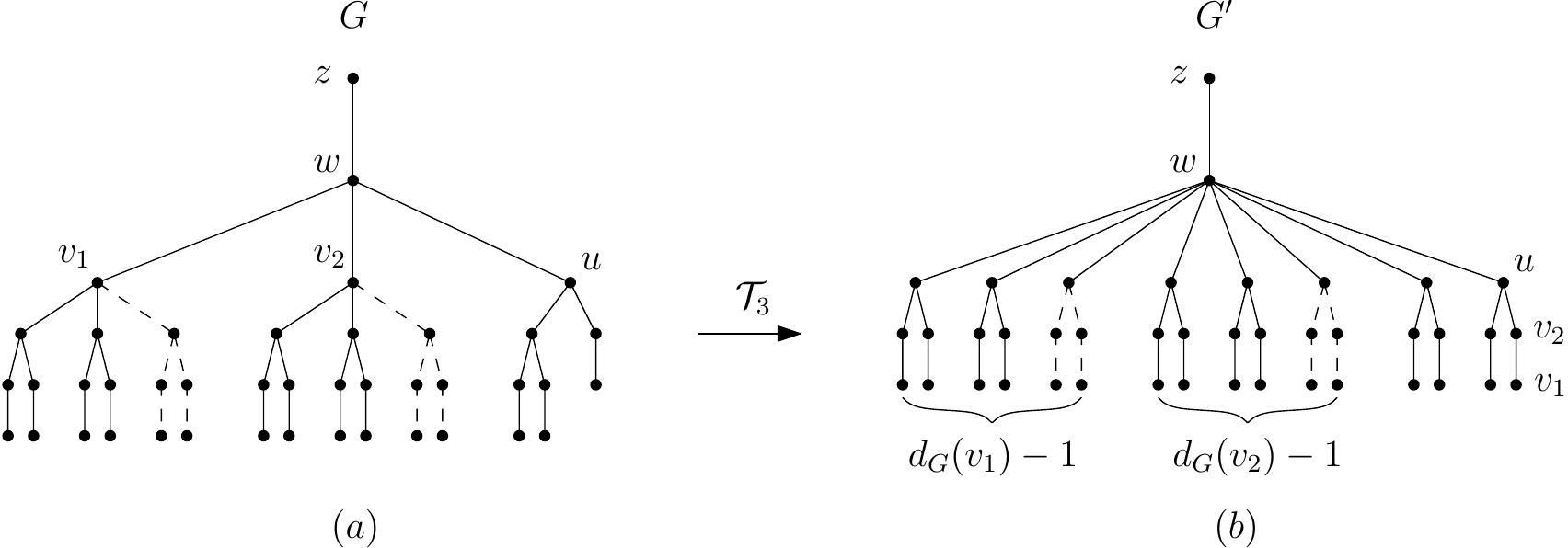}% [height=3.5cm,width=15cm]
\caption{An illustration of the transformation $\mathcal{T}_3$ from the proof of Proposition~\ref{pro-B2-20}, the case $d(w)=4$.}
%\label{Unicyclic-Max-M3}
\label{fig-B2-4}
%\vspace{-0.3cm}
\end{center}
\end{figure}
In this case we apply the transformation $\mathcal{T}_3$ depicted in  Figure~\ref{fig-B2-4}.
After applying the transformation $\mathcal{T}_3$ depicted in  Figure~\ref{fig-B2-4}, the degree of 
the  vertex $v_i$, $i=1,2$, decreases from $d(v_i)$ to $2$ and $1$, respectively.
The degree of the vertex $w$ increases from $4$ to $d(v_1)+d(v_2)+1$, while the rest of the vertices
do not change their degrees. The total change of the ABC index of
$G$ is
\beq \label{change-B2-50}
&&-f(d(z),4)+f(d(z), d(v_1)+d(v_2)+1)-f(4,d(v_1))+f(2,1)-f(4,d(v_2))+f(3,2) \nonumber \\
&&-f(4,3)+f(d(v_1)+d(v_2)+1,3) +(d(v_1)-1)(-f(d(v_1),3)+f(d(v_1)+d(v_2)+1,3))\nonumber \\
&& +(d(v_2)-1)(-f(d(v_2,3)+f(d(v_1)+d(v_2)+1,3)) \nonumber \\
&&-f(3,3) +f(d(v_1)+d(v_2)+1,3).
\eeq
(Maybe explain each term?)
Since $d(v_1)+d(v_2)+1 > 4$, by Proposition~\ref{appendix-pro-030-2}, the
expression $-f(d(z),4)+f(d(z), d(v_1)+d(v_2)+1)$ decreases in $z$, and therefore 
it is maximal for $d(z)=4$. 
%The possible values for $d(v_1)$ and $d(v_2)$ are $3$ and $4$.
Out of the four possible combinations of the values of $d(v_1)$ and $d(v_2)$
(recall that $d(v_1)$ and $d(v_2)$ can be either $3$ or $4$), 
(\ref{change-B2-50}) is maximal for $d(x_1)=d(x_2)=4$ and is  $\approx -0.186635$.

If $w$ is a root vertex of $G$, i.e., $z$ is a child of $w$, then, $z$ for the same reasons as $v_1$ and $v_2$ must have degree $3$ or $4$.
By Proposition~\ref{appendix-pro-030-2}, it follows that (\ref{change-B2-50}) reaches it maximum of $\approx -0.16395$ for $d(z)=3$.
\end{proof}

Next, we present an upper bound on the number of $B_2$-branches that may be attached to a vertex of a minimal-ABC tree.

\begin{lemma} \label{lemma-B2-30}
A vertex $w$ of a minimal-ABC tree $G$ cannot be a parent of more than eleven $B_2$-branches.
Moreover, if $w$ is a root of $G$, then it cannot be a parent of more than ten $B_2$-branches.
%Moreover, if these upper bounds are sharp only in cases when $w$ has only $B_2$-branches as its immediate children. 
%
\end{lemma} 
\begin{proof}

Let $w$ has $n_3$  children of degree at least $4$, $n_2$ children of degree $3$ and $n_1$ 
children of degree $2$, where $n_3 \geq 0$,  $n_2  > 11$ and $n_1  \geq 0$.
By Theorem~\ref{thm-DS}, it follows that a vertex of degree $3$ can have a children of degree at most $3$.
Further, by Propositions~\ref{pro-B2-10} and ~\ref{pro-B2-20}, it follows  that a vertex of degree $3$ can be a parent
only of $B_2$-branches, but only if it is a root vertex.
Let $v$ be a children of $w$ with degree $\geq 4$.
By Lemma~\ref{lemma-15}, $w$ cannot have simultaneously $B_4$-branches and $B_2$-branches as its children.
By Theorems~\ref{thm-DS} and \ref{te-no5branches-10}, $v$ cannot be a root of $B_{\geq 5}$-branch.
Thus, it follows that $v$ is a parent of  $B_k$-branches, $1 \leq k \leq 2$, or  $v$ is a root of
a $B_3$-branch. Notice that, if $v$ is a not a root of  a $B_3$-branch, then by Theorem~\ref{thm-DS}, it follows
that $n_1=0$.

We distinguish two cases regarding if $w$ is the root vertex of a minimal-ABC tree or not.
Further subcases  that depend on the children of $w$ are introduced. 

%\vspace{0.1cm}
\smallskip
\noindent
{\bf Case $1$.} $w$ is not the root vertex of a minimal-ABC tree.

\smallskip
\noindent
{\bf Subcase $1.1$.} $w$ is a parent only  of $B_k$-branches, $1 \leq k \leq 3$.

\noindent 
In this case the structure of $G$ is illustrated in Figure~\ref{fig-B2-5}.
%By Lemma~\ref{lemma-15}, $w$ cannot have simultaneously $B_4$-branches and $B_2$-branches as its children. 
%Together with Theorem~\ref{te-no5branches-10}, we can conclude
%that $w$ may have only $B_3$, $B_2$, $B_1$-branches as its children.
%
\begin{figure}[h!]
\begin{center}
%\vspace{-0.3cm}
\includegraphics[scale=0.75]{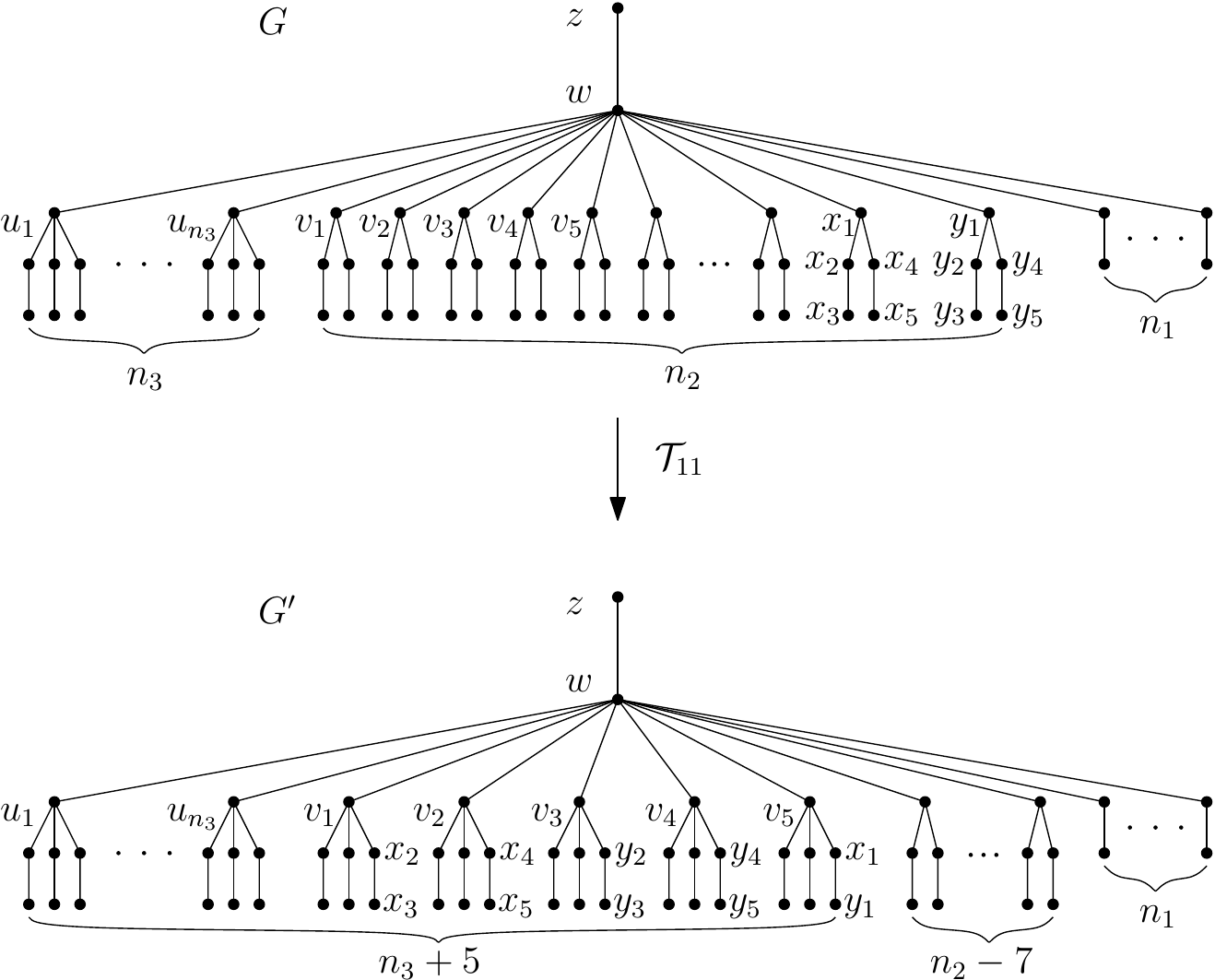}% [height=3.5cm,width=15cm]
\caption{An illustration of the transformation $\mathcal{T}_{11}$ from the proof of Lemma~\ref{lemma-B2-30}, Subcase~$1.1$.}
%\label{Unicyclic-Max-M3}
\label{fig-B2-5}
%\vspace{-0.3cm}
\end{center}
\end{figure}
%What if $w$ contains children that are not root vertices of B-branches?
%
%
%From Propositon~\ref{pro-terminal-branches-10}, it follows that if $w$ has a child of degree $>4$, then it does not have child of degree $2$.
%
After applying the transformation $\mathcal{T}_{11}$ from the same figure, the degrees of the vertices $v_i, i=1, \dots, 5$, increase by one,
the degrees of  $w$ and $y_1$ decrease by two, and the degree of $x_1$ decreases by one.
Thus, the change of the ABC index after applying the transformation $\mathcal{T}$  is
\beq \label{change-B2-60}
&& -f(d(z),d(w))+f(d(z),d(w)-2)+\sum_{i=1}^{n_3}(-f(d(u_i),d(w))+f(d(u_i),d(w)-2))  \nonumber \\
&&+5(-f(3,d(w))+f(4,d(w)-2)) +(n_2-5)(-f(3,d(w))+f(3,d(w)-2)) \nonumber \\
&&-f(d(w),3)+f(3,2)-f(d(w),3)+f(2,1).
\eeq 
By Proposition~\ref{appendix-pro-030}, $-f(d(z),d(w))+f(d(z),d(w)-2)$ 
increases in $d(z)$ and it reaches it maximal values for $d(z) \to \infty$.
By the same argument $-f(d(u_i),d(w))+f(d(u_i),d(w)-2)$ is maximal for 
$d(u_i)=4$.
The expression $-f(3,d(w))+f(3,d(w)-2)$ is positive for $d(w) > 2$. 
From $d(w)=n_3+n_2+n_1+1$, we have that $n_2-5=d(w)- n_3-n_1-6$.
Thus $(n_2-5)(-f(3,d(w))+f(3,d(w)-2))=(d(w)- n_3-n_1-6)(-f(3,d(w))+f(3,d(w)-2))$ is maximal for $n_1=0$, 
and then, $d(w)=n_3+n_2+1$.
%The (\ref{change-B2-60}) is maximal for $u \to \infty$, $n_3=n_1=0$ and $w=11$ and it is $-0.xxxx$.

\smallskip
\noindent
Since $-f(4,d(w))+f(4,d(w)-2) > -f(3,d(w))+f(3,d(w)-2)$,
the sum 
\beq \label{change-B2-65}
&&n_3(-f(4,d(w))+f(4,d(w)-2))+(n_2-5)(-f(3,d(w))+f(3,d(w)-2))  \nonumber \\
&=&(d(w)-n_2-1)(-f(4,d(w))+f(4,d(w)-2))  +(n_2-5)(-f(3,d(w))+f(3,d(w)-2))\nonumber
\eeq 
is maximal when $n_2$ is minimal, i.e., $n_2=12$.
Considering all these, it follows that (\ref{change-B2-60}) is bound from above by
\beq \label{change-B2-66}
&&\lim_{d(z) \to \infty}( -f(d(z),d(w))+f(d(z),d(w)-2))+(d(w)-13)(-f(4,d(w))+f(4,d(w)-2))  \nonumber \\
&&+5(-f(3,d(w))+f(4,d(w)-2)) +7(-f(3,d(w))+f(3,d(w)-2)) \nonumber \\
&&-f(d(w),3)+f(3,2)-f(d(w),3)+f(2,1).
\eeq 
Next, consider the following functions that are comprised by components of (\ref{change-B2-66}): 
\beq \label{change-B2-68}
g_1(d(w))&=& 4(-f(3,d(w))+f(4,d(w)-2))-f(d(w),3)+f(3,2)-f(d(w),3)+f(2,1), \nonumber
\eeq 
and
\beq \label{change-B2-69}
g_2(d(w)) &=& (d(w)-13)(-f(4,d(w))+f(4,d(w)-2))+7(-f(3,d(w))+f(3,d(w)-2)) \nonumber \\ 
&&-f(3,d(w))+f(4,d(w)-2). \nonumber
\eeq 
After simplifying, we obtain that the first derivative of $g_1(d(w))$ is
\beq \label{change-B2-70}
\frac {d \, g_1(d(w))}{d \, d(w)} &=&\frac{-2 \sqrt{3}+d(w) \left(\sqrt{3}-2 \sqrt{\frac{d(w)+1}{d(w)-2}}\right)}{(d(w)-2) d(w)^2\sqrt{1+\frac{1}{d(w)}} }. \nonumber
\eeq 
For $d(w) >2$, it is easy to verify that the nominator of the last expression is negative, while its denominator is positive.
Thus $d \, g_1(d(w))/d \, d(w)$ is negative, from which follows that $g(d(w))$ is decreasing function in $d(w)$.
The first derivative of $g_2(d(w))$, after a simplification is
\beq \label{change-B2-80}
\frac {d \, g_2(d(w))}{d \, d(w)} &=&
\frac{1}{6} \left(-\frac{7 \sqrt{6-3 d(w)}}{\sqrt{1-d(w)} (-2+d(w))^2}+\frac{3 \sqrt{-d(w)}}{(2-d(w))^{3/2}}+3 \sqrt{\frac{d(w)}{-2+d(w)}} \nonumber \right.\\
&& \;  +\frac{3}{d(w) \sqrt{\frac{2+d(w)}{d(w)}}} -3 \left. \sqrt{\frac{2+d(w)}{d(w)}}+\frac{-\frac{39}{\sqrt{-2-d(w)}}+\frac{11 \sqrt{3}}{\sqrt{-1-d(w)}}-\frac{27 d(w)}{(2-d(w))^{3/2}}}{(-d(w))^{3/2}}\right), \nonumber
\eeq 
and it has no real roots, which means that it is either positive or negative. Since for $d(w)=13$, $d \, g_2(d(w))/d \, d(w) =-0.00661647$, 
it follows that $g_2(d(w))$ decreases in $d(w)$, too.

By Proposition~\ref{appendix-pro-030}, the expression $-f(d(z),d(w))+f(d(z),d(w)-2)$ also decreases in $d(w)$.
We can conclude that (\ref{change-B2-66}), and therefore also (\ref{change-B2-60}) decrease in $d(w)$, and are
maximal when $d(w)$ is minimal, i.e., $d(w)=13$ and their upper bound is
\beq \label{change-B2-67}
&&\lim_{d(z) \to \infty}( -f(d(z),13)+f(d(z),11)) +5(-f(3,13)+f(4,11)) +7(-f(3,13)+f(3,11))  \nonumber \\
&&-f(13,3)+f(3,2)-f(13,3)+f(2,1) < -0.0107055. \nonumber
\eeq 

Observe that the above upper bound of $-0.0107055$ is obtained when $w$ does not have $B_3$ and $B_1$-branches ($n_3=n1=0$) as
it immediate children.
If $w$ has in addition one $B_3$, i.e., $n_3=1$ and $n_1=0$, then  (\ref{change-B2-60}) is negative for $n_2 \geq 11$, or with other words, 
$w$ can have at most $10$ $B_2$-branches. 
If $w$ $n_3=2$ and $n_1=0$, then  (\ref{change-B2-60}) is negative for $n_2 \geq 10$, i.e., 
$w$ can have at most $9$ $B_2$-branches.

If $w$ has in addition one $B_1$, i.e., $n_3=0$ and $n_1=1$, then  (\ref{change-B2-60}) is negative for $n_2 \geq 10$, or with other words, 
$w$ can have at most $9$ $B_2$-branches. If $n_3=0$ and $n_1=2$, then  (\ref{change-B2-60}) is negative for $n_2 \geq 9$,
i.e., $w$ can have at most $8$ $B_2$-branches.

\smallskip
\noindent
{\bf Subcase $1.2$.} $w$ is a parent of one or more vertices that are not roots of $B_k$-branches, $1 \leq k \leq 4$.

\noindent
This case is similar to the previous one, with the difference that
$w$ may have children with degree larger than $4$.
Denote by $x$ a child of $w$ with  $d(x) \geq 4$.
Since $w$ has $B_2$-branches as children, by Theorem~\ref{thm-DS}, it follows that
the children of $w$ are either $B_2$ or $B_1$-branches.
Due to Proposition~\ref{pro-B1-10-00}, we may assume that $d(x) \leq 8$.
If $d(x) \geq 9$ we can apply the transformation from Lemma~\ref{lemma-B1-10-2}, obtaining $d(x) \leq 8$.
%Let $x$ be a child of $w$  that is not a root of $B_k$-branches, $1 \leq k \leq 4$.
If $w$ has a child $y$ of degree $3$, where $y$  is not a root of $B_k$-branches, $1 \leq k \leq 2$,
then, by Propositions~\ref{pro-B2-10} and ~\ref{pro-B2-20} and Theorem~\ref{thm-DS}, it follows that 
$y$ must be a root of a $B_2$-branch. Thus, $4 \leq d(x) \leq 8$.
Also, in this case $w$ does not have $B_1$-branches as children. 
Otherwise, if $w$ does have $B_1$-branches as children, then by Theorem~\ref{thm-DS} a child of $w$ cannot have
a child of degree larger than $2$, which is a contradiction to the main assumption of this subcase.

This subcase and the corresponding transformation that we apply are illustrated in Figure~\ref{fig-B2-5-1}.
\begin{figure}[h!]
\begin{center}
%\vspace{-0.3cm}
\includegraphics[scale=0.75]{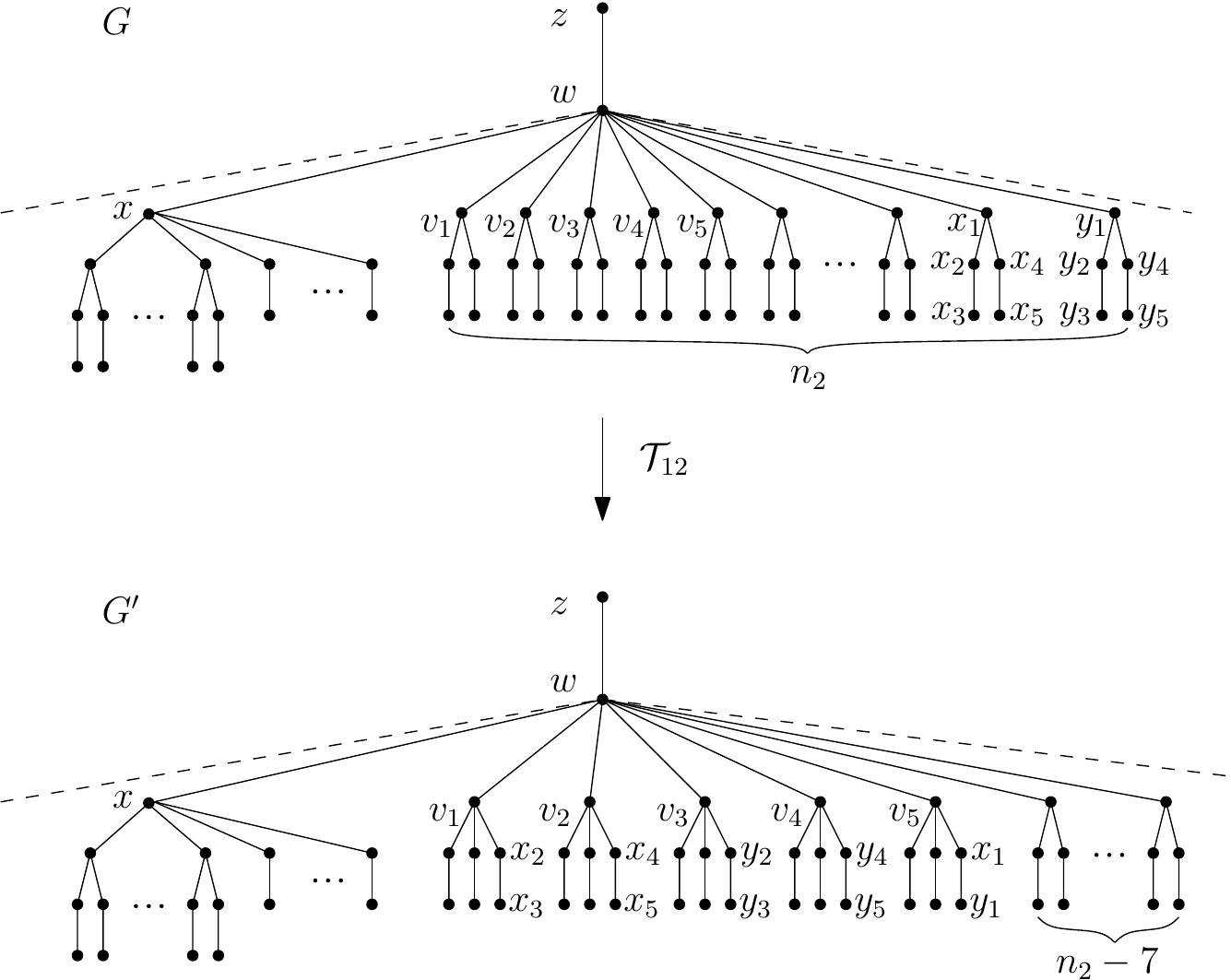}% [height=3.5cm,width=15cm]
\caption{An illustration of the transformation $\mathcal{T}_{12}$ from the proof of Lemma~\ref{lemma-B2-30}, Subcase~$1.2$.}
%\label{Unicyclic-Max-M3}
\label{fig-B2-5-1}
%\vspace{-0.3cm}
\end{center}
\end{figure}
%What if $w$ contains children that are not root vertices of B-branches?
%
Similarly as in previous case, here we obtain that the change of  ABC($G$) is
\beq \label{change-B2-60-12}
&& -f(d(z),d(w))+f(d(z),d(w)-2)+\sum_{i=1}^{d(w)-n_2-1}(-f(d(u_i),d(w))+f(d(u_i),d(w)-2))  \nonumber \\
&&+5(-f(3,d(w))+f(4,d(w)-2)) +(n_2-5)(-f(3,d(w))+f(3,d(w)-2)) \nonumber \\
&&-f(d(w),3)+f(3,2)-f(d(w),3)+f(2,1).
\eeq 
By Proposition~\ref{appendix-pro-030} $-f(d(u_i),d(w))+f(d(u_i),d(w)-2)$ 
increases with $d(u_i)$.
Bearing in mind that $4 \leq d(u_i) \leq 8$, we obtain that (\ref{change-B2-60-12})
is maximal for $d(u_i)=8$, $i=1, \dots, d(w)-n_2-1$.
Applying same arguments as in Subcase~$1.1.$, we also obtain that
that (\ref{change-B2-60-12}) is maximal when $d(w)$ is minimal, i.e, $d(w)=13$ ($n_2=12$), 
and when $d(z) \to \infty$.
Thus, 
\beq \label{change-B2-67-12}
&&\lim_{d(z) \to \infty}( -f(d(z),13)+f(d(z),11)) +5(-f(3,13)+f(4,11)) +7(-f(3,13)+f(3,11))  \nonumber \\
&&-f(13,3)+f(3,2)-f(13,3)+f(2,1) \approx -0.0107055 \nonumber
\eeq 
is an upper bond on (\ref{change-B2-60-12}).

\smallskip
\noindent
{\bf Case $2$.} $w$ is the root vertex of a minimal-ABC tree.

\smallskip
\noindent
{\bf Subcase $2.1$.} $w$ is a parent only of vertices that are roots of $B_k$-branches, $1 \leq k \leq 3$.

\noindent
If $w$ is a root vertex of $G$, then the change of the ABC-index after applying the same 
transformation from Figure~\ref{fig-B2-5} is
\beq \label{change-B2-100}
&&  n_3(-f(4,d(w))+f(4,d(w)-2)) +5(-f(3,d(w))+f(4,d(w)-2)) \nonumber \\
&&+(n_2-5)(-f(3,d(w))+f(3,d(w)-2)) -f(d(w),3)+f(3,2)-f(d(w),3)+f(2,1).\nonumber \\
\eeq 
where $d(w)=n_1+n_2+n_3$, $n_3 \geq 0$,  $n_2  > 10$ and $n_1  \geq 0$.
Almost identical analysis as in the Subcase $1.1.$, shows that (\ref{change-B2-100}) decreases in $d(w)$ and is maximal 
for $n_1=0$ and $n_3$ as large as possible, which in this case is $d(w)-n_2=d(w)-10$.
Thus, (\ref{change-B2-100}) is bounded from above by
\beq \label{change-B2-67}
&&(d(w)-11)(-f(4,11)+f(4,9))  +5(-f(3,11)+f(4,9)) +7(-f(3,11)+f(3,9)) \nonumber \\
&&-f(11,3)+f(3,2)-f(11,3)+f(2,1) < -0.00974369. \nonumber
\eeq 
 
If $w$ has in addition one $B_3$, i.e., $n_3=1$ and $n_1=0$, then  (\ref{change-B2-100}) is negative for $n_2 \geq 10$, or with other words, 
$w$ can have at most $9$ $B_2$-branches. 
If $w$ $n_3=2$ and $n_1=0$, then  (\ref{change-B2-100}) is negative for $n_2 \geq 9$, i.e., 
$w$ can have at most $8$ $B_2$-branches.

If $w$ has in addition one $B_1$, i.e., $n_3=0$ and $n_1=1$, then  (\ref{change-B2-100}) is negative for $n_2 \geq 9$, or with other words, 
$w$ can have at most $8$ $B_2$-branches. If $n_3=0$ and $n_1=2$, then  (\ref{change-B2-60}) is negative for $n_2 \geq 8$,
i.e., $w$ can have at most $7$ $B_2$-branches.

\smallskip
\noindent
{\bf Subcase $2.2$.} $w$ is a parent of one or more vertices that are not roots of $B_k$-branches, $1 \leq k \leq 3$.
 
\noindent
Here we apply the same transformation as in Figure~\ref{fig-B2-5-1}. 
Since  $w$ is a root vertex of $G$, the change of the ABC-index now is
\beq \label{change-B2-100-22}
&&  (d(w)-n_2)(-f(13,d(w))+f(13,d(w)-2)) +5(-f(3,d(w))+f(4,d(w)-2)) \nonumber \\
&&+(n_2-5)(-f(3,d(w))+f(3,d(w)-2)) -f(d(w),3)+f(3,2)-f(d(w),3)+f(2,1).\nonumber \\
\eeq 
where  $n_2  > 10$.
Almost identical analysis as in the Subcase $1.2.$, shows that (\ref{change-B2-100-22}) decreases in $d(w)$,
and is maximal for $d(w)=11$ and when the children of $w$ that are not roots of $B_k$-branches, $1 \leq k \leq 4$,
have maximal degrees, which after Subcase $2.1$ do not exceed $11$. 
Thus, (\ref{change-B2-100-22}) is bounded from above by
\beq \label{change-B2-67-22}
&&(d(w)-11)(-f(11,11)+f(11,9))  +5(-f(3,11)+f(4,9)) +7(-f(3,11)+f(3,9)) \nonumber \\
&&-f(11,3)+f(3,2)-f(11,3)+f(2,1) < -0.00974369. \nonumber
\eeq 
 \end{proof}

\noindent 
The following  proposition follows from Lemma~\ref{lemma-B2-30}  and will be used in the
proof of Theorem~\ref{te-no2branches-10a}.
 
 \begin{pro} \label{pro-B2-30}
The trees depicted in Figure~\ref{fig-B2-5-1-co}  are not minimal-ABC trees. 
 \end{pro}
 \begin{proof}
 Trees depicted in Figure~\ref{fig-B2-5-1-co} have the structure that belongs
 to Subcase $2.1$ of Lemma~\ref{lemma-B2-30} (actually, it is a special case of  Subcase $2.1$, since here we do not have $B_1$-brances). 
 The change of the ABC-index after applying the transformation from Figure~\ref{fig-B2-5-1-co} is
 given in (\ref{change-B2-100}).
\begin{figure}[h!]
\begin{center}
%\vspace{-0.3cm}
\includegraphics[scale=0.75]{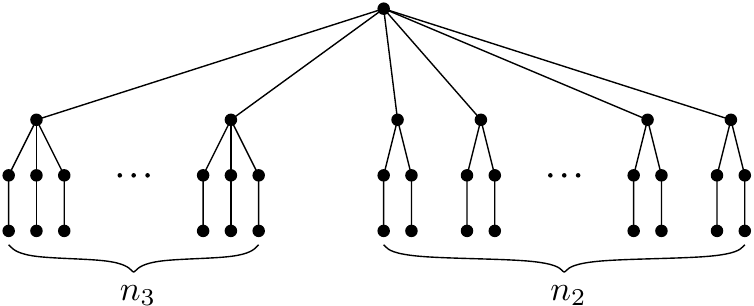}% [height=3.5cm,width=15cm]
\caption{Trees with parameters $n_2=10$ and $n_3 > 0$, $n_2=9$ and $n_3 > 1$,  $n_2=8$ and $n_3 > 3$, $n_2=7$ and $n_3 > 4$, that are
not minimal-ABC trees.}
%\label{Unicyclic-Max-M3}
\label{fig-B2-5-1-co}
%\vspace{-0.3cm}
\end{center}
\end{figure}
%What if $w$ contains children that are not root vertices of B-branches?
%
 Here it holds that  $d(w)=n_2+n_3$, $n_3 \geq 0$,  $7 \leq n_2  \leq 10$.
 By substituting particular values for $n_2$ and $n_3$, one can obtain that
 (\ref{change-B2-100}) is negative for
 $n_2=10$ and $n_3 > 0$; $n_2=9$ and $n_3 > 1$;  $n_2=8$ and $n_3 > 3$; $n_2=7$ and $n_3 > 4$.
 
 Observe that trees in Figure~\ref{fig-B2-5-1-co} have less than $66$ vertices.
In~\cite{d-ectmabci-2013} all minimal-ABC trees with up to $300$ were computed,
and no of them are in Figure~\ref{fig-B2-5-1-co}.
 \end{proof}
 
 \noindent
 The next result is a specialized version of Lemma~\ref{lemma-B2-30}.
 
\begin{lemma} \label{lemma-B2-10}
Let $w$ be a vertex of a minimal-ABC tree $G$ different than the root of $G$. 
%Let  $w$ have only $B_2$-branches as its children.
%Then, $w$ is a parent of at most six $B_2$-branches.
If  $w$ has only $B_2$-branches as its children, then their number is at most $6$.
\end{lemma} 
\begin{proof}
By Lemma~\ref{lemma-B2-30}, $d(w) \leq 13$.
Since $w$ is a parent of $B_2$-branches, 
by Theorem~\ref{thm-DS}  it follows that $z$ can not be a parent  of a $B_1$-branch.
Assume that $w$  has more than $6$ $B_2$-branches as its children.
We distinguish few cases with respect to $d(w)$ and the degrees of the children vertices
of $z$.

\bigskip
\noindent
{\bf Case $1$.} $d(w)=8$.

\noindent
In this case we apply the transformation $\mathcal{T}_1$
illustrated in Figure~\ref{fig-B2-5-5-1}. After this transformation the
degree of the vertex $z$ increases by $d(w)-4$, the degrees of five children vertices 
of $w$ increase from $3$ to $4$, the degree of the vertex $w$ decreases to $1$, 
while two children vertices of $w$  decrease their degrees from $3$ to $2$ and $1$, respectively.
\begin{figure}[h!]
\begin{center}
%\vspace{-0.3cm}
\includegraphics[scale=0.75]{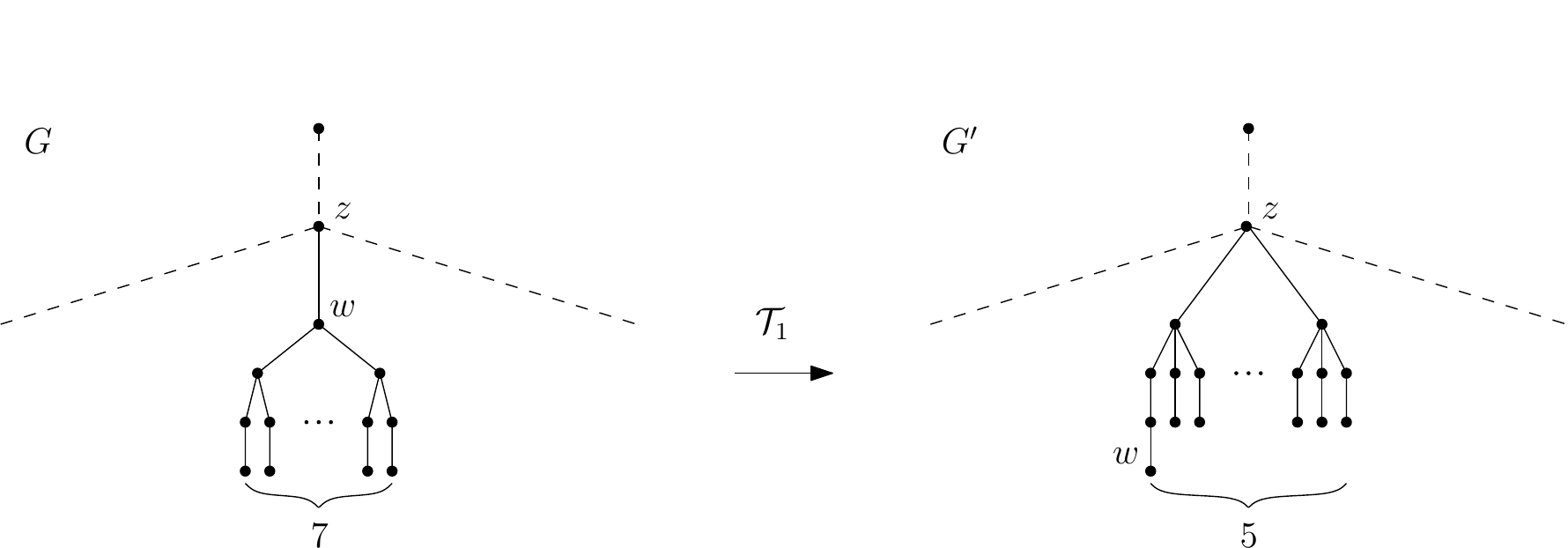}% [height=3.5cm,width=15cm]
\caption{An illustration of the transformation $\mathcal{T}_1$ from the proof of Lemma~\ref{lemma-B2-10}, Case $1$.}
%\label{Unicyclic-Max-M3}
\label{fig-B2-5-5-1}
%\vspace{-0.3cm}
\end{center}
\end{figure}
After applying the transformation $\mathcal{T}_1$, the change of the ABC index is
bounded from above by
\beq \label{lemma-change-B2-10-10}
g(d(z), d(w)) &=&  -f(d(z),d(w))+f(2,1))  +5(-f(3,d(w))+f(4,d(w)+d(z)-4))  \nonumber \\
&&-f(3,d(w))+f(4,2) -f(3,d(w))+f(2,1) \nonumber \\
&&+\sum_{i=1}^{d(z) - 2}( -f(x_i,d(w))+f(x_i,d(w)+d(z)-4))).   \nonumber
\eeq 
By Proposition~\ref{appendix-pro-030-2}  $-f(x_i,d(w))+f(x_i,d(w)+d(z)-4)$ decreases in $x_i$,
i.e., it is maximal for $x_i=3$. This together with $d(w)=8$ gives us the following upper bound 
on $g(d(z)$, $d(w))$:
\beq \label{lemma-change-B2-10-20}
g(d(z), 8) &=&  -f(d(z),8)+f(2,1))  +5(-f(3,8)+f(4,d(z)+4))  \nonumber \\
&&-f(3,8)+f(4,2) -f(3,8)+f(2,1) \nonumber \\
&&+(d(z) - 2)( -f(3,8)+f(3,d(z)+4))). \nonumber
\eeq 
The function $g(d(z), 8)$ does not have local extremal points,
and it is a decreasing function in $d(z)$.
Thus,  $g(d(z), 8)$ is maximal when $d(z)$ is smallest possible, i.e., when $d(z)=d(w)=8$.
Then $g(8, 8)=-0.00136859$, and therefore, the 
change of the ABC index after applying the transformation $\mathcal{T}_1$ is negative.
This is a contradiction to the assumption that $G$ is a tree with minimal ABC index. 
%$w$ has more than $6$ $B_2$-branches as children.

\noindent
If $z$ is the root vertex of G then the change of the ABC-index is bounded by
\beq \label{lemma-change-B2-10-20-root}
g_r(d(z), 8) &=&  -f(d(z),8)+f(2,1))  +5(-f(3,8)+f(4,d(z)+4))  \nonumber \\
&&-f(3,8)+f(4,2) -f(3,8)+f(2,1) \nonumber \\
&&+(d(z) - 1)( -f(3,8)+f(3,d(z)+4))), \nonumber
\eeq 
which is bounded from above by the negative function $g(d(z), 8)$.

\bigskip
\noindent
{\bf Case $2$.} $9 \leq d(w) \leq 13$.

\smallskip
\noindent
{\bf Subcase $2.1$.} $z$ has at least three  children of degree $3$.

\noindent
In this case we apply the transformation $\mathcal{T}_{21}$
illustrated in Figure~\ref{fig-B2-5-5-21a}. 
\begin{figure}[h!]
\begin{center}
%\vspace{-0.3cm}
\includegraphics[scale=0.75]{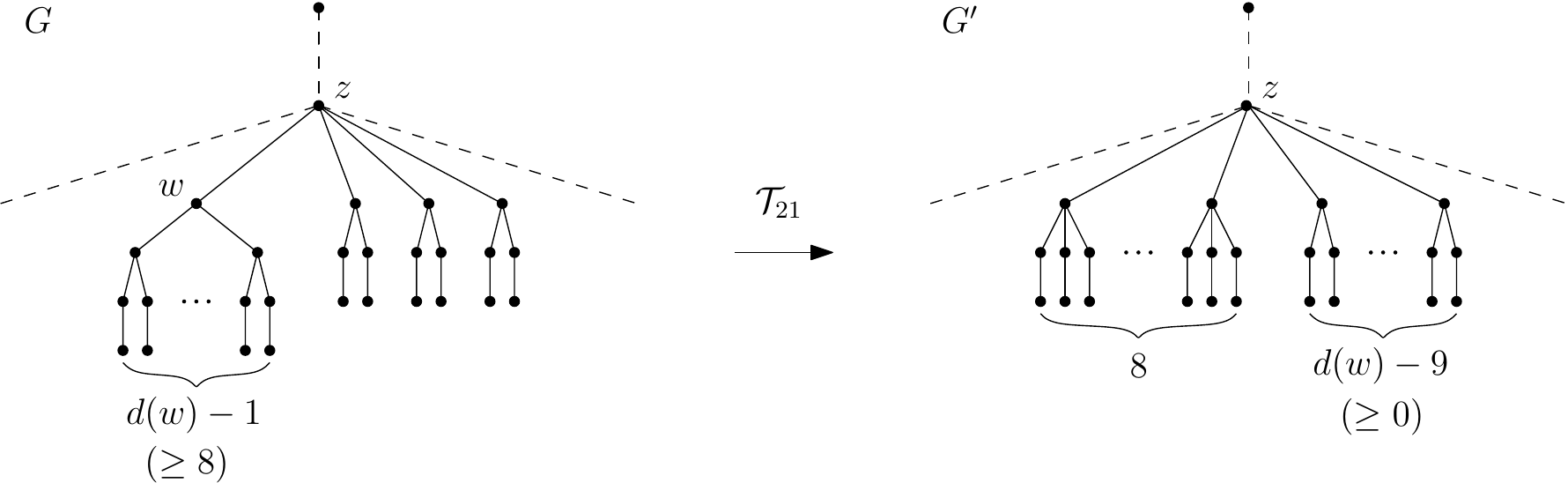}% [height=3.5cm,width=15cm]
\caption{An illustration of the transformation $\mathcal{T}_{21}$ 
from the proof of Lemma~\ref{lemma-B2-10}, Subcase $2.1$.}
%\label{Unicyclic-Max-M3}
\label{fig-B2-5-5-21a}
%\vspace{-0.3cm}
\end{center}
\end{figure}
After this transformation the
degree of the vertex $z$ increases by $d(w)-5$, the degrees of five children vertices 
of $w$ and three  children vertices of $z$ increase from $3$ to $4$, the degree of the vertex $w$ decreases to $1$, 
while three children vertices of $w$ decrease their degrees from $3$ to $2$, $2$ and $1$, respectively.
The rest of the vertices do not change their degrees.
The change of the ABC index after applying $\mathcal{T}_{21}$ is at most
\beq \label{lemma-change-B2-10-30}
 &&-f(d(z),d(w))+f(2,1)) -f(3,d(w))+f(2,1)) +2(-f(3,d(w))+f(4,2))  \nonumber \\
&& +5(-f(3,d(w))+f(4,d(w)+d(z)-5)) + 3(-f(3,d(z))+f(4,d(w)+d(z)-5)) \nonumber \\
&&+ (d(w)-9)(-f(3,d(w))+f(3,d(w)+d(z)-5)) \nonumber \\
&&+\sum_{i=1}^{d(z) - 5}( -f(x_i,d(w))+f(x_i,d(w)+d(z)-5))). 
\eeq 
By Proposition~\ref{appendix-pro-030-2}  $-f(x_i,d(w))+f(x_i,d(w)+d(z)-4)$ decreases in $x_i$,
i.e., it is maximal for $x_i=3$. Thus,
\beq \label{lemma-change-B2-10-40}
g(d(z), d(w)) &=&  -f(d(z),d(w))+f(2,1)) -f(3,d(w))+f(2,1)) +2(-f(3,d(w))+f(4,2))  \nonumber \\
&& +5(-f(3,d(w))+f(4,d(w)+d(z)-5))  \nonumber \\
&& + 3(-f(3,d(z))+f(4,d(w)+d(z)-5)) \nonumber \\
&& + (d(w)-9)(-f(3,d(w))+f(3,d(w)+d(z)-5)) \nonumber \\
&&+(d(z) - 5)( -f(3,d(w))+f(3,d(w)+d(z)-5))),   \nonumber
\eeq
is an upper bound on (\ref{lemma-change-B2-10-30}).
It can be verified that
$g(d(z), 9) \leq g(9, 9) \approx -0.0514586$,
$g(d(z), 10) \leq g(10, 10) \approx -0.0538142$,
$g(d(z), 11) \leq g(11, 11) \approx -0.0541005$,
$g(d(z), 12) \leq g(12, 12) \approx -0.0531217$, and
$g(d(z), 13) \leq g(12, 13) \approx -0.0510972$.
Thus, 
the change of the ABC index after applying the transformation $\mathcal{T}_{21}$ is negative.

\noindent
If $z$ is root vertex of G then the change of the ABC-index is bounded by
\beq \label{lemma-change-B2-10-40-r}
g_r(d(z), d(w)) &=&  -f(d(z),d(w))+f(2,1)) -f(3,d(w))+f(2,1)) +2(-f(3,d(w))+f(4,2))  \nonumber \\
&& +5(-f(3,d(w))+f(4,d(w)+d(z)-5))  \nonumber \\
&& + 3(-f(3,d(z))+f(4,d(w)+d(z)-5)) \nonumber \\
&& + (d(w)-9)(-f(3,d(w))+f(3,d(w)+d(z)-5)) \nonumber \\
&&+(d(z) - 4)( -f(3,d(w))+f(3,d(w)+d(z)-5))),   \nonumber
\eeq
which is bounded from above by the negative function $g(d(z), d(w))$.

\smallskip
\noindent
{\bf Subcase $2.2$.} $z$ has at most two  children of degree $3$.

\noindent
We distinguish further two subcases with respect to $d(w)$.

\smallskip
\noindent
{\bf Subcase $2.2.1$.} $d(w)=9, 10, 11$.

\noindent
In this subcase we apply the transformation $\mathcal{T}_{221}$
illustrated in Figure~\ref{fig-B2-5-5-221}. 
\begin{figure}[h!]
\begin{center}
%\vspace{-0.3cm}
\includegraphics[scale=0.75]{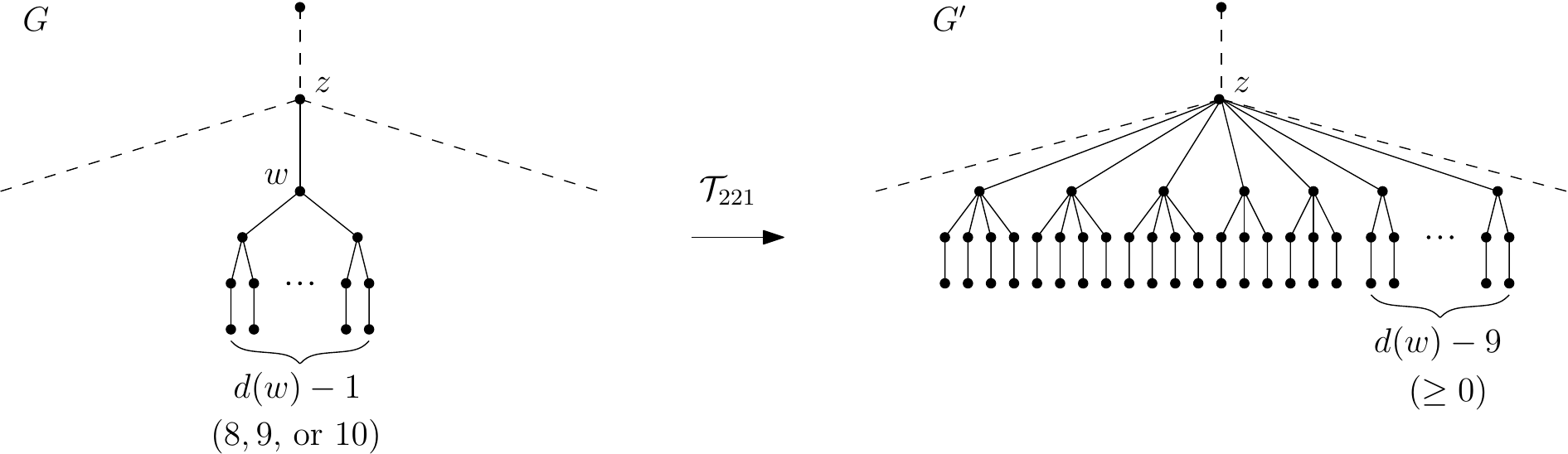}% [height=3.5cm,width=15cm]
\caption{An illustration of the transformation $\mathcal{T}_{221}$ from the proof of Lemma~\ref{lemma-B2-10}, Subcase $2.2.1$.}
%\label{Unicyclic-Max-M3}
\label{fig-B2-5-5-221}
%\vspace{-0.3cm}
\end{center}
\end{figure}
After the transformation $\mathcal{T}_{221}$ the
degree of the vertex $z$ increases by $d(w)-5$, the degrees of three children vertices 
of $w$ increase from $3$ to $5$, the degrees of two children vertices 
of $w$ increase from $3$ to $4$, the degree of the vertex $w$ decreases to $1$, 
while three children vertices of $w$,  decrease their degrees from $3$ to $2$, $2$ and $1$, respectively.
The rest of the vertices do not change their degrees.
The change of the ABC index after applying $\mathcal{T}_{221}$ is at most
\beq \label{lemma-change-B2-10-50}
 &&-f(d(z),d(w))+f(2,1)) -f(3,d(w))+f(2,1)) +2(-f(3,d(w))+f(4,2))  \nonumber \\
&& +3(-f(3,d(w))+f(5,d(w)+d(z)-5)) +2(-f(3,d(w))+f(4,d(w)+d(z)-5))   \nonumber \\
&& + (d(w)-9)(-f(3,d(w))+f(3,d(w)+d(z)-5)) \nonumber \\
&&+\sum_{i=1}^{d(z) - 2}( -f(x_i,d(w))+f(x_i,d(w)+d(z)-5))). 
\eeq 
Due to the same argument as in the previous cases,
$-f(x_i,d(w))+f(x_i,d(w)+d(z)-5))$ is maximal when $x_i=3$ is minimal.
Since $z$ may have at most two children of degree $2$ and 
$-f(x_i,d(w))+f(x_i,d(w)+d(z)-5))$ is strictly negative, then
by setting $x_i=4$, and considering only $d(z) - 4$ children of
$z$ that they have degree at least $4$ we obtain that
\beq \label{lemma-change-B2-10-60}
g(d(z), d(w)) &=&-f(d(z),d(w))+f(2,1)) -f(3,d(w))+f(2,1)) +2(-f(3,d(w))+f(4,2))  \nonumber \\
&& +3(-f(3,d(w))+f(5,d(w)+d(z)-5))  \nonumber \\
&& +2(-f(3,d(w))+f(4,d(w)+d(z)-5))  \nonumber \\
&& + (d(w)-9)(-f(3,d(w))+f(3,d(w)+d(z)-5)) \nonumber \\
&&+(d(z) - 4)( -f(4,d(w))+f(4,d(w)+d(z)-5))) \nonumber
\eeq 
is an upper bound on (\ref{lemma-change-B2-10-50}).
It can be verified that $g(d(z), d(w))$ is negative function
for $d(w)=9, 10$, and  $11$, and 
$g(d(z), 9) \leq g(9, 9) \approx -0.000496363$,
$g(d(z), 10) \leq g(10, 10) \approx -0.00763911$, and
$g(d(z), 11) \leq \lim_{d(z) \to \infty } g(d(z), 11) \approx -0.00696979$.
Thus,  also in this case, 
the change of the ABC index after applying the transformation $\mathcal{T}_{221}$ is negative.

\noindent
If $z$ is root vertex of G then the change of the ABC-index is bounded by
\beq \label{lemma-change-B2-10-60-r}
g_r(d(z), d(w)) &=&-f(d(z),d(w))+f(2,1)) -f(3,d(w))+f(2,1)) +2(-f(3,d(w))+f(4,2))  \nonumber \\
&& +3(-f(3,d(w))+f(5,d(w)+d(z)-5))  \nonumber \\
&& +2(-f(3,d(w))+f(4,d(w)+d(z)-5))  \nonumber \\
&& + (d(w)-9)(-f(3,d(w))+f(3,d(w)+d(z)-5)) \nonumber \\
&&+(d(z) - 3)( -f(4,d(w))+f(4,d(w)+d(z)-5))) \nonumber
\eeq 
which is bounded from above by the negative function $g(d(z), d(w))$.

\smallskip
\noindent
{\bf Subcase $2.2.2$.} $d(w)=12, 13$.

%I think that for $d(w)=12, 13$ in this case the transformation $\mathcal{T}_{1}$ (Case $1$) works. CHECK IT!!!

\noindent
Here we apply the transformation $\mathcal{T}_{222}$
illustrated in Figure~\ref{fig-B2-5-5-222}.
\begin{figure}[h!]
\begin{center}
%\vspace{-0.3cm}
\includegraphics[scale=0.75]{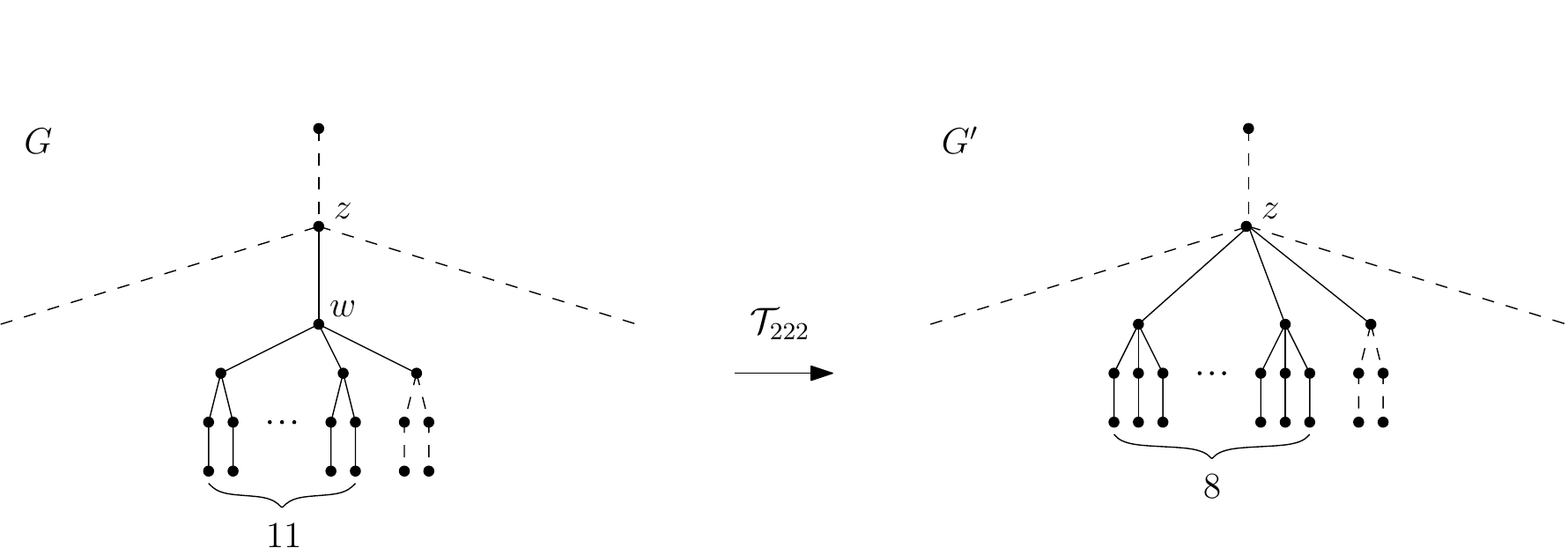}% [height=3.5cm,width=15cm]
\caption{An illustration of the proof of Lemma~\ref{lemma-B2-10}, Subcase $2.2.2$.}
%\label{Unicyclic-Max-M3}
\label{fig-B2-5-5-222}
%\vspace{-0.3cm}
\end{center}
\end{figure}
After the transformation $\mathcal{T}_{222}$ the
degree of the vertex $z$ increases by $d(w)-5$, the degrees of eight children vertices 
of $w$ increase from $3$ to $4$, the degree of the vertex $w$ decreases to $1$, 
while three children vertices of $w$,  decrease their degrees from $3$ to $2$, $2$ and $1$, respectively.
The rest of the vertices do not change their degrees.
The change of the ABC index after applying $\mathcal{T}_{222}$ is at most
\beq \label{lemma-change-B2-10-70}
 &&-f(d(z),d(w))+f(2,1)) -f(3,d(w))+f(2,1)) +2(-f(3,d(w))+f(4,2))  \nonumber \\
&& +8(-f(3,d(w))+f(4,d(w)+d(z)-5))   \nonumber \\
&& + (d(w)-12)(-f(3,d(w))+f(3,d(w)+d(z)-5)) \nonumber \\
&&+\sum_{i=1}^{d(z) - 2}( -f(x_i,d(w))+f(x_i,d(w)+d(z)-5))). 
\eeq 
Applying the same argument as in the previous case,
we obtain an upper bound on  (\ref{lemma-change-B2-10-70})
by considering only $d(z) - 4$ children of
$z$ that they have degree at least $4$: 
\beq \label{lemma-change-B2-10-80}
 g(d(z), d(w)) &=&-f(d(z),d(w))+f(2,1)) -f(3,d(w))+f(2,1)) +2(-f(3,d(w))+f(4,2))  \nonumber \\
&& +8(-f(3,d(w))+f(4,d(w)+d(z)-5))   \nonumber \\
&& + (d(w)-12)(-f(3,d(w))+f(3,d(w)+d(z)-5)) \nonumber \\
&&+(d(z) - 4)( -f(4,d(w))+f(4,d(w)+d(z)-5))). \nonumber
\eeq 
It can be verified that $g(d(z), d(w))$ is negative function
for $d(w)=12$ and $13$, and  
$g(d(z), 12) \leq \lim_{d(z) \to \infty } g(d(z), 12) \approx -0.0704253$,
and
$g(d(z), 13) \leq \lim_{d(z) \to \infty } g(d(z), 13) \approx -0.061309$.
Thus,  also in this case, 
the change of the ABC index after applying the transformation $\mathcal{T}_{222}$ is negative.

\noindent
If $z$ is root vertex of G then the change of the ABC-index is bounded by
\beq \label{lemma-change-B2-10-80-r}
 g(d(z), d(w)) &=&-f(d(z),d(w))+f(2,1)) -f(3,d(w))+f(2,1)) +2(-f(3,d(w))+f(4,2))  \nonumber \\
&& +8(-f(3,d(w))+f(4,d(w)+d(z)-5))   \nonumber \\
&& + (d(w)-12)(-f(3,d(w))+f(3,d(w)+d(z)-5)) \nonumber \\
&&+(d(z) - 3)( -f(4,d(w))+f(4,d(w)+d(z)-5))). \nonumber
\eeq which is bounded from above by the negative function $g(d(z), d(w))$.

\noindent
This concludes the proof of the lemma.
\end{proof}

\noindent
The following proposition will be used in the proof of Lemma~\ref{lemma-B2-20}, 
which is an improvement of Lemma~\ref{lemma-B2-10}.

\begin{pro} \label{pro-Tk-B1-10}
Let $T$ be a proper $T_k$-branch that contains more than $6$ $B_2$-branches
Then, $T$ cannot be a proper subtree of a minimal-ABC tree or a minimal-ABC tree itself.
%A minimal-ABC tree does not contain a proper $T_k$-branch with more than $6$ $B_2$-branches.
%Moreover, a minimal-ABC tree cannot be a proper $T_k$-tree itself if $k \geq 12$.
\end{pro}
\begin{proof}
%It follows from  Lemma~\ref{lemma-10}($a$) and Theorems~\ref{te-no5branches-10}, 
%~\ref{thm-330}, and~\ref{thm-350}.
%Denote by $T$ a proper $T_k$-branch that is a subtree of a tree with minimal-ABC index $G$.
Let $u$ be the root vertex of  $T$.
From Theorems~\ref{thm-DS},~\ref{te-no5branches-10}, and Lemma~\ref{lemma-15}($a$), it follows that $T$ does not contain
$B_k$-branches, $k \geq 4$.
Assume that the number of  $B_1$-branches  contained in $T$ is $k_1 >0$,  the number of $B_2$ is $k_2 > 6$ and $B_3$-branches is $k_3 \geq 0$.
It holds that $k_1+k_2+k_3=k$.
Perform the transformation $\mathcal{T}$ depicted in Figure~\ref{fig-B1-1b}.
\begin{figure}[h]
\begin{center}
%\vspace{-0.3cm}
\includegraphics[scale=0.750]{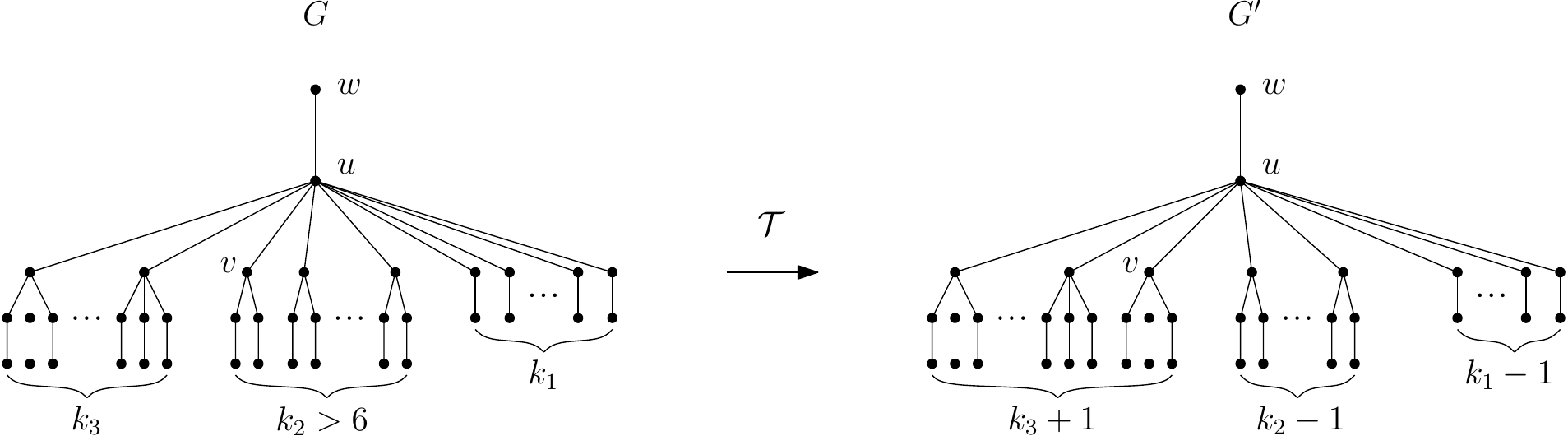}% [height=3.5cm,width=15cm]
\caption{Transforamation $\mathcal{T}$ from  Lemma~\ref{lemma-B1-10}. Note that 
in this illustration all $k_2$ children of $u$ are of degree $4$. As it is shown in the proof below,
in that case the change of the ABC index is the largest.}
%\label{Unicyclic-Max-M3}
\label{fig-B1-1b}
%\vspace{-0.3cm}
\end{center}
\end{figure}
After this transformation the degree of the vertex $u$ decreases by one, while the degree of the vertex $v$ increases by one.
The degrees of other vertices remain unchanged.
The change of the ABC index is
\beq \label{change-10-b}
&&-f(d(u),3)+f(d(u)-1,4)+(k_2-1)(-f(d(u),3)+f(d(u)-1,3)) \nonumber \\
&&+(d(u)-k_2-k_1-1)(-f(d(u),4)+f(d(u)-1,4)) -  f(d(u),d(w))+f(d(u)-1,d(w)),  \nonumber \\
%&&+\sum_{i=1}^{d(u)-k_1-8}(-f(d(u),d(x_i))+f(d(u)-1,d(x_i))) -f(d(u),d(w))+f(d(u)-1,d(w)), \nonumber \\
%-\sqrt{\frac{d(u_2)+6-2}{ 6 d(u_2)}} + \sqrt{\frac{d(u_2)+ 6 -2}{ 5 (d(u_2)+1)}}.
\eeq
By Proposition~\ref{appendix-pro-030-2}, the expression
$-f(d(u),d(w))+f(d(u)-1,d(w))$ increases  in $w$,  and thus, 
it is maximal when  $d(w) \to \infty$. 
By the same proposition, $-f(d(u),4)+f(d(u)-1,4) > -f(d(u),3)+f(d(u)-1,3)$,
and therefore, (\ref{change-10-b}) is maximal when $k_2$ is minimal,
i.e., $k_2=7$. 
Since $-f(d(u),4)+f(d(u)-1,4)$ is strictly positive for any $d(u)$,
(\ref{change-10-b}) is maximal when $k_1$  minimal,
i.e., $k_1=1$. 
Hence,
\beq \label{change-20-b}
&&-f(d(u),3)+f(d(u)-1,4)+6(-f(d(u),3)+f(d(u)-1,3)) \nonumber \\
&&+(d(u)-9)(-f(d(u),4)+f(d(u)-1,4)) + \lim_{d(w) \to \infty} (-  f(d(u),d(w))+f(d(u)-1,d(w))) \nonumber \\
%&&+\sum_{i=1}^{d(u)-k_1-8}(-f(d(u),d(x_i))+f(d(u)-1,d(x_i))) -f(d(u),d(w))+f(d(u)-1,d(w)), \nonumber \\
%-\sqrt{\frac{d(u_2)+6-2}{ 6 d(u_2)}} + \sqrt{\frac{d(u_2)+ 6 -2}{ 5 (d(u_2)+1)}}.
\eeq
is an upper bound on  (\ref{change-10-b}).
It can be verified that the expression (\ref{change-20-b})  is negative for $d(u) \geq 9$,
and reaches its maximum for of $\approx -0.05141846$ for $d(u)=15$.
Hence, the  change of the ABC index (\ref{change-10-b}), 
after applying the transformation $\mathcal{T}$, is negative, which is a contradiction to the assumption 
that $T$ is a subtree
of a tree with minimal-ABC index.

Consider now the case when $u$ is the root vertex of the tree with a minimal-ABC index.
We have the same configuration and apply the same transformation as in Figure~\ref{fig-B1-1}.
Here, it holds that $d(u)=k_1+k_2+k_3$.
Now the change of the ABC index is
\beq \label{change-30-b}
&&-f(d(u),3)+f(d(u)-1,4)+(k_2-1)(-f(d(u),3)+f(d(u)-1,3)) \nonumber \\
&&+(d(u)-k_2-k_1)(-f(d(u),4)+f(d(u)-1,4)) -  f(d(u),d(w))+f(d(u)-1,d(w)),  \nonumber \\
%&&+\sum_{i=1}^{d(u)-k_1-8}(-f(d(u),d(x_i))+f(d(u)-1,d(x_i))) -f(d(u),d(w))+f(d(u)-1,d(w)), \nonumber \\
%-\sqrt{\frac{d(u_2)+6-2}{ 6 d(u_2)}} + \sqrt{\frac{d(u_2)+ 6 -2}{ 5 (d(u_2)+1)}}.
\eeq
Similarly as above we obtain that  (\ref{change-30-b}) as most
\beq \label{change-40-b}
&&-f(d(u),3)+f(d(u)-1,4)+6(-f(d(u),3)+f(d(u)-1,3)) \nonumber \\
&&+(d(u)-8)(-f(d(u),4)+f(d(u)-1,4)) + \lim_{d(w) \to \infty} (-  f(d(u),d(w))+f(d(u)-1,d(w))). \nonumber \\
%&&+\sum_{i=1}^{d(u)-k_1-8}(-f(d(u),d(x_i))+f(d(u)-1,d(x_i))) -f(d(u),d(w))+f(d(u)-1,d(w)), \nonumber \\
%-\sqrt{\frac{d(u_2)+6-2}{ 6 d(u_2)}} + \sqrt{\frac{d(u_2)+ 6 -2}{ 5 (d(u_2)+1)}}.
\eeq
and it is maximal for $k_1=1$ and $k_2=7$.
The expression (\ref{change-40-b}), and therefore (\ref{change-30-b}),
is always negative, and it maximal value of $\approx -0.048948$ is obtained for $d(u)=14$.
%The smallest $d(u)$ for which (\ref{change-70}) is negative ($\approx -0.000580929$) is $d(u)=12$.
%By Proposition~\ref{appendix-pro-040} the expression $-f(x,y)+f(x-1,y)$ increases in $y$, which yields
% $\lim_{d(w)\to \infty}$ $(-f(d(u),d(w))+f(d(u)-1,d(w)))>-f(u,4)+f(d(u)-1,4)$.
%Since, (\ref{change-30}) is negative, it follows that (\ref{change-70}), and consequently (\ref{change-60}) are negative.
Thus, in this case we again obtain that  
after applying the transformation $\mathcal{T}$, the value of the ABC index decrease, which is a contradiction to the assumption 
that $T$, $k \geq 7$, is a tree with minimal-ABC index.
\end{proof}

\begin{lemma} \label{lemma-B2-20}
Let $w$ be a vertex of a minimal-ABC tree $G$ different than the root of $G$.
Then, $w$ is a parent of at most six $B_2$-branches.
\end{lemma} 
\begin{proof}

If  $w$  is a parent of a $B_1$-branch, then by Proposition~\ref{pro-Tk-B1-10}, $G$
contains at most $6$ $B_1$-branches, and the lemma holds.
Thus, we assume that $w$ does not have $B_1$-branches as children.
%$w$ does not have $B_1$-branches as children.
Let $n_2$ be the number of $B_2$-branches that are children of $w$ and
let $n_3$ the number of children vertices of $w$ with degree at least $4$.
By Lemma~\ref{lemma-B2-30},  $n_2 \leq 11$.
Assume that $w$ has more than six $B_2$-branches.
We consider three possible transformations  $\mathcal{T}_{1}$ and $\mathcal{T}_{2}$,
illustrated in Figure~\ref{fig-B2-20}.
\begin{figure}[h]
\begin{center}
%\vspace{-0.3cm}
\includegraphics[scale=0.750]{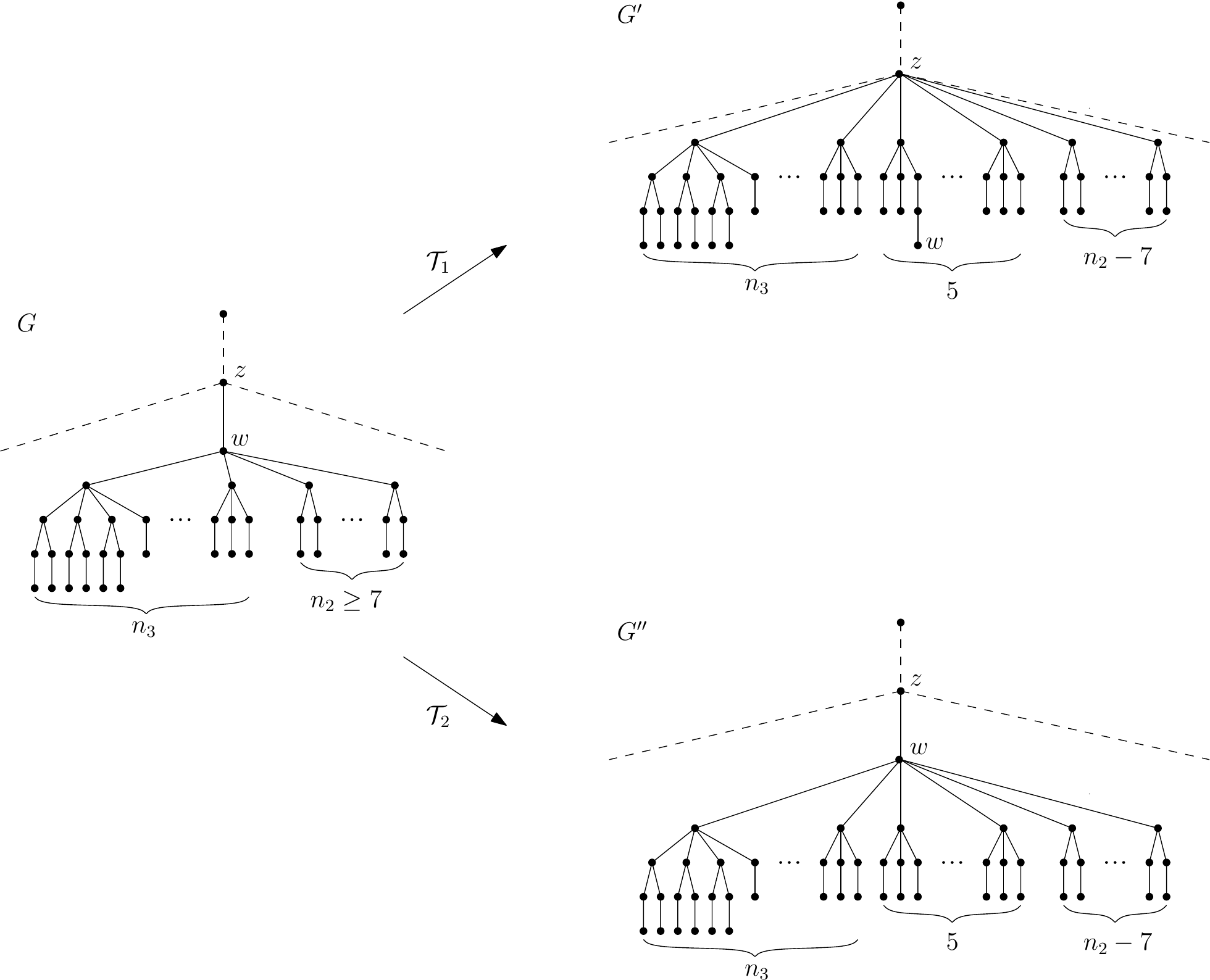}% [height=3.5cm,width=15cm]
\caption{Transforamations $\mathcal{T}_{1}$ and $\mathcal{T}_{2}$  from  the proof of Lemma~\ref{lemma-B2-20}.}
%\label{Unicyclic-Max-M3}
\label{fig-B2-20}
%\vspace{-0.3cm}
\end{center}
\end{figure}
After applying $\mathcal{T}_{1}$ the
degree of the vertex $z$ increases by $d(w)-4$, the degrees of five children vertices 
of $w$  increase from $3$ to $4$, the degree of the vertex $w$ decreases to $1$, 
while two children vertices of $w$ decrease their degrees from $3$ to $2$ and $1$, respectively.
The rest of the vertices do not change their degrees.
The change of the ABC index after applying $\mathcal{T}_{1}$ is at most
\beq \label{lemma-change-B2-20-10}
&&-f(d(z),d(w)) -f(3,d(w))+f(2,1)) -f(3,d(w))+f(4,2)  \nonumber \\
&& +5(-f(3,d(w))+f(4,d(z)+d(w)-4)) + (n_2-7)(-f(3,d(w))+f(3,d(z)+d(w)-4)) \nonumber \\
&&+\sum_{i=1}^{n_3}( -f(x_i,d(w))+f(x_i,d(z)+d(w)-4)) \nonumber \\
&&+(d(z)-2)(-f(4,d(w))+f(4,d(z)+d(w)-4)).
\eeq 
By Proposition~\ref{appendix-pro-030-2}  $-f(x_i,d(w))+f(x_i,d(z)+d(w)-4)$ decreases in $x_i$,
i.e., it is maximal for $x_i=4$. Together with  $d(w)=n_3+n_2+1$ we obtain that
\beq \label{lemma-change-B2-20-10-upper}
f_1(d(z), d(w), n_2)&=&-f(d(z),d(w))+f(2,1)) -f(3,d(w))+f(2,1)) -f(3,d(w))+f(4,2)  \nonumber \\
&& +5(-f(3,d(w))+f(4,d(z)+d(w)-4))  \nonumber \\
&& +(n_2-7)(-f(3,d(w))+f(3,d(z)+d(w)-4)) \nonumber \\
&&+(d(w)-n_2-1)( -f(4,d(w))+f(4,d(z)+d(w)-4)) \nonumber \\
&&+(d(z)-2)(-f(4,d(w))+f(4,d(z)+d(w)-4)) \nonumber 
\eeq 
is an upper bound on $(\ref{lemma-change-B2-20-10})$.
Let consider the expression
$g(d(z),d(w))=-f(d(z),d(w))+f(2,1))+(d(z)-2)(-f(4,d(w))+f(4,d(z)+d(w)-4))$
comprised of the components of $f_1(d(z), d(w), n_2)$.
The first derivative of $g(d(z),d(w))$ with respect to $d(z)$ is
\beq \label{lemma-change-B2-20-10-5-upper}
\frac{\partial g(d(z),d(w))}{\partial d(z)}&=&
\frac{1}{2} \left(-\sqrt{\frac{2+d(w)}{d(w)}}+\frac{-2+d(w)}{d(w) d(z)^2 \sqrt{\frac{-2+d(w)+d(z)}{d(w) d(z)}}} \right.  \nonumber \\ 
&& \left. -\frac{-2+d(z)}{(-4+d(w)+d(z))^2 \sqrt{\frac{-2+d(w)+d(z)}{-4+d(w)+d(z)}}}+\sqrt{\frac{-2+d(w)+d(z)}{-4+d(w)+d(z)}}\right).  \nonumber 
\eeq 
For $d(z) \geq d(w) \geq 9$, we have that 
$$
\sqrt{\frac{2+d(w)}{d(w)}}  > \sqrt{\frac{-2+d(w)+d(z)}{-4+d(w)+d(z)}}.
$$
Next we show that
\beq \label{lemma-change-B2-20-10-6-upper}
\frac{-2+d(w)}{d(w) d(z)^2 \sqrt{\frac{-2+d(w)+d(z)}{d(w) d(z)}}} < \frac{-2+d(z)}{(-4+d(w)+d(z))^2 \sqrt{\frac{-2+d(w)+d(z)}{-4+d(w)+d(z)}}}.
\eeq
Indeed, from (\ref{lemma-change-B2-20-10-6-upper}), it follows that
\beq \label{lemma-change-B2-20-10-7-upper}
(-2+d(w))^2(-4+d(w)+d(z))^3 < (-2+d(z))^2d(w) d(z)^3, \qquad \text{or} \nonumber
\eeq
\beq \label{lemma-change-B2-20-10-8-upper}
(-4+d(w)+d(z))^3 < d(w) d(z)^3, \qquad \text{or} \nonumber
\eeq
\beq \label{lemma-change-B2-20-10-8-upper}
(2d(z))^3 < d(w) d(z)^3, \nonumber
\eeq
which is satisfied since $d(w) \geq 9$.
Thus, $\partial g(d(z),d(w)) / \partial d(z) < 0$.
Because $f(x,y)$ decreases in $x$, it follows that 
$5(-f(3,d(w))+f(4,d(z)+d(w)-4)) ,  (n_2-7)(-f(3,d(w))+f(3,d(z)+d(w)-4))$, and $(d(w)-n_2-1)( -f(4,d(w))+f(4,d(z)+d(w)-4))$ decrease in $d(z)$.
Thus, $f_1(d(z), d(w), n_2)$ is maximal when $d(z)$ is minimal,i.e., $d(z)=d(w)$.

After applying $\mathcal{T}_{2}$ the degrees of five children vertices 
of $w$  increase from $3$ to $4$, the degree of the vertex $w$ decreases by $2$, 
while two children vertices of $w$ decrease their degrees from $3$ to $2$ and $1$, respectively.
The rest of the vertices do not change their degrees.
The change of the ABC index after applying $\mathcal{T}_{1}$ is at most
\beq \label{lemma-change-B2-20-20}
 &&-f(d(z),d(w))+f(d(z),d(w)-2) -f(3,d(w))+f(2,1)) -f(3,d(w))+f(4,2)  \nonumber \\
&& +5(-f(3,d(w))+f(4,d(w)-2)) + (n_2-7)(-f(3,d(w))+f(3,d(w)-2)) \nonumber \\
&&+\sum_{i=1}^{n_3}( -f(x_i,d(w))+f(x_i,d(w)-2)). 
\eeq 
By Proposition~\ref{appendix-pro-030-2}  $-f(x_i,d(w))+f(x_i,d(z)+d(w)-4)$ increase in $x_i$.
Because $w$ has children of degree $3$, by Theorems~\ref{thm-LG-10} and \ref{thm-DS}, it follows that $x_i$ may have
children of degree $3$ or $2$. Together with Lemma~\ref{lemma-B1-10-2}, we conclude that the maximal possible value of  $x_i$ is $8$.
Together with  $d(w)=n_3+n_2+1$ we obtain that
\beq \label{lemma-change-B2-20-20-upper}
 f_2(d(z), d(w), n_2)&=& -f(d(z),d(w))+f(d(z),d(w)-2)  \nonumber \\
&& -f(3,d(w))+f(2,1)) -f(3,d(w))+f(4,2)  \nonumber \\
&& +5(-f(3,d(w))+f(4,d(w)-2))  \nonumber \\
&& + (n_2-7)(-f(3,d(w))+f(3,d(w)-2)) \nonumber \\
&&+(d(w)-n_2-1)(-f(8,d(w))+f(8,d(w)-2)) \nonumber
\eeq 
is an upper bound on $(\ref{lemma-change-B2-20-20})$.
The expression $-f(d(z),d(w))$, and therefore also $f_2(d(z), d(w), n_2)$,  increases in $d(z)$, 
and  it is maximal $d(z) \to \infty$.

%For $7 \leq n_2 \leq 11$, it can be verified that  $f_1(d(z), d(w), n_2)$, $f_2(d(z), d(w), n_2)$, and $f_3(d(z),$ $d(w), n_2)$
%are increasing function in $d(z)$, so they are maximal when $d(z) \to \infty$. (ARGUE ABOUT THIS).

\noindent
We distinguish three cases with respect to $n_2$.

\smallskip
\noindent
{\bf Case $1$. $n_2=7.$}

\noindent
It holds that 
$f_1(d(z), d(w), 7) < f_1(d(w), d(w), 7) <0$ for $d(w) \leq 228$ and
$f_2(d(z), d(w), 7) < \lim_{d(z) \to \infty}f_2(d(z), d(w), 7)< 0$ for $d(w) > 16$.
So, we apply $\mathcal{T}_{1}$ for $d(w) \leq 16$ and $\mathcal{T}_{2}$ for $d(w)  > 16$,
and obtain trees with smaller ABC index than that of $G$, which is a contradiction to the
claim of the lemma.

\smallskip
\noindent
{\bf Case $2$. $n_2=8,9,10$.}

\noindent
It holds that 
$f_1(d(z), d(w), 8) <  f_1(d(w), d(w,) 8)  <0$ for $d(w) \leq 225$,  $f_1(d(z), d(w), 9) <  f_1(d(w), d(z,) 9)  <0$ for $d(w) \leq 221$ and
$f_1(d(z), d(w), 10) <  f_1(d(w), d(z,) 10)  <0$ for $d(w) \leq 218$.
Also, we have 
$f_2(d(z), d(w), 8) < \lim_{d(z) \to \infty}f_2(d(z), d(w), 8)< 0$ for $d(w) > 15$,
$f_2(d(z), d(w), 9) < \lim_{d(z) \to \infty}f_2(d(z), d(w), 9)< 0$ for $d(w) > 14$, and
$f_2(d(z), d(w), 10) < \lim_{d(z) \to \infty}f_2(d(z), d(w), 10)< 0$ for $d(w) > 12$.

\noindent
Here, we apply $\mathcal{T}_{1}$ for $d(w) \leq 15$  and  $\mathcal{T}_{2}$ for $d(w) > 15$,
and obtain trees with smaller ABC index than that of $G$, which is a contradiction to the
claim of the lemma.

\smallskip
\noindent
{\bf Case $3$. $n_2=11$}.

\noindent
Here for  for $d(w) > 8$, it holds that $f_2(d(z), d(w), 11) < \lim_{d(z) \to \infty}f_2(d(z), d(w), 11)< 0$, and thus
obtain a tree with smaller ABC index than that of $G$, which is a contradiction to the
claim of the lemma. This completes the proof.
\end{proof}
 
 \noindent
Next, we present the main result of this section.

\begin{te} \label{te-no2branches-10a}
A minimal-ABC tree does not contain more than eleven $B_2$-branches.
\end{te}
\begin{proof} 
%Denote by $G$  a minimal-ABC tree.  
First, consider the case  when the $B_2$-branches have $k \geq 3$ different parent vertices, denoted
by $w_1, w,_2, \dots, w_{k-1}, w_k$, such that, $d(w_1) \geq d(w_2)\geq \dots \geq d(w_{k-1}) \geq d(w_k)$.
By Theorem~\ref{thm-DS} only $w_1$ and $w_k$ may have children that are not roots of  $B_2$-branches:
$w_1$ may have in addition children vertices of degree $\geq 3$, while $w_k$ may have in addition only
children vertices of degree $2$. Moreover, by Theorem~\ref{te-no5branches-10} and Lemma~\ref{lemma-15}, 
if $w_k$ has a children of degree $2$, then $w_1$ in addition to $B_2$-branches may have only  $B_3$-branches
as its children. Note that  $w_2, w_3, \dots, w_{k-1}, w_k$ cannot be root vertices and by Lemma~\ref{lemma-B2-20} they can be
parents of at most $6$ $B_2$-branches.

We apply a transformation on the $B_2$-branches that are children of the $w_1$ and $w_{k-1}$ vertices.
There are three distinct cases regarding the parent vertices of $w_1$ and $w_{k-1}$, denoted by $z_1$ and $z_{k-1}$:
$z_1 \neq z_{k-1}$, $z_1 = z_{k-1}$ and $z_{k-1}=w_1$.

\bigskip
\noindent
{\bf Case $1$.} $z_1 \neq z_{k-1}$.

\noindent
Notice, that  $z_1$ and  $z_{k-1}$ may belong to different levels of $G$,
but by Theorem~\ref{thm-DS}, it follows that their distance to the root vertex of $G$ may differ for at most $1$.
Let $w_1$ be a parent of $n_1$ $B_2$-branches, and $w_{k-1}$ be a parent of $n_{k-1}=d(w_{k-1})-1$ $B_2$-branches.
By Theorem~\ref{thm-GFI-10}, $z_{k-1}$ cannot have a child of degree $2$.
Also, by Proposition~\ref{pro-B2-10},  $w_{k-1}$ cannot have two $B_2$-branches as its children, thus it follows that
it  has at least $3$ $B_2$-branches as its children.
With respect to the number of $B_2$-branches that are attached to $w_{k-1}$, we distinguish two further subcases.

\smallskip
\noindent
{\bf Subcase $1.1.$} $d(w_{k-1})=4$.

\noindent
In this case we apply the transformation $\mathcal{T}_{1}$ illustrated in Figure~\ref{fig-te-numberB2-10}.
\begin{figure}[h]
\begin{center}
%\vspace{-0.3cm}
\includegraphics[scale=0.75]{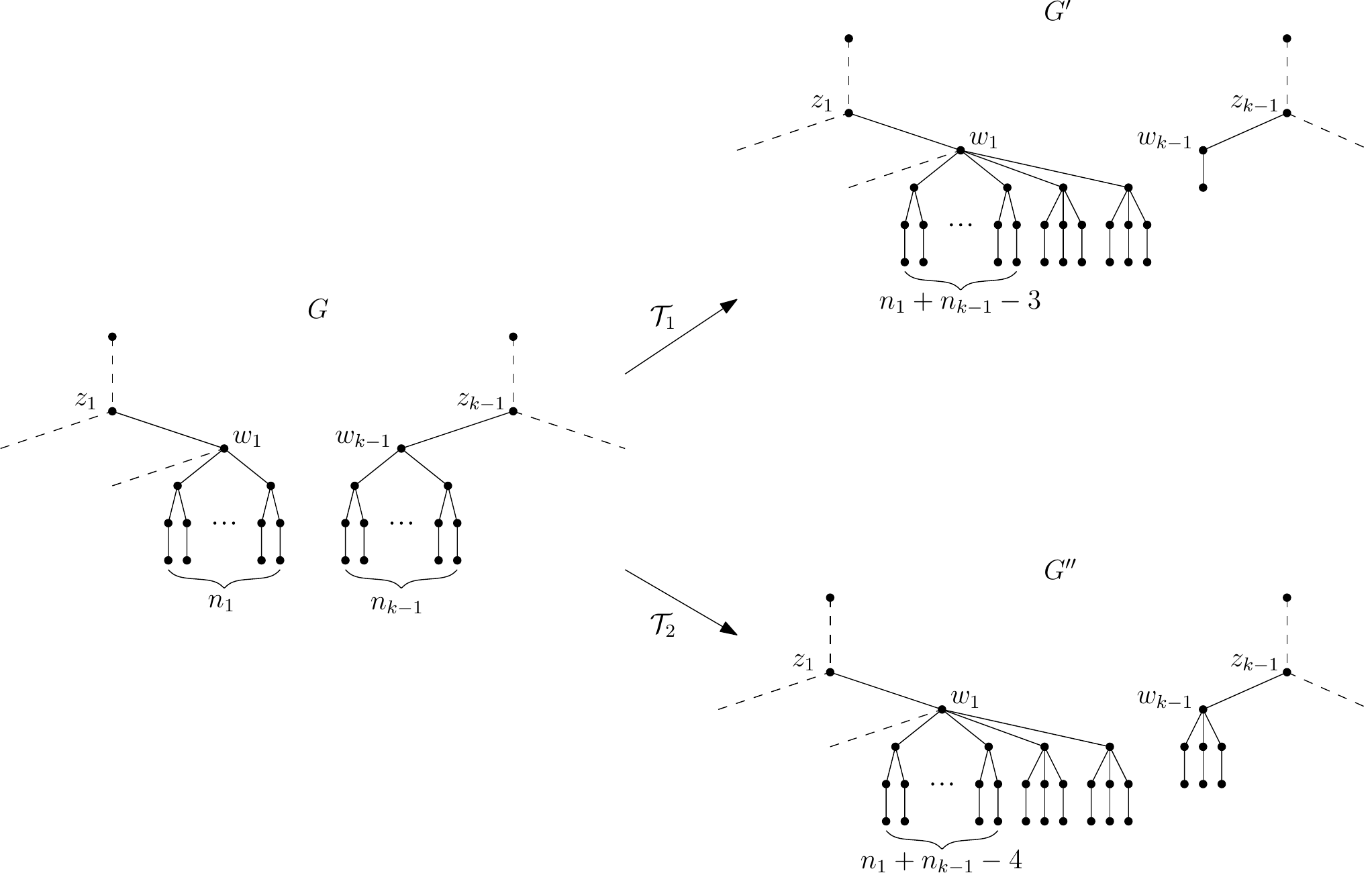}% [height=3.5cm,width=15cm]
\caption{An illustration of the transformations $\mathcal{T}_{1}$ and $\mathcal{T}_{2}$ from the proof of Theorem~\ref{te-no2branches-10a} (Case $1$).
$\mathcal{T}_{1}$ is applied where $n_{k-1}=3$, while for $n_{k-1}=4, 5$ and $6$, we apply  $\mathcal{T}_{2}$.}
%\label{Unicyclic-Max-M3}
\label{fig-te-numberB2-10}
%\vspace{-0.3cm}
\end{center}
\end{figure}
After applying $\mathcal{T}_{1}$ 
the degree of the vertex $w_1$ increases by $2$, 
the degree of the vertex $w_{k-1}$ decreases to $2$, 
two children vertices of the vertex $w_{k-1}$ increase their degrees from $3$ to $4$, 
while one child  of $w_{k-1}$ decreases its degree to $1$.
The rest of the vertices do not change their degrees.
The change of the ABC index after applying $\mathcal{T}_{1}$ is bounded from above by
\beq \label{thm-noB2-10}
 &&-f(d(z_{k-1}),4)+f(d(z_{k-1}),2)) -f(4,3)+f(2,1)) \nonumber \\
 &&+2( -f(4,3)+f(d(w_{1})+2,4))  \nonumber \\
&&+n_1( -f(d(w_1),3)+f(d(w_{1})+2,3))  \nonumber \\
&&+\sum_{i=1}^{d(w_1)-n_1-1}(-f(d(w_1), x_i)+f(d(w_{1})+2, x_i)) \nonumber \\
&&-f(d(z_1),d(w_1))+f(d(z_1), d(w_{1})+2).
\eeq 
By Proposition~\ref{appendix-pro-030-2}  $-f(d(w_1), x_i)+f(d(w_{1})+2, x_i)$ decrease in $x_i$,
so it is maximal for $x_i=4$.
The expression $-f(d(z_{k-1}),4)+f(d(z_{k-1}),2))$ increases in $d(z_{k-1})$, so we obtain 
an upper bound on (\ref{thm-noB2-10}), if we set $d(z_{k-1})=d(z_1)$.
The expression $$g_1(d(z_{1}), d(w_{1}))=-f(d(z_{1}),4)+f(d(z_1),2)) -f(d(z_1),d(w_1))+f(d(z_1), d(w_{1})+2)$$ increases in
$d(z_{1})$, because its first derivative with respect to $d(z_{1})$,
\beq \label{thm-noB2-10b}
 \frac{\partial g_1(d(z_{1}), d(w_{1}))}{\partial d(z_1)} &=&\frac{1}{4} \left(\frac{2}{d(z_1)^2 \sqrt{\frac{2+d(z_1)}{d(z_1)}}}+
\frac{2 (-2+d(w_1))}{d(w_1) d(z_1)^2 \sqrt{\frac{-2+d(w_1)+d(z_1)}{d(w_1) d(z_1)}}} \right. \nonumber \\
&& \left.  -\frac{2 d(w_1) \sqrt{\frac{d(w_1)+d(z_1)}{(2+d(w_1)) d(z_1)}}}{d(w_1) d(z_1)+d(z_1)^2}\right) \nonumber \\
\eeq 
is positive for $d(z_{1}) \geq d(w_1) \geq 4$. Thus, 
\beq \label{thm-noB2-10c}
 \lim_{d(z_{1}) \to \infty} g(d(z_{1}), d(w_{1}), n_1) &=& \lim_{d(z_{1}) \to \infty} -f(d(z_{1}),4)+f(d(z_{1}),2)) -f(4,3)+f(2,1)) \nonumber \\
 &&+2( -f(4,3)+f(d(w_{1})+2,4))   \nonumber \\
&&+n_1( -f(d(w_1),3)+f(d(w_{1})+2,3))  \nonumber \\
&&+(d(w_1)-n_1-1)(-f(d(w_1), 4)+f(d(w_{1})+2, 4)) \nonumber \\
&&-f(d(z_1),d(w_1))+f(d(z_1), d(w_{1})+2) \nonumber
\eeq 
is an upper bound on (\ref{thm-noB2-10}).
Considering the maximum of $\lim_{d(z_{1}) \to \infty} g(d(z_{1}), d(w_{1}), n_1)$ for $n_{1}=1,2,3,4,5,6$,
we show the $\lim_{d(z_{1}) \to \infty} g(d(z_{1}), d(w_{1}), n_1)$ is always negative.
Namely, 
\beq \label{thm-noB2-10d}
 &&\lim_{\substack{d(z_{1}) \to \infty } } g(d(z_{1}), d(w_{1}), 1) \leq   \lim_{\substack{d(z_{1}) \to \infty \\ d(w_{1}) \to \infty} } g(d(z_{1}), d(w_{1}), 1) =  -0.0222781, \nonumber \\
 &&\lim_{\substack{d(z_{1}) \to \infty } } g(d(z_{1}), d(w_{1}), 2) \leq   \lim_{\substack{d(z_{1}) \to \infty \\ d(w_{1}) \to \infty} } g(d(z_{1}), d(w_{1}), 2) =  -0.0222781, \nonumber \\
&& \lim_{\substack{d(z_{1}) \to \infty } } g(d(z_{1}), d(w_{1}), 3) \leq   \lim_{\substack{d(z_{1}) \to \infty \\ d(w_{1}) \to \infty} } g(d(z_{1}), d(w_{1}), 3) =  -0.0222781, \nonumber \\
&& \lim_{\substack{d(z_{1}) \to \infty } } g(d(z_{1}), d(w_{1}), 4) \leq   \lim_{\substack{d(z_{1}) \to \infty }} g(d(z_{1}), 5, 4) =  -0.0186023, \nonumber \\
&& \lim_{\substack{d(z_{1}) \to \infty } } g(d(z_{1}), d(w_{1}), 5) \leq   \lim_{\substack{d(z_{1}) \to \infty }} g(d(z_{1}), 6, 5) =  -0.0151247,  \quad \text{and} \nonumber \\
&& \lim_{\substack{d(z_{1}) \to \infty } } g(d(z_{1}), d(w_{1}), 6) \leq   \lim_{\substack{d(z_{1}) \to \infty }} g(d(z_{1}), 7, 6) =  -0.0131643. \nonumber
 \eeq
From here follows that also the change of the ABC index  (\ref{thm-noB2-10}) is negative.

\smallskip
\noindent
{\bf Subcase $1.2.$} $d(w_{k-1})=5, 6$ or $7$.

\noindent
In this case, for $n_{1}=1,2,3,4,5,6$, we apply the transformation $\mathcal{T}_{2}$ 
illustrated also in Figure~\ref{fig-te-numberB2-10}.
After applying $\mathcal{T}_{2}$ 
the degree of the vertex $w_1$ increases by $d(w_{k-1})-3=n_{k_1}-2$, 
the degree of the vertex $w_{k-1}$ decreases to $4$, 
two children vertices of the vertex $w_{k-1}$ increase their degrees from $3$ to $4$, 
while two children vertices of $w_{k-1}$ decrease their degrees to $2$ and $1$, respectively.
The rest of the vertices do not change their degrees.
The change of the ABC index after applying $\mathcal{T}_{2}$ is bounded from above by
\beq \label{thm-noB2-20}
 &&-f(d(z_{k-1}),d(w_{k-1}))+f(d(z_{k-1}),4)) -f(d(w_{k-1}),3)+f(4,2))\nonumber \\
 && -f(d(w_{k-1}),3)+f(2,1)) +2( -f(d(w_{k-1}),3)+f(d(w_1)+n_{k-1}-2,4))   \nonumber \\
&&+n_1( -f(d(w_1),3)+f(d(w_1)+n_{k-1}-2,3))  \nonumber \\
&&+\sum_{i=1}^{d(w_1)-n_1-1}(-f(d(w_1), x_i)+f(d(w_1)+n_{k-1}-2, x_i)) \nonumber \\
&&+(n_{k-1}-4)( -f(d(w_{k-1}),3)+f(d(w_1)+n_{k-1}-2,3)) \nonumber \\
&& -f(d(z_1),d(w_1))+f(d(z_1), d(w_{1})+n_{k-1}-2)).
\eeq 
By Proposition~\ref{appendix-pro-030-2}  $-f(d(w_1), x_i)+f(d(w_{1})+2, x_i)$ decrease in $x_i$,
so it is maximal for $x_i=4$.

Next, we consider  the case $d(w_{k-1})=5$ ($n_{k-1}=4$).
The expression $-f(d(z_{k-1}),5)+f(d(z_{k-1}),4))$ increases in $d(z_{k-1})$, so we obtain 
an upper bound on (\ref{thm-noB2-20}), if we set $d(z_{k-1})=d(z_1)$.
The expression $$g_1(d(z_{1}), d(w_{1}))=-f(d(z_{1}),5)+f(d(z_1),4)) -f(d(z_1),d(w_1))+f(d(z_1), d(w_{1})+2)$$ increases in
$d(z_{1})$, because its first derivative with respect to $d(z_{1})$,
\beq \label{thm-noB2-20b}
 \frac{\partial g_1(d(z_{1}), d(w_{1}))}{\partial d(z_1)} &=&\frac{1}{10} \left(-\frac{5}{d(z_1)^2 \sqrt{\frac{2+d(z_1)}{d(z_1)}}}
 +\frac{3 \sqrt{5}}{d(z_1)^2 \sqrt{\frac{3+d(z_1)}{d(z_1)}}} \right. \nonumber \\
 && \qquad \left. +\frac{5 (-2+d(w_1))}{d(w_1) d(z_1)^2 \sqrt{\frac{-2+d(w_1)+d(z_1)}{d(w_1) d(z_1)}}}
 -\frac{5 d(w_1) \sqrt{\frac{d(w_1)+d(z_1)}{(2+d(w_1)) d(z_1)}}}{d(w_1) d(z_1)+d(z_1)^2}\right) \nonumber
\eeq 
is positive for $d(z_{1}) \geq d(w_1) \geq 5$. Thus, 
\beq \label{thm-noB2-20c}
 \lim_{d(z_{1}) \to \infty} g(d(z_{1}), d(w_{1}), n_1) &=&
\lim_{d(z_{1}) \to \infty}   -f(d(z_1),d(w_1))+f(d(z_{k-1}),4)) -f(5,3)+f(4,2))\nonumber \\
&&-f(5),3)+f(2,1)) +2( -f(5,3)+f(d(w_1)+2,4))   \nonumber \\
&& +n_1( -f(d(w_1),3)+f(d(w_1)+2,3))  \nonumber \\
&& +(d(w_1)-n_1-1)(-f(d(w_1), 4)+f(d(w_1)+2, 4)) \nonumber \\
&& -f(d(z_1),d(w_1))+f(d(z_1), d(w_{1})+2)). \nonumber
\eeq 
is an upper bound on (\ref{thm-noB2-20}).
Considering the maximum of $\lim_{d(z_{1}) \to \infty} g(d(z_{1}), d(w_{1}), n_1)$ for $n_{1}=1,2,3,4,5,6$,
we show the $\lim_{d(z_{1}) \to \infty} g(d(z_{1}), d(w_{1}), n_1)$ is always negative.
Namely, 
\beq \label{thm-noB2-10d}
 &&\lim_{\substack{d(z_{1}) \to \infty } } g(d(z_{1}), d(w_{1}), 1) \leq   \lim_{\substack{d(z_{1}) \to \infty \\ d(w_{1}) \to \infty} } g(d(z_{1}), d(w_{1}), 1) =  -0.0628222, \nonumber \\
 &&\lim_{\substack{d(z_{1}) \to \infty } } g(d(z_{1}), d(w_{1}), 2) \leq   \lim_{\substack{d(z_{1}) \to \infty \\ d(w_{1}) \to \infty} } g(d(z_{1}), d(w_{1}), 2) =  -0.0628222, \nonumber \\
&& \lim_{\substack{d(z_{1}) \to \infty } } g(d(z_{1}), d(w_{1}), 3) \leq   \lim_{\substack{d(z_{1}) \to \infty \\ d(w_{1}) \to \infty} } g(d(z_{1}), d(w_{1}), 3) =  -0.0628222, \nonumber \\
&& \lim_{\substack{d(z_{1}) \to \infty } } g(d(z_{1}), d(w_{1}), 4) \leq   \lim_{\substack{d(z_{1}) \to \infty }} g(d(z_{1}), 5, 4) =  -0.0591464, \nonumber \\
&& \lim_{\substack{d(z_{1}) \to \infty } } g(d(z_{1}), d(w_{1}), 5) \leq   \lim_{\substack{d(z_{1}) \to \infty }} g(d(z_{1}), 6, 5) =  -0.0556687,  \quad \text{and} \nonumber \\
&& \lim_{\substack{d(z_{1}) \to \infty } } g(d(z_{1}), d(w_{1}), 6) \leq   \lim_{\substack{d(z_{1}) \to \infty }} g(d(z_{1}), 7, 6) =  -0.0537084. \nonumber
 \eeq
From here follows that also the change of the ABC index  (\ref{thm-noB2-20}) when  $d(w_{k-1})=5$ is negative.
Analogous proofs, one can obtain for $d(w_{k-1})=6,7$, so we omit them.

After the transformation $\mathcal{T}_{1}$ $d(w_{k-1})=2$, and after the transformation $\mathcal{T}_{2}$ $d(w_{k-1})=4$.
Let $z_{k}$ be the parent vertex of $w_{k}$. Notice that degree of $w_k$ is at least $4$.
Now, we interchange the labels of the vertices $w_{k-1}$ and $w_k$. 
%Now, if after applying $\mathcal{T}_{1}$ or $\mathcal{T}_{2}$ it holds that  $d(z_{k-1}) \geq d(z_{k})$, we apply the following transformation:
%we move $w_{k}$ and all his children (immediate and not immediate) and attach to $z_{k-1})$,
%and we move $w_{k-1}$ and his child and attached to $z_{k}$.
Since, $d(z_{k-1}) \geq d(z_{k})$, after this relabeling the ABC index does not increase.
Thus, finally we have obtained a tree with smaller ABC index than $G$  and with $k-1$
different parent vertices of the $B_2$-branches.

\bigskip
\noindent
{\bf Case $2$.} $z_1 = z_{k-1}$.

\noindent
In the previous case, the upper bound of the change of the ABC index of $G$ was obtained for
$d(z_1) = d(z_{k-1})$ ($\to \infty$). Thus, applying here the same transformation from Case~$1$, we obtain the same upper bounds
on the change of the ABC index, which are all negative.

\bigskip
\noindent
{\bf Case $3$.} $z_{k-1}=w_1$.

\noindent
Since $z_{k-1}$ has at same time children of degree larger than $3$ and also children that are roots of $B_2$-branches, by the Theorem~\ref{thm-DS}, 
%greedy property of the minimal-ABC trees, 
it follows that $z_{k-1}$ must be also a parent vertex of $w_{k}$.
Similarly, as in Case $1$, regarding the degree of $w_{k-1}$ we apply two transformations (illustrated in
Figure~\ref{fig-te-numberB2-30}):
when $w_{k-1}=4$, we apply the transformation $\mathcal{T}_{3}$, and when $w_{k-1}=5,6,7$,
we apply the transformation $\mathcal{T}_{4}$.

\smallskip
\noindent
{\bf Subcase $3.1.$} $d(w_{k-1})=4$.

\noindent
After applying $\mathcal{T}_{3}$ 
the degree of the vertex $w_1$ ($z_{k-1}$) increases by $d(w_{k-1})-2=n_{k_1}-1$, 
the degree of the vertex $w_{k-1}$ decreases to $2$, 
two children vertices of the vertex $w_{k-1}$ increase their degrees from $3$ to $4$, 
while one child  of $w$ decreases its degree to $1$.
The rest of the vertices do not change their degrees.
\noindent
\begin{figure}[h]
\begin{center}
%\vspace{-0.3cm}
\includegraphics[scale=0.75]{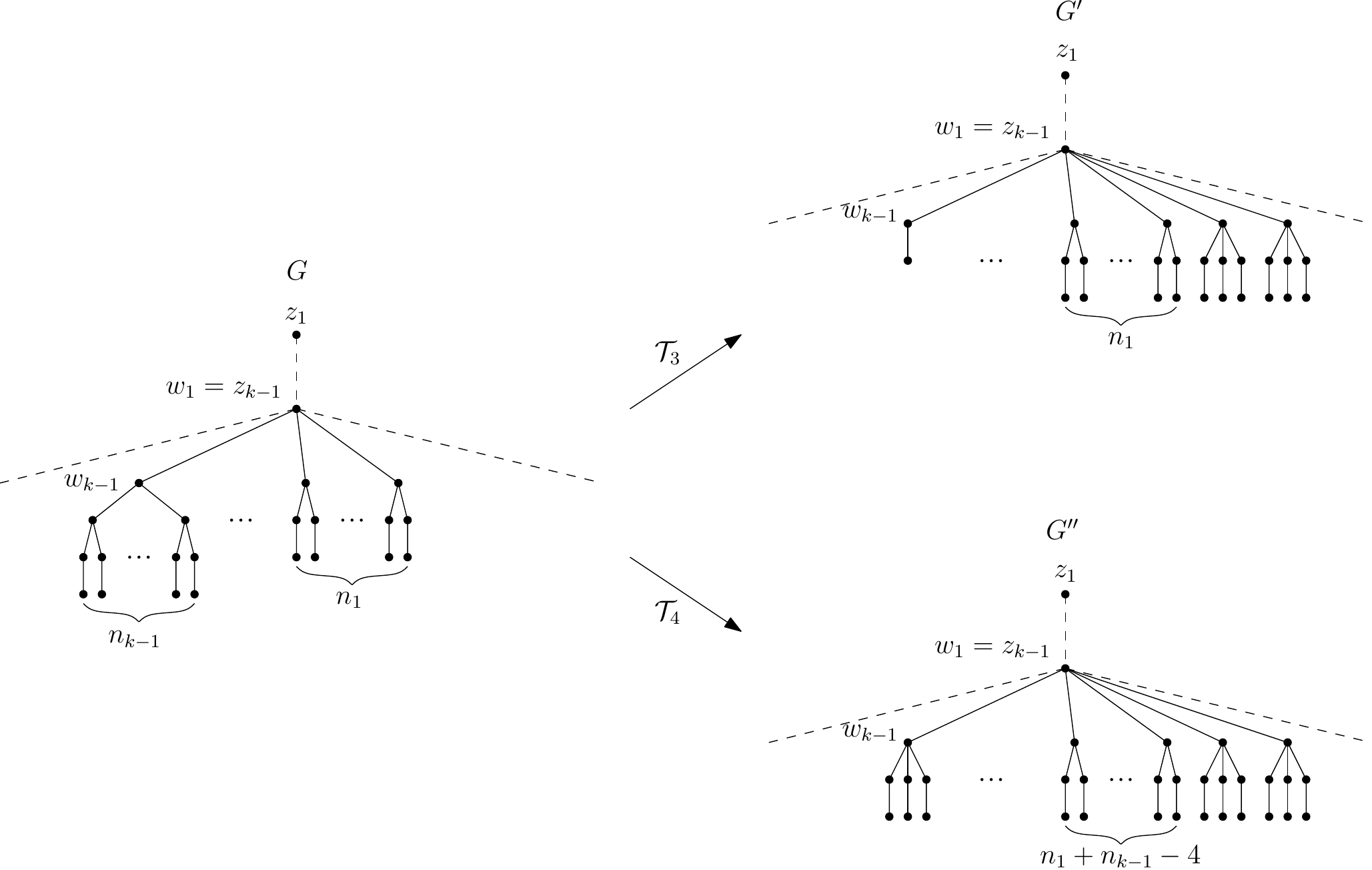}% [height=3.5cm,width=15cm]
\caption{An illustration of the transformations $\mathcal{T}_{3}$ and $\mathcal{T}_{4}$ from the proof of Theorem~\ref{te-no2branches-10a} .
%(Case $3$).$\mathcal{T}_{3}$ is applied where $n_{k-1}=3$, while for $n_{k-1}=4, 5$ and $6$, we apply  $\mathcal{T}_{4}$.
}
%\label{Unicyclic-Max-M3}
\label{fig-te-numberB2-30}
%\vspace{-0.3cm}
\end{center}
\end{figure}
The change of the ABC index after applying $\mathcal{T}_{1}$ is bounded from above by
\beq \label{thm-noB2-100}
 &&-f(d(w_1),4)+f(d(w_1),2)) -f(4,3)+f(2,1)) \nonumber \\
 &&+2( -f(4,3)+f(d(w_{1})+2,4))   \nonumber \\
&&+n_1( -f(d(w_1),3)+f(d(w_{1})+2,3))  \nonumber \\
&&+\sum_{i=1}^{d(w_1)-n_1-2}(-f(d(w_1), x_i)+f(d(w_{1})+2, x_i)) \nonumber \\
&&-f(d(z_1),d(w_1))+f(d(z_1), d(w_{1})+2).
\eeq 
By Proposition~\ref{appendix-pro-030-2}  $-f(d(w_1), x_i)+f(d(w_{1})+2, x_i)$ decrease in $x_i$,
so it is maximal for $x_i=4$.
By the same proposition the expression $-f(d(z_1),d(w_1))+f(d(z_1), d(w_{1})+2)$ decreases in
$d(z_{1})$, and its upper bound is $-f(d(w_1),d(w_1))+f(d(w_1), d(w_{1})+2)$.
Thus, 
\beq \label{thm-noB2-100c}
 g(d(w_{1}), n_1) &=& -f(d(w_{1}),4)+f(d(w_{1}),2)) -f(4,3)+f(2,1)) \nonumber \\
 &&+2( -f(4,3)+f(d(w_{1})+2,4))   \nonumber \\
&&+n_1( -f(d(w_1),3)+f(d(w_{1})+2,3))  \nonumber \\
&&+(d(w_1)-n_1-2)(-f(d(w_1), 4)+f(d(w_{1})+2, 4)) \nonumber \\
&&-f(d(w_1),d(w_1))+f(d(w_1), d(w_{1})+2)
\eeq 
is an upper bound on (\ref{thm-noB2-100}).
Considering the maximum of $g(d(w_{1}), n_1)$ for $n_{1}=1,2,3,4,5,6$,
we show the $g(d(w_{1}), n_1)$ is always negative.
Namely, 
\beq \label{thm-noB2-100d}
 &&g(d(w_{1}), n_1) \leq   \lim_{ d(w_{1}) \to \infty } g(d(z_{1}), d(w_{1}), n_1) =  -0.0222781, \nonumber
 \eeq
 for  $n_{1}=1,2,3,4,5,6$.
From here follows that also the change of the ABC index  (\ref{thm-noB2-100}) is negative.

\smallskip
\noindent
{\bf Subcase $3.2.$} $d(w_{k-1})=5,6$, or $7$.

%For $6 \geq n_{k-1}\geq 4$ and $n_{1}=1,2,3,4,5,6$, we apply the transformation $\mathcal{T}_{4}$ 
%illustrated also in Figure~\ref{fig-te-numberB2-10}.
%
\noindent
After applying $\mathcal{T}_{4}$ 
the degree of the vertex $w_1$ increases by $d(w_{k-1})-3=n_{k_1}-2$, 
the degree of the vertex $w_{k-1}$ decreases to $4$, 
two children vertices of the vertex $w_{k-1}$ increase their degrees from $3$ to $4$, 
while two children vertices of $w$ decreases their degrees to $2$ and $1$, respectively.
The rest of the vertices do not change their degrees.
The change of the ABC index after applying $\mathcal{T}_{4}$ is bounded from above by
\beq \label{thm-noB2-200}
 &&-f(d(w_1),d(w_{k-1}))+f(d(w_1),4)) -f(d(w_{k-1}),3)+f(4,2))\nonumber \\
 && -f(d(w_{k-1}),3)+f(2,1)) +2( -f(d(w_{k-1}),3)+f(d(w_1)+n_{k-1}-2,4))   \nonumber \\
&&+n_1( -f(d(w_1),3)+f(d(w_1)+n_{k-1}-2,3))  \nonumber \\
&&+\sum_{i=1}^{d(w_1)-n_1-2}(-f(d(w_1), x_i)+f(d(w_1)+n_{k-1}-2, x_i)) \nonumber \\
&&+(n_{k-1}-4)( -f(d(w_{k-1}),3)+f(d(w_1)+n_{k-1}-2,3)) \nonumber \\
&& -f(d(z_1),d(w_1))+f(d(z_1), d(w_{1})+n_{k-1}-2)).
\eeq 
By Proposition~\ref{appendix-pro-030-2}  $-f(d(w_1), x_i)+f(d(w_{1})+2, x_i)$ decrease in $x_i$,
so it is maximal for $x_i=4$.
By the same proposition the expression $-f(d(z_1),d(w_1))+f(d(z_1), d(w_{1})+2)$ decreases in
$d(z_{1})$, and its upper bound is $-f(d(w_1),d(w_1))+f(d(w_1), d(w_{1})+2)$.

Next, we consider  the case $d(w_{k-1})=5$ ($n_{k-1}=4$).
The expression $-f(d(z_1),d(w_1))+f(d(z_1), d(w_{1})+n_{k-1}-2))$ decreases in $d(z_{k-1})$, so it is 
bounded from above by $-f(d(w_1),d(w_1))+f(d(w_1), d(w_{1})+n_{k-1}-2))$.
We obtain an upper bound on (\ref{thm-noB2-200}) by setting $d(z_{k-1})=d(z_1)$.
Thus, 
\beq \label{thm-noB2-200c}
g(d(w_{1}), n_1) &=& -f(d(w_1),5)+f(d(z_{k-1}),4)) -f(5,3)+f(4,2))\nonumber \\
&& -f(5,3)+f(2,1)) +2( -f(5,3)+f(d(w_1)+2,4))   \nonumber \\
&&+n_1( -f(d(w_1),3)+f(d(w_1)+2,3))  \nonumber \\
&&+(d(w_1)-n_1-2)(-f(d(w_1), 4)+f(d(w_1)+2, 4)) \nonumber \\
&& -f(d(w_1),d(w_1))+f(d(w_1), d(w_{1})+2)). \nonumber
\eeq 
is an upper bound on (\ref{thm-noB2-200}).
Considering the maximum of $\lim_{d(z_{1}) \to \infty} g(d(z_{1}), d(w_{1}), n_1)$ for $n_{1}=1,2,3,4,5,6$,
we show the $\lim_{d(z_{1}) \to \infty} g(d(z_{1}), d(w_{1}), n_1)$ is always negative.
Namely, 
\beq \label{thm-noB2-10d}
 && g(d(z_{1}), d(w_{1}), 1) \leq   g( 5, 1) =  -0.0515202, \nonumber \\
 && g(d(z_{1}), d(w_{1}), 2) \leq   g( 5, 2) =  -0.0421011, \nonumber \\
&& g(d(z_{1}), d(w_{1}), 3) \leq    g(5, 3) =  -0.0326819, \nonumber \\
&&  g(d(z_{1}), d(w_{1}), 4) \leq   g(6, 4) =  -0.0399417, \nonumber \\
&&  g(d(z_{1}), d(w_{1}), 5) \leq   g(7, 5) =  -0.0442738,  \quad \text{and} \nonumber \\
&&  g(d(z_{1}), d(w_{1}), 6) \leq   g(7, 6) =  -0.0470986. \nonumber
 \eeq
From here follows that also the change of the ABC index  (\ref{thm-noB2-200}) when  $d(w_{k-1})=5$ is negative.
Analogous proofs, one can obtain for $d(w_{k-1})=6,7$, so we omit them.

After applying transformations $\mathcal{T}_{3}$ and $\mathcal{T}_{4}$ $d(w_{k-1})=4$, 
we have obtained a tree with smaller ABC index than $G$  and with $k-1$
different parent vertices of the $B_2$-branches.

We can repeatedly apply the transformations from the above three cases until we end with
a tree whose all $B_2$-branches have at most two different parent vertices,  $w_1$ and $w_2$.
Observe that vertex $w_2$ and its children were not affected by the eventual prior modification
from Cases~$1$, $2$, and $3$. Also notice that $w_1$ after the above transformations may gain
only new $B_2$ and $B_3$-branches.
Next, if the tree has more than $11$ $B_2$-branches, we proceed with further transformation obtaining a tree with smaller ABC index, with
all of $B_2$-branches attached to only one vertex.
Here we distinguish two main cases.

\bigskip
\noindent
{\bf Case $A$.} $w_2$ does not have children of degree $2$.

\noindent
In this case we can apply one of the transformations from Cases~$1$, $2$, and $3$,
after which only $w_1$ will have $B_2$-branches as children. 
If there are more than  $11$ $B_2$-branches, we can apply the transformation from
Lemma~\ref{lemma-B2-30}, and obtain a tree with smaller ABC index and maximal $11$ $B_2$-branches.

\bigskip
\noindent
{\bf Case $B$.} $w_2$ has a children of degree $2$.

\noindent
Let $n_{22}$ and $n_{21}$ be the number of $B_2$ and $B_1$-branches, respectively, 
that are attached to $w_2$, and let $n_{13}$ and $n_{12}$ be the number of $B_3$ and $B_2$-branches, respectively, 
that are attached to $w_1$. 
%If $w_1$ is not the root of a tree, then we can apply transformations from Lemma~,
%and obtain a that each of $w_1$ and $w_2$ has at most $6$ $B_2$-branches as children.
%If $w_1$ is a root of a tree and the parent of $w_2$, then we can apply the transformation from
%Lemma~, and obtain a tree with smaller ABC index and maximal $11(12)$ $B_2$-branches,
%which concludes the proof of the theorem.

\bigskip
\noindent
{\bf Subcase $B.1$.} $w_1$ is parent of $w_2$.

\noindent
We consider further two subcases regarding if $w_1$ is the root vertex or not.

\smallskip
\noindent
{\bf Subcase $B.1.1$.} $w_1$ is not a root vertex.

\noindent
In this case, by Lemma~\ref{lemma-B2-20}, $w_2$ cannot be a parent of more than $6$ $B_2$-branches.
If $w_1$ is a parent of more than $6$ $B_2$-branches, then we can apply the transformations from
Lemma~\ref{lemma-B2-20}, and obtain that also $w_1$ is not a parent of more than $6$ $B_2$-branches.

If $n_{22} \leq 6$ and $n_{12} \leq 5$, or $n_{22} \leq 5$ and $n_{12} \leq 6$, the theorem holds.
If $n_{22} = 6$ and $n_{12} = 6$, we apply the transformation $\mathcal{T}_{5}$ illustrated in  Figure~\ref{fig-te-numberB2-50}.
\begin{figure}[h]
\begin{center}
%\vspace{-0.3cm}
\includegraphics[scale=0.75]{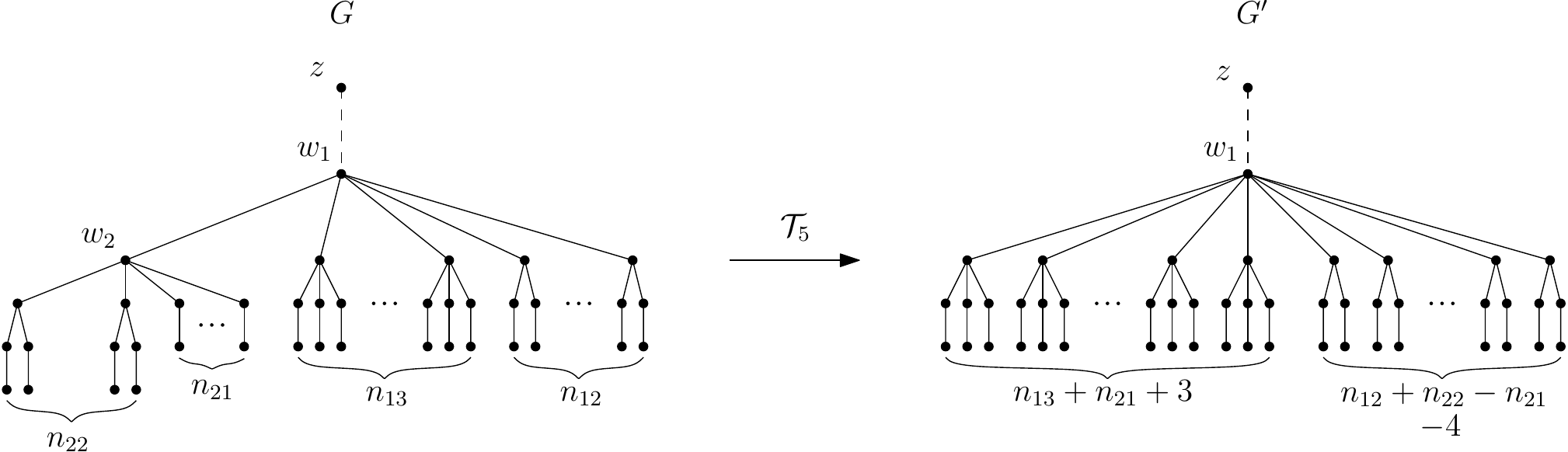}% [height=3.5cm,width=15cm]
\caption{An illustration of the transformation $\mathcal{T}_{5}$ from the proof of Theorem~\ref{te-no2branches-10a}, Subcase~$B.1.1$.}
%\label{Unicyclic-Max-M3}
\label{fig-te-numberB2-50}
%\vspace{-0.3cm}
\end{center}
\end{figure}
By Proposition~\ref{pro-B1-10-00}, it follows that $n_{21}=0$ or  $n_{21}=1$. Since the case $n_{21}=0$ 
is resolved in Case~$A$, we consider here that $n_{21}=1$.
In this case, after applying $\mathcal{T}_{5}$ 
the degree of the vertex $w_1$ increases by $4$, 
the degree of the vertex $w_2$ decreases to $1$,
one child vertex of $w_2$ decreases its degree from $3$ to $2$, and
four children vertices of $w_1$ increase there degrees from $3$ to $4$.
The rest of the vertices do not change their degrees.
The change of the ABC index after applying $\mathcal{T}_5$ is
\beq \label{thm-noB2-200ah}
 && -f(d(z),d(w_1))+f(d(z), d(w_1)+4))-f(d(w_1),d(w_2))+f(2,1)-f(d(w_2),3)+f(4,2)  \nonumber \\
&&+(n_{22}-1)( -f(d(w_2),3)+f(d(w_1)+4,3)) +n_{13}( -f(d(w_1),4)+f(d(w_1)+4,4))  \nonumber \\
&&+4( -f(d(w_1),3)+f(d(w_1)+4,4))  +(n_{12}-4)( -f(d(w_1),3)+f(d(w_1)+4,3)). \nonumber \\
\eeq 
By Proposition~\ref{appendix-pro-030-2}  $-f(d(z),d(w_1))+f(d(z), d(w_1)+4))$ decrease in $d(z)$,
so it is maximal for $d(z)=d(w_1)$.
Together with $d(w_1)=n_{13}+n_{12}+2=n_{13}+8$, we have that 
\beq \label{thm-noB2-210a}
g(n_{13})&=&  -f(n_{13}+8,n_{13}+8)+f(n_{13}+8, n_{13}+12))  \nonumber \\
&& -f(n_{13}+8,8)+f(2,1)-f(8,3)+f(4,2)  \nonumber \\
&&+5( -f(8,3)+f(n_{13}+12,3)) +n_{13}( -f(n_{13}+8,4)+f(n_{13}+12,4))  \nonumber \\
&&+4( -f(n_{13}+8,3)+f(n_{13}+12,4))  +2( -f(n_{13}+8,3)+f(n_{13}+12,3)), \nonumber
\eeq 
is an upper bound on (\ref{thm-noB2-200ah}).
The function $g(n_{13})$ is maximal when $n_{13} \to \infty$, and $\lim_{n_{13} \to \infty} g(n_{13})$ $=$ $-0.0362243$.

\smallskip
\noindent
{\bf Subcase $B.1.2$.} $w_1$ is a root vertex.

\noindent
Due to the fact that $w_1$ and $w_2$ are the only vertices that are parent to
$B_2$-branches and by Theorems~\ref{thm-DS},~\ref{te-no5branches-10} and Lemma~\ref{lemma-15}, it follows
that $w_1$ beside $B_2$-branches, may have only $B_3$-branches as children. 
% (THIS is NOT true, one should consider what kind of children are possible after applying the above trnsforamtions).
%Here, if  $n_{13} > 3$, we apply the transformation $\mathcal{T}_{5}$, otherwise we apply the transformation
%$\mathcal{T}_{6}$, both illustrated in  Figure~\ref{fig-te-numberB2-50}.
Notice that if $n_{12}+n_{22} \leq 11$, the theorem in this case holds, so we assume that $n_{12}+n_{22} \geq 12$,
and we can apply the transformation $\mathcal{T}_5$ from  Figure~\ref{fig-te-numberB2-50}.
Since $n_{22} \leq 6$, we consider further the cases when $12 -n_{22}  \leq n_{12} \leq 11$.

\smallskip
\noindent
{\bf Subcase $B.1.2.1$.} $n_{12} =6$.

\noindent
Here we apply the same transformation as in Subcase $B.1.1$.
The only difference here is that there is $w_1$ does not have a parent vertex,
and therefore the expression $-f(d(z),d(w_1))+f(d(z), d(w_1)+4))$  is not
included in (\ref{thm-noB2-200ah}), and  $d(w_1)=n_{13}+n_{12}+1=n_{13}+7$. Thus, the change of the ABC index in this case
is bounded from above by
\beq \label{thm-noB2-210b}
g(n_{13})&=&  -f(n_{13}+7,8)+f(2,1)-f(8,3)+f(4,2)  \nonumber \\
&&+5( -f(8,3)+f(n_{13}+11,3)) +n_{13}( -f(n_{13}+7,4)+f(n_{13}+11,4))  \nonumber \\
&&+4( -f(n_{13}+7,3)+f(n_{13}+11,4))  +2( -f(n_{13}+7,3)+f(n_{13}+11,3)). \nonumber
\eeq 
The function $g(n_{13})$ is negative for $n_{13} \geq 0$, it is maximal when $n_{13} \to \infty$, and $\lim_{n_{13} \to \infty} g(n_{13})$ $=$ $-0.0362242$.
%In the case $n_{13} =0$, $G$ is a tree with $64$ vertices.
%The tree that has smaller ABC-index than $G$ is a tree depicted in Figure~\ref{fig-te-numberB2-50c}
%(actually in~\cite{d-ectmabci-2013}  it was computed that the tree from Figure~\ref{fig-te-numberB2-50c} 
%has the minimal ABC index among all trees with $64$ vertices).
%
%\begin{figure}[h]
%\begin{center}
%\includegraphics[scale=0.75]{Figures/fig-te-numberB2-50c.pdf}% [height=3.5cm,width=15cm]
%\caption{The minimal-ABC tree with $64$ vertices.}
%\label{fig-te-numberB2-50c}
%\end{center}
%\end{figure}
%

\smallskip
\noindent
{\bf Subcase $B.1.2.2$.} $7  \leq n_{12} \leq 11$.

\noindent
After applying $\mathcal{T}_5$, 
the degree of the vertex $w_1$ increases by $n_{22}-2$, 
the degree of the vertex $w_2$ decreases to $1$,
one child vertex of $w_2$ decreases its degree from $3$ to $2$, and
four children vertices of $w_1$ increase there degrees from $3$ to $4$.
The change of the ABC index is bounded from above by
\beq \label{thm-noB2-200a}
 && -f(d(w_1),d(w_2))+f(2,1)  -f(d(w_2),3)+f(4,2)  \nonumber \\
&& +(n_{22}-1)( -f(d(w_2),3)+f(d(w_1)+n_{22}-2,3))   \nonumber \\
&&+n_{13}( -f(d(w_1),4)+f(d(w_1)+n_{22}-2,4))   \nonumber \\
&& +n_{21}(-f(d(w_1),3)+f(d(w_1)+n_{22}-2,4))  \nonumber \\
&&+ 3(-f(d(w_1),3)+f(d(w_1)+n_{22}-2,4))  \nonumber \\
&&+(n_{12}+n_{22}-n_{21}-4)( -f(d(w_1),3)+f(d(w_1)+n_{22}-2,3)).
\eeq 
Considering that $d(w_1)=n_{13}+n_{12}+1$ and $d(w_2)=n_{22}+n_{21}+1$, we can write
(\ref{thm-noB2-200a}) as
\beq \label{thm-noB2-200ab}
 g(n_{13},n_{12}, n_{22}, n_{21}) && -f(d(w_1),d(w_2))+f(2,1)  -f(d(w_2),3)+f(4,2)  \nonumber \\
&& +(n_{22}-1)( -f(d(w_2),3)+f(d(w_1)+n_{22}-2,3))   \nonumber \\
&&+n_{13}( -f(d(w_1),4)+f(d(w_1)+n_{22}-2,4))   \nonumber \\
&& +n_{21}(-f(d(w_1),3)+f(d(w_1)+n_{22}-2,4))  \nonumber \\
&&+ 3(-f(d(w_1),3)+f(d(w_1)+n_{22}-2,4))  \nonumber \\
&&+(n_{12}+n_{22}-n_{21}-4)( -f(d(w_1),3)+f(d(w_1)+n_{22}-2,3)). \nonumber
\eeq 
By Prepositions~\ref{pro-B1-10-00} and~\ref{pro-B2-30} we have additional constrains on $n_{13},n_{12}, n_{22}$, and $n_{21}$.
Namely, it holds that $7 \leq n_{12} \leq 9$, $0 \leq n_{13} \leq 4$, $1 \leq n_{22} \leq 6$, and $1 \leq n_{21} \leq 4$.
Also, it holds that  $n_{12} + n_{13} \leq 11$ and $n_{22} + n_{21} \leq 7$.
For all feasible values of the parameters $n_{13},n_{12}, n_{22}$, and $n_{21}$, 
the function $g(n_{13},n_{12}, n_{22}, n_{21})$ obtain the maximal value of $-0.0640574=g(0,9,3,1)$.

Observe that with the above possible values of the parameters $n_{22}, n_{21}, n_{13}$ and $n_{13}$
the graph $G$ has strictly less than $134$ vertices.
In~\cite{d-ectmabci-2013} all minimal-ABC trees with up to $300$ were computed,
and no of them has the structure of $G$.

\bigskip
\noindent
{\bf Subcase $B.2$.} $w_1$ is not parent of $w_2$.

\noindent
%We distinguish four case regarding the $n_{22}$, $n_{22}=1$, $n_{22}=2$, $n_{22}=3$ and $6 \geq n_{21} \geq 3$, 
%and for each case respectively we apply the transformations $\mathcal{T}_{7}$, $\mathcal{T}_{7}$, 
%$\mathcal{T}_{7}$ and $\mathcal{T}_{7}$, respectively,  illustrated in Figure~.
%
If $n_{12}> 6$ or $n_{22}> 6$ we can apply transformations from Lemma~\ref{lemma-B2-20}, so that afterwards we obtain $n_{12} \leq 6$ and $n_{22} \leq 6$.
%Thus, if  $n_1+n_2 > 10$, it follows that $n_1 \geq 5$ and $n_2 \geq 5$.
The theorem holds, if $n_{22} \leq 6$ and $n_{12} \leq 5$, or $n_{22} \leq 5$ and $n_{12} \leq 6$, the theorem holds.
%If $n_{22} = 6$ and $n_{12} = 6$, we apply the transformation $\mathcal{T}_{5}$ illustrated in  Figure~\ref{fig-te-numberB2-50}.
By Proposition~\ref{pro-B1-10-00}, if $n_{22}=6$, $w_2$ may have at most one $B_1$-branch as a child. If $w_2$ does not have a $B_1$-branch as a child,
we apply the transformations from Case~$1$ of this proof, and the proof is completed.
If $w_2$ has one $B_1$-branch as a child, and $n_{22} = 6$ and $n_{12} = 6$, 
then we proceed with the transformation $\mathcal{T}_{6}$ illustrated in  Figure~\ref{fig-te-numberB2-60}.
\begin{figure}[h]
\begin{center}
%\vspace{-0.3cm}
\includegraphics[scale=0.75]{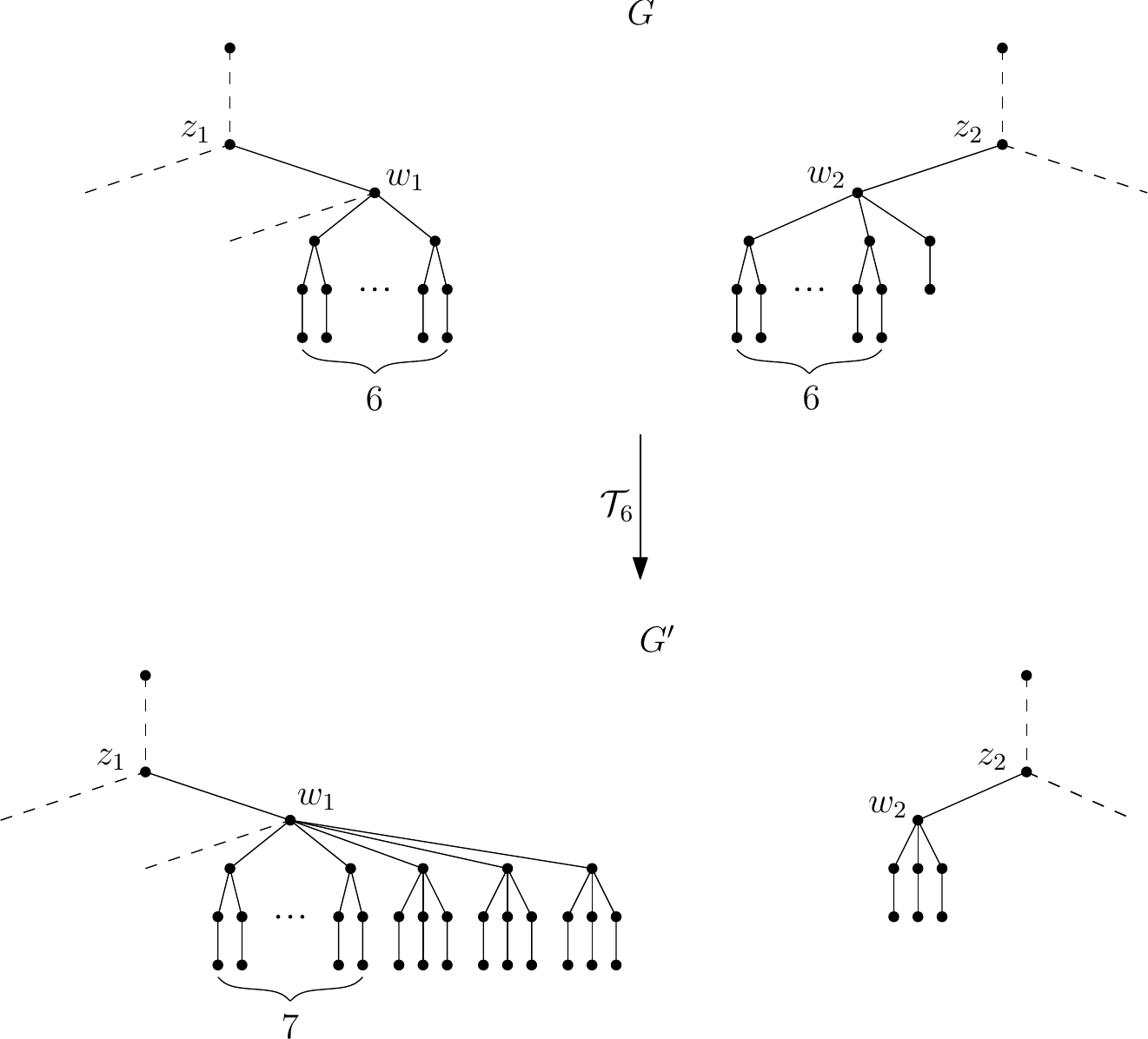}% [height=3.5cm,width=15cm]
\caption{An illustration of the transformations $\mathcal{T}_6$  from the proof of Theorem~\ref{te-no2branches-10a}, Subcase $B.2$.}
%\label{Unicyclic-Max-M3}
\label{fig-te-numberB2-60}
%\vspace{-0.3cm}
\end{center}
\end{figure}
After applying $\mathcal{T}_6$, 
the degree of the vertex $w_1$ increases by $4$, 
the degree of the vertex $w_2$ decreases from $8$ to $4$,
two children vertices of $w_2$ decrease their degrees from $3$ to $2$ and $1$, respectively.
The rest of the vertices do not change their degrees.
Thus, the change of the ABC index is smaller than
\beq \label{thm-noB2-300a}
&& -f(8,d(z_2))+f(4,d(z_2)) -f(8,3)+f(2,1) -f(8,3)+f(4,2)  \nonumber \\
&&+ 3(-f(8,3)+f(d(w_1)+4,4)) -f(d(w_1),3)+f(d(w_1)+4,4))  \nonumber \\
&&+ 6(-f(d(w_1),3)+f(d(w_1)+4,3)) .
\eeq 
By Proposition~\ref{appendix-pro-030-2}   $-f(8,d(z_2))+f(4,d(z_2))$ is maximal for $d(z_2) \to \infty$.
Thus,
\beq \label{thm-noB2-400a}
g(d(w_1))&=&\lim_{d(z_2) \to \infty} (-f(8,d(z_2))+f(4,d(z_2))) -f(8,3)+f(2,1) -f(8,3)+f(4,2)  \nonumber \\
&&+ 3(-f(8,3)+f(d(w_1)+4,4)) -f(d(w_1),3)+f(d(w_1)+4,4))  \nonumber \\
&&+ 6(-f(d(w_1),3)+f(d(w_1)+4,3)), \nonumber
\eeq 
is an upper bound on (\ref{thm-noB2-300a}).
The function $g(d(w_1))$ obtains its maximum of $-0.0154895$ for $d(w_1)=11$.
Observe that $g(d(w_1))$ is independent on $d(z_1)$, so the proof holds for any value of $d(z_1)$,
including the case $z_1 = z_2$.
\end{proof}

\section[Conclusion]{Conclusion}\label{sec:Conclusion}

The main contributions of this work are the upper bounds on the 
number of $B_1$ and $B_2$-branches in the minimal-ABC trees
presented in Theorems~\ref{theorem-B1-10} and~\ref{te-no2branches-10a}, respectively.
The theorems state that a minimal-ABC tree may have at most four  $B_1$-branches and at most
eleven $B_2$-branches.
Morover, it was shown that these two bounds are attained in special cases that cannot occur simultaneously.
Based on the experimental results~\cite{d-ectmabci-2013} and 
obtained (counter)examples of minmal-ABC trees \cite{gfahsz-abcic-2013, d-ectmabci-2013, adgh-dctmabci-14},
it is very likely that here presented upper bounds on $B_1$ and $B_2$-branches 
%in the minimal-ABC trees
are not sharp, although they are quite close to the conjectured sharp bounds below.

\begin{conjecture}
 A minimal-ABC tree can contain at most three $B_1$-branches.
\end{conjecture}

\begin{conjecture}
 A minimal-ABC tree can contain at most nine $B_2$-branches.
\end{conjecture}

However, Theorem~\ref{theorem-B1-10} also states that  if a minmal-ABC tree is a $T_k$-branch itself, then 
it can contain at most three $B_1$-branches. This is the best possible bound since the minimal-ABC trees with  $14$ and $19$ vertices 
contain three $B_1$-branches \cite{gfi-ntmabci-12, d-ectmabci-2013}.

The results presented here, together with Theorems~\ref{te-no5branches-10} and~\ref{thm-20}
from \cite{d-sptmabci-2014}, show that beside a very small number of $B_1$, $B_2$ and $B_4$-branches,
the minimal-ABC trees are comprised of $B_3$-branches and additional number of internal vertices.
This goes in line with Conjecture~$3.2.$  \cite{adgh-dctmabci-14}, which states that  
enough large minimal-ABC trees are comprised exclusively of  $B_3$-branches and internal vertices.

%As it was already indicated in \cite{adgh-dctmabci-14}, we conjecture that enough large minimal-ABC trees are comprised 
%exclusively of  $B_3$-branches and internal vertices.

%
% -------------------------------------------------------------------------
% ----------------------           bibliography       ---------------------
% -------------------------------------------------------------------------
%

%


\begin{thebibliography}{999}
\setlength{\itemsep}{0pt}


\bibitem{adgh-dctmabci-14}
M.~B.~Ahmadi, D.~Dimitrov, I.~Gutman, S.~A.~Hosseini,
\textit{Disproving a conjecture on trees with minimal atom-bond connectivity index}, 
MATCH Commun. Math. Comput. Chem. \textbf{72} (2014) 685--698.


\bibitem{ahs-tmabci-13}
M.~B.~Ahmadi, S.~A.~Hosseini, P.~Salehi~Nowbandegani,
\textit{On trees with minimal atom bond connectivity index}, 
MATCH Commun. Math. Comput. Chem. \textbf{69} (2013) 559--563.


\bibitem{ahz-ltmabci-13}
M.~B.~Ahmadi, S.~A.~Hosseini, M.~Zarrinderakht,
\textit{On large trees with minimal atom--bond connectivity index}, 
MATCH Commun. Math. Comput. Chem. \textbf{69} (2013) 565--569.

\bibitem{as-abciic-10}
M.~B.~Ahmadi, M.~Sadeghimehr,
\textit{Atom bond connectivity index of an infinite class ${\rm NS_1}[n]$ of dendrimer nanostars}, 
Optoelectron.  Adv. Mat. \textbf{4}  (2010) 1040--1042.


%\bibitem{b-hddbti-82}
%A.~T.~Balaban,
%\textit{Highly discriminating distance-based topological index}, 
%Chem. Phys. Lett. \textbf{89} (1982) 399--404.


%\bibitem{b-ithiccs-83} 
%D.~Bonchev, \textit{Information Theoretic Indices for Characterization of
%Chemical Structures}, Research Studies Press, Chichester, 1983.


\bibitem{cg-eabcig-11}
J.~Chen, X.~Guo,
\textit{Extreme atom-bond connectivity index of graphs}, 
MATCH Commun. Math. Comput. Chem. \textbf{65} (2011) 713--722.

\bibitem{cg-abccbg-12}
J.~Chen, X.~Guo,
\textit{The atom-bond connectivity index of chemical bicyclic graphs}, 
Appl. Math. J. Chinese Univ. \textbf{27} (2012) 243--252.


\bibitem{clg-subabcig-12}
J.~Chen, J.~Liu, X.~Guo,
\textit{Some upper bounds for the atom-bond connectivity index of graphs}, 
Appl. Math. Lett. \textbf{25} (2012) 1077--1081.

\bibitem{cll-abcbsp-13}
J.~Chen, J.~Liu, Q.~Li,
\textit{The atom-bond connectivity index of catacondensed polyomino graphs}, 
Discrete Dyn. Nat. Soc. \textbf{2013}  (2013) ID 598517.


\bibitem{d-abcig-10}
K.~C.~Das,
\textit{Atom-bond connectivity index of graphs}, 
Discrete Appl. Math. \textbf{158} (2010) 1181--1188.

\bibitem{dgf-abci-11}
K.~C.~Das, I.~Gutman, B.~Furtula,
\textit{On atom-bond connectivity index}, 
Chem. Phys. Lett. \textbf{511} (2011) 452--454.

\bibitem{dgf-abci-12}
K.~C.~Das, I.~Gutman, B.~Furtula,
\textit{On atom-bond connectivity index}, 
Filomat \textbf{26} (2012) 733--738.


\bibitem{dt-cbfgaiabci-10}
K.~C.~Das, N.~Trinajsti{\' c},
\textit{Comparison between first geometric-arithmetic index and atom-bond connectivity index}, 
Chem. Phys. Lett. \textbf{497} (2010) 149--151.


%\bibitem{dgv-iihdpg-12}
%M.~Dehmer, M.~Grabner, K.~Varmuza,
%\textit{Information indices with high discriminative power for graphs}, 
%PLoS ONE \textbf{7} (2012) e31214.


%\bibitem{dk-epge-12}
%M.~Dehmer, V.~Kraus,
%\textit{On extremal properties of graph entropies}, 
%MATCH Commun. Math. Comput. Chem. \textbf{68} (2012) 889--912.


%\bibitem{da-hgem-11}
%M.~Dehmer, A.~Mowshowitz,
%\textit{A history of graph entropy measures}, 
%Inf. Sci. \textbf{181} (2011) 57--78.


%\bibitem{db-tird-99}
%J.~Devillers, A~T.~Balaban (Eds.),  {\em Topological indices and related descriptors in
%QSAR and QSPR}, Wiley--VCH, Gordon and Breach, Amsterdam, 1999.


\bibitem{d-ectmabci-2013} 
D.~Dimitrov, \textit{Efficient computation of trees with minimal atom-bond connectivity index}, 
Appl. Math. Comput. \textbf{224} (2013) 663--670.

\bibitem{d-sptmabci-2014} 
D.~Dimitrov, \textit{On structural properties of trees with minimal atom-bond connectivity index}, 
Discrete Appl. Math. \textbf{172} (2014) 28--44.

%\bibitem{e-c3dms-00}
%E.~Estrada, \textit{Characterization of 3D molecular structure}, 
%Chem. Phys. Lett. \textbf{319} (2000)  713--718.


\bibitem{e-abceba-08}
E.~Estrada, \textit{Atom-bond connectivity and the energetic of branched alkanes}, 
Chem. Phys. Lett. \textbf{463} (2008) 422--425.

\bibitem{etrg-abc-98}
E.~Estrada, L.~Torres, L.~Rodr{\' i}guez, I.~Gutman, \textit{An atom-bond connectivity index: {M}odelling the enthalpy of formation of alkanes},
Indian J. Chem. \textbf{37}A (1998) 849--855.

\bibitem{ftvzag-siabcigo-2011}
G.~H.~Fath-Tabar, B~Vaez-Zadeh, A.~R.~Ashrafi, A.~Graovac,
{\em Some inequalities for the atom-bond connectivity index of graph
operations,} 
Discrete Appl. Math.  {\bf  159} (2011) 1323--1330.


\bibitem{fgv-abcit-09}
B.~Furtula, A.~Graovac, D.~Vuki{\v c}evi{\' c},
\textit{Atom-bond connectivity index of trees}, 
Discrete Appl. Math. \textbf{157} (2009) 2828--2835.

%\bibitem{fgd-ssdbti-13}
%B.~Furtula, I.~Gutman, M.~Dehmer,
%\textit{On structure-sensitivity of degree-based topological indices}, 
%Appl. Math. Comput. \textbf{219} (2013) 8973--8978.

\bibitem{fgiv-cstmabci-12}
B.~Furtula, I.~Gutman, M.~Ivanovi{\' c}, D.~Vuki{\v c}evi{\' c},
\textit{Computer search for trees with minimal ABC index}, 
Appl. Math. Comput. \textbf{219} (2012) 767--772.



\bibitem{ghl-srabcig-11}
L.~Gan, H.~Hou, B.~Liu,
\textit{Some results on atom-bond connectivity index of graphs}, 
MATCH Commun. Math. Comput. Chem. \textbf{66} (2011) 669--680.

\bibitem{gly-abctgds-12}
L.~Gan, B.~Liu, Z.~You,
\textit{The ABC index of trees with given degree sequence}, 
MATCH Commun. Math. Comput. Chem. \textbf{68} (2012) 137--145.

\bibitem{gg-nwabci-10}
A.~Graovac, M.~Ghorbani,
\textit{A new version of the atom-bond connectivity index}, 
Acta Chim. Slov. \textbf{57} (2010) 609--612.


%\bibitem{g-eg-78}
%I.~Gutman,
%\textit{The energy of a graph}, 
%Ber. Math.--Statist. Sekt. Forschungsz. Graz \textbf{103} (1978) 1--22.

%\bibitem{g-dbti-13}
%I.~Gutman,
%\textit{Degree-based topological indices}, 
%Croat. Chem. Acta \textbf{86} (2013) 351--361.

\bibitem{gf-tsabci-12}
I.~Gutman, B.~Furtula,
\textit{Trees with smallest atom-bond connectivity index}, 
MATCH Commun. Math. Comput. Chem. \textbf{68} (2012) 131--136.


\bibitem{gfahsz-abcic-2013}
I.~Gutman, B.~Furtula, M.~B.~Ahmadi, S.~A.~Hosseini, P.~Salehi Nowbandegani, M.~Zarrinderakht,
{\em The ABC index conundrum,} 
Filomat  {\bf  27} (2013) 1075--1083.


\bibitem{gfi-ntmabci-12}
I.~Gutman, B.~Furtula, M.~Ivanovi{\' c},
\textit{Notes on trees with minimal atom-bond connectivity index}, 
MATCH Commun. Math. Comput. Chem. \textbf{67} (2012) 467--482.


\bibitem{gtrm-abcica-12}
I.~Gutman, J.~To{\v s}ovi{\' c}, S.~Radenkovi{\' c}, S.~Markovi{\' c},
\textit{On atom-bond connectivity index and its chemical applicability}, 
Indian J. Chem. \textbf{51}A (2012) 690--694.

%\bibitem{GT}
%I.~Gutman, N.~Trinajsti\'c, \textit{ Graph Theory and Molecular
%Orbitals. Total $\pi-$electron Energy of Alternant Hydrocarbons},
%Chem. Phys. Lett. \textbf{17} (1971) 535--538.

\bibitem{gzx-rabchi-2014}
I.~Gutman, L.~Zhong, K.~Xu, \textit{ Relating the ABC and harmonic indices},
J. Serb. Chem. Soc. \textbf{79} (2014) 557--563.


%\bibitem{h-ti-1971}
%H.~Hosoya, 
%\textit{Topological index. A newly proposed quantity characterizing the topological nature of structural isomers of saturated hydrocarbons}, 
%Bull. Chem. Soc. Japan \textbf{44} (1971) 2332--2339.


\bibitem{hag-ktmabci-14}
S.~A.~Hosseini, M.~B.~Ahmadi, I.~Gutman,
\textit{Kragujevac trees with minimal atom-bond connectivity index}, 
MATCH Commun. Math. Comput. Chem. \textbf{71} (2014) 5--20.



\bibitem{k-abcibsfc-12}
X.~Ke,
\textit{Atom-bond connectivity index of benzenoid systems and fluoranthene congeners}, 
Polycycl. Aromat. Comp. \textbf{32}  (2012) 27--35.


\bibitem{lccglc-fcstmaibtds-14}
W.~Lin, J.~Chen, Q.~Chen, T.~Gao, X.~Lin, B.~Cai,
{\em Fast computer search for trees with minimal ABC index based on tree degree sequences,} 
MATCH Commun. Math. Comput. Chem. \textbf{72} (2014) 699--708.


\bibitem{llgw-pcgctmabci-13}
W.~Lin, X.~Lin, T.~Gao, X.~Wu,
{\em Proving a conjecture of Gutman concerning trees with minimal ABC index,} 
MATCH Commun. Math. Comput. Chem. \textbf{69} (2013) 549--557.


\bibitem{p-rubabci-14}
J.~L.~Palacios,
{\em A resistive upper bound for the ABC index,} 
MATCH Commun. Math. Comput. Chem. \textbf{72} (2014) 709--713.

%\bibitem{r-cmb-1975}
%M.~Randi{\' c},
%{\em On characterization of molecular branching,} 
%J. Am. Chem. Soc. \textbf{97} (1975), 6609-6615.


%\bibitem{tc-ndc-09}
%R.~Todeschini, V.~Consonni,  {\em Molecular descriptors for chemoinformatics}, Wiley--VCH, Weinheim, 2009.


\bibitem{vh-mabcict-2012}
T.~S.~Vassilev, L.~J.~Huntington,
\textit{On the minimum ABC index of chemical trees}, 
Appl. Math. \textbf{2} (2012) 8--16.
 

\bibitem{w-etwgdsri-2008}
H.~Wang,
\textit{Extremal trees with given degree sequence for the Randi{\' c} index}, 
Discrete Math. \textbf{308} (2008) 3407--3411.


%\bibitem{w-rppiams-1948}
%H.~Wiener,
%\textit{Relation of the physical properties of the isomeric alkanes to molecular structure. Surface tension, specific dispersion, and critical solution temperature in aniline}, 
%J. Phys. Chem. \textbf{52} (1948) 1082--1089.

\bibitem{xz-etfdsabci-2012}
R.~Xing, B.~Zhou,
{\em Extremal trees with fixed degree sequence for atom-bond connectivity index,} 
Filomat  {\bf  26} (2012) 683--688.

\bibitem{xzd-abcicg-2011}
R.~Xing, B.~Zhou, F.~Dong,
{\em On atom-bond connectivity index of connected graphs,} 
Discrete Appl. Math.  {\bf  159} (2011) 1617--1630.

\bibitem{xzd-frabcit-2010}
R.~Xing, B.~Zhou, Z.~Du,
{\em Further results on atom-bond connectivity index of trees,} 
Discrete Appl. Math.  {\bf  158} (2011) 1536--1545.

\bibitem{yxc-abcbsp-11}
J.~Yang, F.~Xia, H.~Cheng,
\textit{The atom-bond connectivity index of benzenoid systems and phenylenes}, 
Int. Math. Forum \textbf{6} (2011) 2001--2005.

%
\end{thebibliography}
\end{document}